\documentclass[11pt, draftcls, onecolumn]{IEEEtran}

\usepackage{cite,bbm,graphicx,amsmath,amssymb,mathrsfs,epsf} \usepackage{multirow}
\usepackage{subfigure,cite,graphicx,epsfig,amsmath,amssymb,mathrsfs,epsf,graphics,color,enumerate}

\newtheorem{theorem}{Theorem}
\newtheorem{lemma}{Lemma}
\newtheorem{example}{Example}
\newtheorem{corollary}{Corollary}
\newtheorem{proposition}{Proposition}
\newtheorem{remark}{Remark}

\def \arg{\operatorname{arg}}
\def \E{\operatorname{E}}

\def \var{\operatorname{Var}}

\newcommand{\uin}{}
\newcommand{\lin}{\bar}

\newcommand{\uset}{}
\newcommand{\lset}{\bar}

\begin{document}
\title{A Unified Approach for Network Information Theory}
\author{Si-Hyeon~Lee,~\IEEEmembership{Member,~IEEE,} and~Sae-Young~Chung,~\IEEEmembership{Senior Member,~IEEE}
\thanks{S.-H. Lee is with the Department of Electrical and Computer Engineering,
University of Toronto, Toronto, Canada (e-mail: sihyeon.lee@utoronto.ca). This work was done when she was at KAIST. S.-Y. Chung is with the Department of Electrical Engineering, KAIST, Daejeon, South Korea (e-mail: sychung@ee.kaist.ac.kr). The material in this paper will be presented in part at IEEE ISIT 2015 \cite{LeeChung:15ISIT1}.}
}

\maketitle
\begin{abstract}
In this paper, we take a unified approach for network information theory and prove a coding theorem, which can recover most of the achievability results in network information theory that are based on random coding. The final single-letter expression has a very simple form, which was made possible by many novel elements such as a unified framework that represents various network problems in a simple and unified way, a unified coding strategy that consists of a few basic ingredients but can emulate many known coding techniques if needed, and new proof techniques beyond the use of standard covering and packing lemmas. For example, in our framework, sources, channels, states and side information are treated in a unified way and various constraints such as cost and distortion constraints are unified as a single joint-typicality constraint.

Our theorem can be useful in proving many new achievability results easily and in some cases gives simpler rate expressions than those obtained using conventional approaches. Furthermore, our unified coding can strictly outperform existing schemes. For example, we obtain a generalized decode-compress-amplify-and-forward bound as a simple corollary of our main theorem and show it strictly outperforms previously known coding schemes. Using our unified framework, we formally define and characterize three types of network duality based on channel input-output reversal and network flow reversal combined with packing-covering duality.

\end{abstract}

%

\IEEEpeerreviewmaketitle

\section{Introduction}

In network information theory, we study the fundamental limits of information flow and processing in a network and develop coding strategies that can approach the limits closely. Instead of studying a fully general network, however, we often study simple canonical models such as the multiple-access channel \cite{liao_thesis}, relay channel \cite{CoverElGamal:79}, and distributed source coding  \cite{SlepianWolf:73} because they are easier to study and more importantly because we can get useful insights from studying them. Once such insights are obtained, one can try to develop a more general theory that is applicable to general networks. 

However, such a task is challenging and only partial results have been known so far  \cite{Yeung:00,KramerGastparGupta:05,AvestimehrDiggaviTse:11,YassaeeAref:11, Lim:10, HouKramer:arxiv13, RiniGoldsmith:ITA13, MineroLimKim:15, LimKimKim:14, LimKimKim:14ITW, YassaeeArefGohari:14}, in which network model and/or applied coding technique is limited. 
For example, network coding~\cite{Yeung:00} and compress-and-forward (CF)~\cite{CoverElGamal:79}  were unified as noisy network coding in~\cite{YassaeeAref:11,Lim:10}, but does not include decode-and-forward (DF) \cite{CoverElGamal:79}. DF and partial DF \cite{CoverElGamal:79} were generalized for single-source multiple-relay single-destination networks \cite{KramerGastparGupta:05} and for multicast and broadcast networks \cite{LimKimKim:14,LimKimKim:14ITW}, respectively. In \cite{HouKramer:arxiv13}, noisy network coding was combined with network DF \cite{KramerGastparGupta:05}, but does not allow a relay to perform both partial DF and CF simultaneously.  
For joint source-channel coding problems, a hybrid analog/digital coding strategy \cite{MineroLimKim:15} was proposed that recovers and generalizes many previously known results. Such a hybrid coding scheme was applied to some relay networks and was shown to unify both amplify-and-forward (AF) \cite{ScheinGallager:00} and CF \cite{CoverElGamal:79}.  In \cite{YassaeeArefGohari:14}, a novel framework for proving achievability was proposed  based on output statistics of random binning and source--channel duality. One important feature of this framework is that the addition of secrecy is free, i.e.,  once an achievability result  is obtained for a network model using this framework, an achievability result with additional secrecy constraint is immediately obtained.  We note that \cite{MineroLimKim:15} and \cite{YassaeeArefGohari:14} took a bottom-up approach in a sense that achievability results are separately obtained for each of various network models.  

In this paper, we take a top-down approach and prove a unified achievability theorem for a general network scenario with arbitrarily many nodes. Our setup is general enough such that any combination of source coding, channel coding, joint source-channel coding, and coding for computing problems can be treated. 
Our result recovers most of the exiting achievability results in network information theory as long as they are based on random coding. Some examples of known results recovered by our theorem are listed as follows:
\begin{itemize}
\item Channel coding: Gelfand-Pinsker coding~\cite{GelfandPinsker:80}, Marton's inner bound for the broadcast channel~\cite{Marton:79}, Han-Kobayashi inner bound for the interference channel~\cite{HanKobayashi:81}, \cite{ChongMotaniGargElGamal:08},  coding for channels with action-dependent states \cite{Weissman:10}, interference decoding for a 3-user interference channel~\cite{BandemerElGamal:11}, \cite{BandemerElGamal:Allerton11}, Cover-Leung inner bound for the multiple access channel with feedback~\cite{CoverLeung:81}, a combination of partial DF and CF for the relay channel~\cite{CoverElGamal:79}, network DF~\cite{KramerGastparGupta:05},
noisy network coding \cite{YassaeeAref:11,Lim:10}, short message noisy network coding with a DF option~\cite{HouKramer:arxiv13},  offset encoding for the multiple access relay channel~\cite{SankarKramerMandayam:07}.
 
\item Source coding: Slepian-Wolf coding \cite{SlepianWolf:73}, Wyner-Ziv coding~\cite{WynerZiv:76},  Berger-Tung inner bound for distributed lossy compression~\cite{Berger:77}, \cite{tung_thesis},  Zhang-Berger inner bound for multiple description coding \cite{ZhangBerger:78}.

\item Joint source-channel coding:  hybrid coding~\cite{MineroLimKim:15} and all previous results recovered by hybrid coding including sending arbitrarily correlated sources over multiple access channels \cite{CoverElGamalSalehi:80}, broadcast channels \cite{HanCosta:87}, and interference channels \cite{LiuChen:11,LiuChen:10}. 

\item Coding for computing: coding for computing~\cite{Yamamoto:82}, cascade coding for computing \cite{CuffSuElGamal:09}.
\end{itemize}

In Table~\ref{table:comparison}, we compare many approaches that attemped to unify various coding strategies.\footnote{The check mark `\checkmark' means that the corresponding unification approach subsumes both the model and the achievability bound of previous result. 
} 

\begin{table}[h]
\caption{Comparison of approaches that attempted unification of various network models and coding strategies}
\begin{center}
   \begin{tabular}{|l || c | c|  c | c | c |}
     \hline
Previous Results &  SWC \cite{YassaeeAref:11} & DDF \cite{LimKimKim:14,LimKimKim:14ITW}  & NNC-DF \cite{HouKramer:arxiv13} & HC \cite{MineroLimKim:15} & Our result \\  \hline \hline
 Gelfand-Pinsker coding \cite{GelfandPinsker:80} &  &   &          & $\checkmark$ & $\checkmark$ \\   \hline
 Marton coding\cite{Marton:79} &  &   \checkmark   & &    \checkmark        & $\checkmark$ \\  \hline
 Han-Kobayashi coding \cite{HanKobayashi:81} &  &      &       &         \checkmark   & $\checkmark$ \\   \hline
 Interference decoding \cite{BandemerElGamal:11} &  & &             &            & $\checkmark$ \\   \hline
  Cover-Leung coding \cite{CoverLeung:81} &&&&\checkmark& \checkmark \\ \hline
       DF  \cite{CoverElGamal:79}&  &  \checkmark &    \checkmark        &            & $\checkmark$ \\  \hline
     Partial DF \cite{CoverElGamal:79} &  &  \checkmark &       &            & $\checkmark$ \\  \hline
     AF \cite{ScheinGallager:00}&  &      &       & $\checkmark$ & $\checkmark$ \\   \hline 
      CF \cite{CoverElGamal:79} & $\checkmark$ & & $\checkmark$ & $\checkmark$ & $\checkmark$ \\  \hline 
      Combination of partial DF and CF \cite{CoverElGamal:79}  &&&&& \checkmark\\ \hline 
       Network coding \cite{Yeung:00}  &  $\checkmark$ &\checkmark & \checkmark &    \checkmark        & $\checkmark$ \\   \hline
     NNC \cite{YassaeeAref:11,Lim:10} &  $\checkmark$ & & \checkmark &        & $\checkmark$ \\  \hline 
     Wyner-Ziv coding \cite{WynerZiv:76} &  &        &     & $\checkmark$ & $\checkmark$ \\   \hline
      Slepian-Wolf coding \cite{SlepianWolf:73} &  $\checkmark$ & & &         \checkmark   & $\checkmark$ \\   \hline
     Berger-Tung coding \cite{Berger:77}, \cite{tung_thesis}&   &            & &      \checkmark      & $\checkmark$ \\ \hline
         Zhang-Berger coding \cite{ZhangBerger:78} &   &             &&      \checkmark      & \checkmark \\  \hline 
     \scriptsize Joint source-channel coding over &&&&&\\ 
    \scriptsize   multiple access channels \cite{CoverElGamalSalehi:80}, broadcast channels \cite{HanCosta:87},  &&&& \checkmark &\checkmark \\ 
     \scriptsize  and interference channels \cite{LiuChen:11,LiuChen:10}  &&&&&\\ \hline
    Hybrid coding\cite{MineroLimKim:15} &  &     &        & $\checkmark$ & $\checkmark$ \\   \hline    
      Coding for computing \cite{Yamamoto:82}  &   &  & &  & \checkmark \\ \hline 
      Cascade coding for computing \cite{CuffSuElGamal:09} & & & & & \checkmark \\ \hline
   \end{tabular}
\end{center} [Abbreviations] SWC: Slepian-Wolf coding over networks, DDF: distributed decode-and-forward, NNC-DF: noisy network coding with a DF option, HC: hybrid coding
\label{table:comparison}
\end{table}

Our theorem can be useful in proving new achievability results easily and in some cases gives simpler rate expressions than those obtained using conventional approaches. Furthermore, our unified coding can strictly outperform existing schemes. To illustrate this, we show that a generalized decode-compress-amplify-and-forward bound for acyclic networks can be obtained as a simple corollary of our main theorem and show it strictly outperforms previously known coding schemes. As another special case of our main theorem, we derive a generalized decode-compress-and-forward bound for a discrete memoryless network (DMN) in~\cite{LeeChung:15ISIT2}, which recovers both noisy network coding~\cite{Lim:10} and distributed decode-and-forward~\cite{LimKimKim:14} bounds. This is the first time the partial-decode-compress-and-forward bound (Theorem 7) by Cover and El Gamal~\cite{CoverElGamal:79} is generalized for DMN's such that each relay performs both partial DF and CF simultaneously. 

Our unified coding theorem enables us to state various types of duality arising in network information theory. Specifically, we formally define and characterize three types of network duality based on channel input-output reversal and network flow reversal combined with packing-covering duality.  Our duality results include as special cases many known duality relationships in network information theory, e.g., the duality between coding for multiple-access channel  \cite{liao_thesis} and distributed sources  \cite{Berger:77}, \cite{tung_thesis} (type-I duality), the duality between Gelfand-Pinsker coding \cite{GelfandPinsker:80} and Wyner-Ziv coding \cite{WynerZiv:76} (type-II duality), and the duality between coding for multiple-access channel \cite{liao_thesis}  and broadast channel  \cite{Marton:79}  (type-III duality).

Our unified achievability result is enabled by many novel elements such as a unified framework that represents various network problems in a simple and unified way, a unified coding strategy that consists of a few basic ingredients but can emulate known coding techniques if needed, and new proof techniques beyond the use of standard covering and packing lemmas. In our framework, sources, channels, states and side information are treated in a unified way and various constraints such as cost and distortion constraints are combined as a joint-typicality constraint, which is specified by a single joint distribution. Furthermore, we mainly consider acyclic discrete memoryless networks (ADMN) in this paper, where information flows in an acyclic manner. However, we also show our coding theorem can also be applied to general DMN's by unfolding the network. Graph unfolding was first used in~\cite{Yeung:00} for network coding.

Our coding scheme has four main ingredients, i.e., superposition coding, simultaneous nonunique decoding, simultaneous compression, and symbol-by-symbol mapping. We note that our coding scheme does not explicitly include binning and multicoding, but is still general enough to emulate them if needed. Although each of these coding ingredients is not new, these are tweaked and combined in a special way to enable unification of many previous approaches. In our coding scheme, covering codebooks are used to compress information that each node observes and decodes. These covering codebooks are generated to permit superposition coding \cite{Cover:72}. Each node operates according to the following three steps. The first step is  simultaneous nonunique decoding  \cite{ChongMotaniGargElGamal:08,NairElGamal:09,BandemerElGamalKim:12}, where a node uniquely decodes some covering codewords of other nodes  together with some other covering codewords that do not need to be decoded uniquely. The next step is simultaneous compression, where the node finds covering codewords simultaneously that carry information about a received channel output sequence and decoded codewords. Since we allow general superposition relationship among covering codebooks, a more general analysis beyond multivariate covering lemma \cite{ElGamlvanderMeulen:81,ElGamalKim:11} is needed. The last step is a symbol-by-symbol mapping from a received channel output sequence and decoded and covered codewords to a channel input sequence. The technique of using a symbol-by-symbol mapping was introduced in \cite{Shannon:58}, which is referred to as the Shannon strategy. Our symbol-by-symbol mapping from all three, i.e., the channel output sequence and decoded and covered codewords, was first used in \cite{WangNaghshvar:14} for a three-node noncausal relay channel. 
We note that such a use of symbol-by-symbol mapping results in correlation between a channel input sequence and nonchosen covering codewords and thus the standard packing lemma \cite{ElGamalKim:11} cannot be applied for the error analysis. Such correlation was problematic in many previous works and solved for some simple networks in  \cite{LapidothTinguely:10,MineroLimKim:15,WangNaghshvar:14}. Our proof technique completely solves this correlation issue in a fully general network setup.

This paper is organized as follows. In Section~\ref{sec:framework}, we present our unified framework. In  Section~\ref{sec:main}, we propose a unified coding scheme and present the main theorem of this paper. We also show various examples to illustrate how to utilize our results. In Section~\ref{sec:duality}, we characterize three types of network duality. To demonstrate usefulness of our unified coding theorem, in Section \ref{sec:GDCAF}, we derive a generalized decode-compress-amplify-and-forward bound as a simple corollary of our theorem and show it strictly outperforms previously known coding schemes. In Section~\ref{sec:Gaussian}, we present a unified coding theorem for the Gaussian case. We conclude this paper in Section \ref{sec:conclusion}.

The following notations are used throughout the paper.
\subsection{Notation}
For two integers $i$ and $j$, $[i:j]$ denotes the set $\{i,i+1,\ldots, j\}$. For a set $S$ of real numbers,  $S_{[i]}$ denotes the $i$-th smallest element in $S$ and $S[i]$ denotes $\{j: j\in S, j<i\}$. For constants $u_1,\ldots, u_k$ and $S\subseteq [1:k]$, $u_S$  denotes the  vector $(u_j: j\in S)$ and $u^j_i$ denotes $u_{[i:j]}$ where the subscript is omitted when $i=1$, i.e., $u^j=u_{[1:j]}$.  For random variables $U_1,\ldots, U_k$ and $S\subseteq [1:k]$, $U_{S}$ and  $U^j_i$ are defined similarly. For sets $T_1, \ldots, T_k$ and $S\subseteq [1:k]$, $T_S$ denotes $\bigcup_{j\in S} T_j$ and $T^j_i$ denotes $T_{[i:j]}$ where the subscript is omitted when $i=1$. Consider two real vectors $u=(u_1,\ldots, u_k)$ and $v=(v_1,\ldots, v_k)$ of length $k$. We say that $u$ is smaller than $v$ and write $u<v$ if there exists $k'\in [1:k]$ such that $u_{j}= v_{j}$ for all $j\in [1:k'-1]$ and $u_{k'}<v_{k'}$. Furthermore, we say that $u$ is component-wise smaller than $v$ and write $u\prec v$ if $u_j<v_j$ for all $j\in [1:k]$.
 $\mathbf{1}$ denotes an all-one  vector and $\mathbf{I}$ denotes an identity matrix. When $U$ is a Gaussian random vector with mean $\mu$ and covariance matrix $\Lambda_U$, we write $U \sim \mathcal{N}(\mu, \Lambda_{U})$. $\mathbbm{1}_{u=v}$ is the indicator function, i.e., it is 1 if $u=v$ and 0 otherwise. $\delta(\epsilon)>0$ denotes a function of $\epsilon$ that tends to zero as $\epsilon$ tends to zero. 

We follow the notion of typicality in \cite{Orlitsky:01}, \cite{ElGamalKim:11}. Let $\pi_{x^n}(x)$ denote the number of occurrences of $x\in \mathcal{X}$ in the sequence $x^n$. 
Then, $x^n$ is said to be $\epsilon$-typical (or just typical) for $\epsilon>0$ if for every $x\in \mathcal{X}$,
\begin{align*}
|\pi_{x^n}(x)/n-p(x)|\leq \epsilon p(x) .
\end{align*}
The set of all $\epsilon$-typical $x^n$ is denoted as $\mathcal{T}_{\epsilon}^{(n)}(X)$, which is shortly denoted as $\mathcal{T}_{\epsilon}^{(n)}$. A jointly typical set (or just a typical set) such as $\mathcal{T}_{\epsilon}^{(n)}(X,Y)$ for multiple variables, which will also be denoted as $\mathcal{T}_{\epsilon}^{(n)}$, is naturally defined from the definition of $\mathcal{T}_{\epsilon}^{(n)}(X)$.

\section{Unified Framework} \label{sec:framework}
In this section, we build a unified framework for proving the achievability of many network information theory problems including channel coding, source coding, joint source--channel coding, and coding for computing. Let us first construct a unified framework for point-to-point scenarios and then generalize it to general network scenarios. 
\subsection{Point-to-point scenarios} \label{subsec:framework_p2p}
Consider the standard channel coding and source coding problems \cite{Shannon:48} illustrated in Fig. \ref{fig:p2p}. These two problems can be stated with the following elements: information to be communicated, node interaction and node processing functions, and the definition of achievability. Let us investigate differences between these two coding problems for each element and discuss how we can unify them into a single framework. In the following, $n$ denotes the number of channel uses for channel coding and the number of source symbols for source coding and $R\geq 0$ denotes the rate in each problem. 

\begin{figure}[t]
 \centering\subfigure[]
  {\includegraphics[width=105mm]{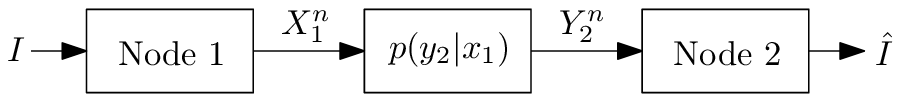}}
 \subfigure[]
  {\includegraphics[width=75mm]{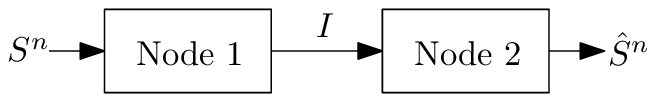}}
 \caption{(a) Channel coding, (b) Source coding} \label{fig:p2p}
\end{figure}

\begin{itemize}
\item Information to be communicated: In channel coding, a message $I$, uniformly distributed over $[1:2^{nR}]$, is communicated from node 1 to node 2. In source coding, a discrete memoryless source (DMS) $(\mathcal{S},p(s))$ is given to node 1 and is reconstructed at node 2 (up to a prescribed distortion level in the case of lossy source coding). We can observe that a message can be regarded as a DMS $(\mathcal{S},p(s))$ such that $H(S)=R$. Hence, in both channel coding and source coding problems, we can say that a DMS  $(\mathcal{S},p(s))$ is given to node 1  and is reconstructed at node 2.

\item Node interaction and node processing functions: In channel coding, node 1 communicates with node 2 through a discrete memoryless channel (DMC) $(\mathcal{X}_1, \mathcal{Y}_2$, $p(y_2|x_1))$. Node 1 maps $s^n$ to a channel input sequence $x_1^n$ and node 2 receives a channel output sequence $y_2^n$ and maps it to $\hat{s}^n$. In source coding, node 1 maps $s^n$ to an index $I\in [1:2^{nR}]$ and node 2 receives $I$ exactly and maps it to $\hat{s}^n$. The noiseless communication of an index in source coding can be regarded as a DMC $(\mathcal{X}_1, \mathcal{Y}_2, p(y_2|x_1))$ such that $\max_{p(x_1)}I(X_1;Y_2)=R$. Hence, in  both channel coding and source coding problems, we can say that node 1 communicates with node 2 through a DMC $(\mathcal{X}_1, \mathcal{Y}_2$, $p(y_2|x_1))$, the processing function at node 1 is a mapping from $S^n$ to $X_1^n$, and the processing function at node 2 is a mapping from $Y_2^n$ to $\hat{S}^n$. By denoting $S$ by $Y_1$ and $\hat{S}$ by $X_2$, we can further unify the notation for sequences and the node processing functions, i.e., a sequence received by node $k$ is denoted by $Y_k^n$, the resultant sequence from processing at node $k$ is denoted by $X_k^n$, and the node processing function at node $k=1,2$ is a mapping from $Y_k^n$ to $X_k^n$.

\item Achievability: In channel coding and lossless source coding problems, a rate $R$ is said to be achievable if there exists a sequence of node processing functions such that $\lim_{n\rightarrow\infty} P_e^{(n)}=0$, where 
$P_e^{(n)}$ denotes the probability of error event given as  
$P(Y_1^n\neq X_2^n)$. 
In lossy source coding problem, a rate--distortion pair $(R,d)$ is said to be achievable if there exists a sequence of node processing functions such that  $\limsup_{n\rightarrow \infty}\E(d(Y_1^n, X_2^n))\leq d$, where $d(\cdot,\cdot)\geq 0$ is a distortion measure between two arguments.

Now, let us introduce a new definition of achievability from which we can show the achievability of both channel coding and source coding problems in a unified way. We say a joint distribution $p^*(x_1,x_2,y_1,y_2)$, shortly denoted as $p^*$, is achievable if there exists a sequence of node processing functions such that $\lim_{n\rightarrow \infty }P_e^{(n)}(p^*,\epsilon)= 0$ for any $\epsilon>0$, where $P_e^{(n)}(p^*,\epsilon)$ denotes the probability \[P((X_1^n, X_2^n, Y_1^n, Y_2^n)\notin T_{\epsilon}^{(n)})\] in which the typical set is defined with respect to $p^*$. Then, the achievability of appropriately  chosen $p^*$ implies the achievability of $R$ or $(R,d)$ in channel coding and source coding problems. For channel coding and lossless source coding problems, $R$ is achievable if $p^*$ such that $X_2=Y_1$ is achievable. For lossy source coding problem, $(R,d)$ is achievable if $p^*$ such that $E(d(Y_1,X_2))\leq d/(1+\epsilon')$, $\epsilon'\rightarrow 0$ is achievable from the typical average lemma  \cite{ElGamalKim:11} and the continuity of the rate-distortion function $R(d)$ in $d$. 
\end{itemize}

To see whether the aforementioned unification approach is general enough for point-to-point scenarios, let us consider more general point-to-point scenarios in Fig. \ref{fig:p2p_side}. First, in channels with noncausal states \cite{GelfandPinsker:80} illustrated in Fig. \ref{fig:p2p_side}-(a), node 1 observes a message $I$ of rate $R$ and a state sequence $S^n\sim p(s)$ and encodes $(I, S^n)$ as $X_1^n$. Then, node 2 receives $Y_2^n \sim p(y_2|s,x_1)$ and estimates $I$ as $\hat{I}$. Achievability is defined in the same way as in the channel coding problem. Let us apply the aforementioned unification approach to this problem. Since $Y_1$ represents all the information node 1 receives, we let $Y_1=(M, S)$ such that $H(M)=R$ and $M$ and $S$ are independent, where $M^n$ corresponds to the message of rate $R$. But, we cannot use the channel form of $p(y_2|x_1)$ to capture the dependency of the channel output $Y_2$ on state $S$. This indicates that a more general channel form of $p(y_2|y_1,x_1)$ is needed in the unified framework. Then, we can let $p(y_2|y_1,x_1)$ be equal to $p(y_2|s,x_1)$. If we choose  $p^*$ such that $X_2=M$, the achievability of $p^*$ implies the achievability of $R$ of the original problem. 
 
Next, in lossy source coding with side information \cite{WynerZiv:76} represented in Fig. \ref{fig:p2p_side}-(b), node 1 receives a source sequence $S^n\sim p(s)$ and encodes it as an index $I\in [1:2^{nR}]$. Then, node 2 receives the index $I$ and side information $T^n\sim p(t|s)$ and reconstructs $S^n$ as $\hat{S}^n$ up to some distortion level. Achievability is defined in the same way as in the lossy source coding problem. 
For this problem, we  apply the unification approach as follows. We let $Y_1=S$. Since node 2 has two channel inputs, we let $Y_2=(Y_2',T)$ and let the  channel $p(y_2|y_1,x_1)$ be decomposed as $p(y_2'|x_1)p(t|y_1)$, where the channel $p(y_2'|x_1)$ corresponds to the communication of $I$ of rate $R$ and hence its capacity is given as $R$, i.e., $\max_{p(x_1)}I(X_1;Y_2')=R$, and the channel $p(t|y_1)$ captures the correlation between $Y_1=S$ and the side information $T$. We pick up the target distribution in the same way as in the lossy source coding problem.  Furthermore, coding for computing problem  \cite{Yamamoto:82}, where node 2 wishes to reconstruct a function $f(S,T)$ of $S$ and $T$ up to  distortion $d$ with respect to a distortion measure $d(\cdot,\cdot)$, can also be included in this framework by choosing $p^*$ such that $\E(d(g(Y_1,Y_2),X_2))\leq d/(1+\epsilon)$, $\epsilon\rightarrow 0$ where $g(Y_1,Y_2)=g(S,Y_2',T)=f(S,T)$. 

In summary, the achievability of the aforementioned point-to-point coding problems can be shown by considering the following unified framework. Network model is given by $(\mathcal{X}_1,\mathcal{X}_2, \mathcal{Y}_1,\mathcal{Y}_2,p(y_1)p(y_2|y_1,x_1))$ as illustrated in Fig. \ref{fig:p2p_framework} and the objective is specified by a target distribution $p^*$. $p^*$ is said to be achievable if there exists a sequence of node processing functions, $Y_k^n\rightarrow X_k^n$, $k=1,2$, such that $\lim_{n\rightarrow \infty}P_e^{(n)}(p^*,\epsilon)=0$ for any $\epsilon>0$.

\begin{figure}[t]
 \centering\subfigure[]
  {\includegraphics[width=102mm]{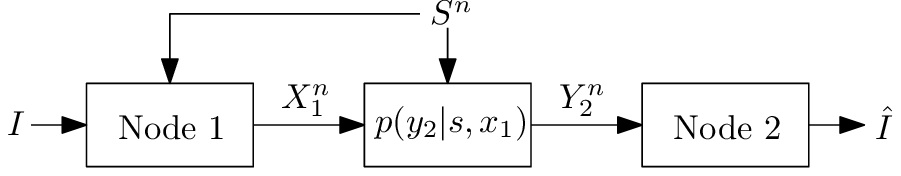}}
 \subfigure[]
  {\includegraphics[width=75mm]{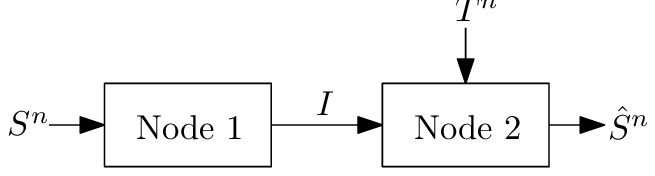}}
 \caption{(a) Channel with noncausal states, (b) Lossy source coding with side information } \label{fig:p2p_side}
\end{figure}
\begin{figure}[t]
 \centering
  {
   \includegraphics[width=112mm]{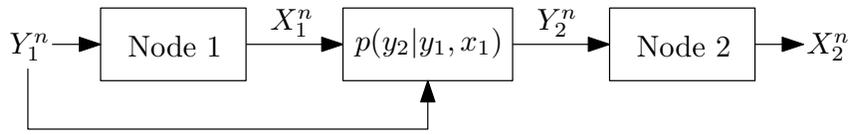}}
  \caption{Unified framework for point-to-point scenarios} \label{fig:p2p_framework}
\end{figure}


\subsection{General scenarios} 
In this subsection, we generalize the unified framework in Section \ref{subsec:framework_p2p} to general $N$-node networks. In our unified framework for $N$ nodes, we define an $N$-node acyclic discrete memoryless network (ADMN) $(\mathcal{X}_1,\ldots,\mathcal{X}_N$, $\mathcal{Y}_1,\ldots, \mathcal{Y}_N,\prod_{k=1}^Np(y_k|y^{k-1},x^{k-1}))$, which consists of a set of alphabet pairs $(\mathcal{X}_k, \mathcal{Y}_k)$, $k\in [1:N]$ and a collection of conditional pmfs $p(y_k|y^{k-1},x^{k-1})$, $k\in [1:N]$. Here, $Y_k$ and $X_k$ represent any information that comes into and goes out of node $k$, respectively. $Y_k$ can be a channel output, message, source, non-causal state information, and any combination of those. $X_k$ can be a channel input, message estimate, reconstructed source, action for generating channel state, and any combination of those. Next, $p(y_k|y^{k-1},x^{k-1})$ signifies the correlation betweeen information prior to node $k$ and information received at node $k$. It can capture channel distribution possibly with states, correlation between distributed sources, and complicated network-wide correlation among sources and channels.

In this network, information flows in one direction and node operations are sequential. Let $n$ denote the number of channel uses. First, $Y_1^n$ is generated according to $\prod_{i=1}^n p(y_{1,i})$ and then node 1 processes $X_1^n$ based on $Y_1^n$. Next, $Y_2^n$ is generated according to $\prod_{i=1}^n p(y_{2,i}|x_{1,i},y_{1,i})$ and then node 2 encodes $X_2^n$ based on $Y_2^n$. Similarly, $Y_k^n$ is generated according to $\prod_{i=1}^np(y_{k,i}|x^{k-1}_i,y^{k-1}_i)$ and node $k$ encodes $X_k^n$ based on $Y_k^n$ for $k\in [1:N]$.  Clearly, any layered network \cite{AvestimehrDiggaviTse:11} or noncausal network (without infinite loop) \cite{ElGamalHassanpourMammen:07} possibly with noncausal state or side information is represented as an ADMN. Furthermore, any strictly causal (usual discrete memoryless network with relay functions having one sample delay) or causal network (relays without delay \cite{ElGamalHassanpourMammen:07})  with blockwise operations can be represented as an ADMN by unfolding the network. Note that our unified achievability theorem (Theorem \ref{thm:main}) still applies to the unfolded network. Therefore, considering only \emph{acyclic} DMN (ADMN) in our unified approach is without loss of generality while greatly simplifying our unification approach. In the following subsection, we show several known examples represented by an ADMN. 

Achievability is specified using a target joint distribution $p^*(x^N,y^N)$, which is shortly denoted as $p^*$. For a set of node processing functions $Y_k^n\rightarrow X_k^n$, $k=1,\ldots, N$, the $\epsilon$-probability of error is defined as $P_e^{(n)}(p^*, \epsilon)=P((X_{[1:N]}^n,Y_{[1:N]}^n) \notin \mathcal{T}_{\epsilon}^{(n)} )$, where the typical set $\mathcal{T}_{\epsilon}^{(n)}$ is defined with respect to $p^*$. We say the target distribution $p^*$ is achievable if there exists a sequence of node processing functions $Y_k^n\rightarrow X_k^n$, $k=1,\ldots, N$, such that $\lim_{n\rightarrow \infty} P_e^{(n)}(p^*,\epsilon)=0$ for any $\epsilon>0$. We note that $p^*$ unifies diverse network demands and constaints. It can be used for designating the source--destination relationship and for imposing distortion and cost constraints. 

\subsection{Examples}
In this subsection, we represent some network information theory problems by an ADMN and a target distribution $p^*$ such that the achievability of $p^*$ implies the achievability of the original problem. Let us first consider some examples of single-hop networks.
\begin{example}[Multiple access channels \cite{liao_thesis}] For multiple access channel problem with rates $R_1$ and $R_2$, we choose $N=3$, $H(Y_1)=R_1$, $p(y_2|x_1,y_1)=p(y_2)$, $H(Y_2)=R_2$, $p(y_3|x_1,x_2,y_1,y_2)=p(y_3|x_1,x_2)$, and $p^*$ such that $X_3=(Y_1,Y_2)$. \label{ex:mac}
\end{example}

\begin{example} [Distributed lossy compression \cite{Berger:77}, \cite{tung_thesis}]
For distributed lossy compression problem with rate--distortion pairs $(R_1,d_1)$ and $(R_2,d_2)$, we let $N=3$,  $p(y_2|x_1,y_1)=p(y_2|y_1)$, $\mathcal{Y}_3=\mathcal{Y}_{3,1}\times \mathcal{Y}_{3,2}$,  $p(y_3|x_{[1:2]},y_{[1:2]})=p(y_{3,1}|x_1)p(y_{3,2}|x_2)$ such that $\max_{p(x_1)}I(X_1;Y_{3,1})=R_1$ and $\max_{p(x_2)}I(X_2;Y_{3,2})=R_2$, and $p^*$ such that $\E[d_k(Y_k,X_{3})]\leq \frac{d_k}{1+\epsilon}$ for $k\in [1:2]$, where $d_k(\cdot,\cdot)\ge 0$ is a distortion  measure between two arguments and $\epsilon\rightarrow 0$. \label{ex:distribuedlc}
\end{example}

\begin{example}[Broadcast channels \cite{Marton:79}] For broadcast channel problem with rates $R_1$ and $R_2$, we choose $N=3$, $Y_1=(Y_{1,1},Y_{1,2})$, $p(y_{1,1},y_{1,2})=p(y_{1,1})p(y_{1,2})$, $H(Y_{1,1})=R_1$, $H(Y_{1,2})=R_2$, $p(y_2|y_1,x_1)=p(y_2|x_1)$, $p(y_3|y_{[1:2]},x_{[1:2]})=p(y_3|x_1,y_2)$, and $p^*$ such that $X_2=Y_{1,1}$, $X_3=Y_{1,2}$. \label{ex:bc}
\end{example}

\begin{example}[Multiple description coding \cite{CoverElGamal:82}] For multiple description coding with rates $(R_1, R_2)$ and distortion triples $d_1, d_2$, and $d_3$, we choose $N=4$, $\mathcal{X}_1=\mathcal{X}_{1,1}\times \mathcal{X}_{1,2}$, $p(y_2|x_1,y_1)=p(y_2|x_{1,1})$ such that $\max_{p(x_{1,1})}I(X_{1,1};Y_2)=R_1$, $p(y_3|x_{[1:2]},y_{[1:2]})=p(y_3|x_{1,2})$ such that $\max_{p(x_{1,2})}I(X_{1,2};Y_3)=R_2$, $Y_4=(Y_2,Y_3)$, and $p^*$ such that $\E[d_k(Y_1,X_k)]\leq \frac{d_k}{1+\epsilon}$ for $k\in [2:4]$, where $d_k(\cdot,\cdot)\ge 0$ is a distortion  measure between two arguments and $\epsilon\rightarrow 0$. \label{ex:mdc}
\end{example}

Next, we show an example of multi-hop networks. 
\begin{example}[Relay channels] 
Consider a three-node relay channel $(\mathcal{X}_1,\mathcal{X}_2,\mathcal{Y}_2,\mathcal{Y}_3$, $p(y_2,y_3|x_1,x_2))$ illustrated in Fig. \ref{fig:three_relay}-(a), where node 1 wishes to send a message to node 3 with the help of node 2.  Let $R$ and $n$ denote the rate and the number of channel uses, respectively, and let $I$ and $\hat{I}$ denote the message of rate $R$ at node 1 and the estimated message at node 3, respectively. Then, the node processing function at node 1 is a mapping from $[1:2^{nR}]$ to $\mathcal{X}_1^n$, the node processing function at node 2 at time $i\in [1:n]$ is a mapping from $\mathcal{Y}_2^{i-1}$ to $\mathcal{X}_{2}$, and the node processing function at node 3 is a mapping from $\mathcal{Y}_3^n$ to $[1:2^{nR}]$. The probability of error is defined as $P_e^{(n)}=P(\hat{I} \neq I )$ and a rate $R$ is said to be achievable if there exists a sequence of node processing functions such that $\lim_{n\rightarrow \infty} P_e^{(n)}=0$. 

If we assume a blockwise operation at each node, we can represent this network as an ADMN by unfolding the network. Assume $B$ transmission blocks, each consisting of $n$ channel uses. 
In the unfolded network illustrated in Fig. \ref{fig:three_relay}-(b), we have $3(B+1)$ nodes and the operation of node $(k,b), k\in [1:3], b\in [1:B+1]$ corresponds to that of node $k$ of the original network at the end of block $b-1$. To reflect the fact that node $(k,b+1)$ is originally the same node as node $(k,b)$, we assume that node $(k,b)$ has an orthogonal  link of sufficiently large rate to node $(k,b+1)$, which is represented as a dashed line in Fig. \ref{fig:three_relay}-(b). Because this unfolded network is acyclic, it can be represented as an ADMN and $p^*$ can be chosen accordingly. 
\end{example}
 
\begin{figure}[t]
 \centering\subfigure[]
  {\includegraphics[width=60mm]{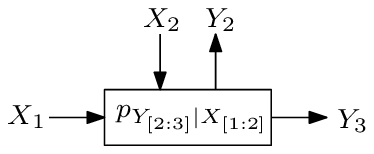}}\\
 \subfigure[]
  {\includegraphics[width=140mm]{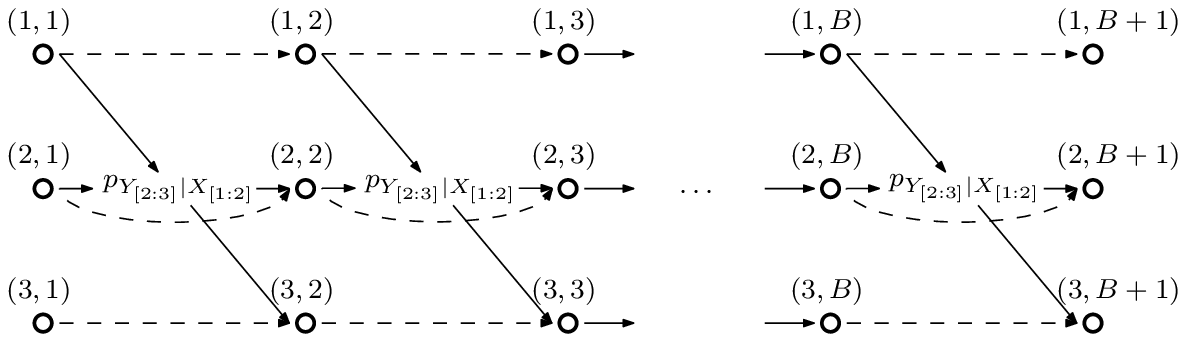}}\\
 \caption{Three-node relay network is illustrated in (a). The corresponding unfolded network is shown in (b). In the unfolded network, the operation of node $(k,b)$ corresponds to that of node $k$ of the original network at the end of block $b-1$.} \label{fig:three_relay}
\end{figure}

\subsection{Introduction of a virtual node}
The following two propositions are obtained by introducing a virtual node in an ADMN, which turn out to be useful in recovering some known achievability results in Section \ref{sec:main}. 
 
\begin{proposition} \label{lemma:virtual_node}
Consider an $N$-node ADMN \[(\mathcal{X}_1,\ldots,\mathcal{X}_N, \mathcal{Y}_1,\ldots, \mathcal{Y}_N, \prod_{k=1}^N p(y_k|y^{k-1},x^{k-1}))\] and target distribution $p^*$. 
For some $v\in [1:N]$ and finite set $\mathcal{Y}$, assume $p(y|x^{v-1},y^{v})$ for $y\in \mathcal{Y}$. Then, we have 
\begin{align*}
p(y|x^{v-1},y^{v-1})&=\sum_{y_v}p(y_v|x^{v-1},y^{v-1})p(y|x^{v-1},y^{v})\\
p(y_{v}|x^{v-1},y^{v-1},y)&=\frac{p(y_v|x^{v-1},y^{v-1})p(y|x^{v-1},y^{v})}{\sum_{y_v}p(y_v|x^{v-1},y^{v-1})p(y|x^{v-1},y^{v})}.
\end{align*}

Now, consider an $(N+1)$-node ADMN \[(\mathcal{X}_1',\ldots,\mathcal{X}_{N+1}', \mathcal{Y}_1',\ldots, \mathcal{Y}_{N+1}', \prod_{k=1}^{N+1} p'(y_k|y^{k-1},x^{k-1}))\] and target distribution $p'^*$ such that 
\begin{align*}
\mathcal{X}_k'=\begin{cases}
\mathcal{X}_k \mbox{ if } k<v \\
\emptyset \mbox{ if } k=v \\
\mathcal{X}_{k-1} \mbox{ if } k>v
\end{cases}, ~~~
\mathcal{Y}_k'=\begin{cases}
\mathcal{Y}_k \mbox{ if } k<v \\
\mathcal{Y} \mbox{ if } k=v \\
\mathcal{Y}_{k-1} \mbox{ if } k>v
\end{cases},
\end{align*}
\begin{align*}
p'(y_{k}|x^{k-1},y^{k-1})=\begin{cases}
p_{Y_{k}|X^{k-1},Y^{k-1}}(y_{k}|x^{k-1},y^{k-1}) \mbox{ if } k<v\\
p_{Y|X^{k-1},Y^{k-1}}(y_{k}|x^{k-1},y^{k-1}) \mbox{ if } k=v\\
p_{Y_{v}|X^{v-1},Y^{v-1},Y}(y_{k}|x^{v-1},y^{v-1},y_v) \mbox{ if } k=v+1\\
p_{Y_{k-1}|X^{k-2},Y^{k-2}}(y_{k}|x_{[1:k-1]\setminus \{v\}}, y_{[1:k-1]\setminus \{v\}}) \mbox{ if } k>v+1
\end{cases},
\end{align*}
and \[\sum_{x_v, y_v}p'^*(x^{N+1},y^{N+1})=p^*(x^{v-1},x_{v+1}^{N+1},y^{v-1},y_{v+1}^{N+1}).\]

Then, if $p'^*$ is achievable for the $(N+1)$-node ADMN, $p^*$ is achievable for the $N$-node ADMN. 
\end{proposition}
\begin{proof}
The proof is straightforward from the observation that the $(N+1)$-node ADMN is obtained by introducing a virtual node, whose channel output is $Y$ and channel input is null, between nodes $v-1$ and $v$ in the $N$-node ADMN and reindexing the nodes.
\end{proof}

\begin{proposition} \label{lemma:common_part}
Consider an $N$-node ADMN \[(\mathcal{X}_1,\ldots,\mathcal{X}_N, \mathcal{Y}_1,\ldots, \mathcal{Y}_N, \prod_{k=1}^N p(y_k|y^{k-1},x^{k-1}))\] such that $p(y_{v_1}|y^{v_1-1},x^{v_1-1})=p(y_{v_1})$ and $p(y_{v_2}|y^{v_2-1},x^{v_2-1})=p(y_{v_2}|y_{v_1})$ for some $v_1\in [1:N]$, $v_2\in [1:N]$, $v_1<v_2$. Let $Y$ denote the common part of two random variables $Y_{v_1}$ and $Y_{v_2}$, where the common part of two discrete memoryless sources is defined in \cite{GacsKorner:73}, \cite{Witsenhausen:75}

Now, consider an $(N+1)$-node ADMN \[(\mathcal{X}_1',\ldots,\mathcal{X}_{N+1}', \mathcal{Y}_1',\ldots, \mathcal{Y}_{N+1}', \prod_{k=1}^{N+1} p'(y_k|y^{k-1},x^{k-1}))\] and target distribution $p'^*$ such that 
\begin{align*}
\mathcal{X}_k'=\begin{cases}
\mathcal{X}_k \mbox{ if } k<v_1 \\
\mathcal{X} \mbox{ if } k=v_1 \\
\mathcal{X}_{k-1} \mbox{ if } k>v_1
\end{cases}, ~~~
\mathcal{Y}_k'=\begin{cases}
\mathcal{Y}_k \mbox{ if } k<v_1 \\
\mathcal{Y} \mbox{ if } k=v_1 \\
\mathcal{Y}_{k-1}\times \mathcal{X} \mbox{ if } k=v_1+1 \mbox{ or } k=v_2+1 \\
\mathcal{Y}_{k-1} \mbox{ otherwise } 
\end{cases},
\end{align*}
\begin{align*}
p'(y_{k}|x^{k-1},y^{k-1})=
\begin{cases}
p_{Y_{k}|X^{k-1},Y^{k-1}}(y_{k}|x^{k-1},y^{k-1}) \mbox{ if } k<v_1\\
p_{Y}(y_{k}) \mbox{ if } k=v_1\\
p_{Y_{v_1}}(y_{k,1})\mathbbm{1}_{y_{k,2}=x_{v_1}} \mbox{ if } k=v_1+1\\
p_{Y_{v_2}|Y_{v_1}}(y_{k,1}|y_{v_1+1,1})\mathbbm{1}_{y_{k,2}=x_{v_1}}  \mbox{ if } k=v_2+1\\
p_{Y_{k-1}|X^{k-2},Y^{k-2}}(y_{k}|x_{[1:k-1]\setminus \{v_1\}}, y_{[1:k-1]\setminus \{v_1\}}) \mbox{ otherwise } \end{cases}
\end{align*}
and \[\sum_{x_{v_1}, y_{v_1}}p'^*(x^{N+1},y^{N+1})=p^*(x^{v_1-1},x_{v_1+1}^{N+1},y^{v_1-1},y_{v_1+1}^{N+1}),\]
where $y_k=(y_{k,1}, y_{k,2})$ for $k=v_1+1$ or $k=v_2+1$ and $|\mathcal{X}|$ can be arbitrarily large. 

Then, if $p'^*$ is achievable for the $(N+1)$-node ADMN, $p^*$ is achievable for the $N$-node ADMN.
\end{proposition}
\begin{proof}
Note that in the $N$-node ADMN, both nodes $v_1$ and $v_2$ observe the common part $Y$ and hence can share any function of $Y^n$. Thus, we can introduce a virtual node whose channel output is $Y$ and channel input is $X$ and assume that $X^n$ is available at nodes $v_1$ and $v_2$.  
\end{proof}

\section{Unified Coding Theorem} \label{sec:main}
In this section, we propose a unified coding scheme and present the main theorem of this paper, followed by various examples that show how to utilize our results. Our scheme consists of the following ingredients:  1) superposition, 2) simultaneous nonunique decoding, 3) simultaneous compression, and 4) symbol-by-symbol mapping. These are tweaked and combined in a special way to enable unification of many previous approaches. Let us first briefly explain the proposed scheme and introduce related coding parameters. Detailed description of our scheme is given in the proof of Theorem \ref{thm:main}.

\begin{itemize} 
\item Codebook generation: In our coding scheme, covering codebooks are used to compress information that each node observes and decodes. 
We generate $\nu$ covering codebooks $\mathcal{C}_1,\ldots,\mathcal{C}_{\nu}$. Let $\mathcal{U}_j$ for $j\in [1:\nu]$ denote the alphabet for the codeword symbol of $\mathcal{C}_j$. For indexing of codewords, we consider $\mu$ index sets $\mathcal{L}_1,\ldots, \mathcal{L}_{\mu}$, where $\mathcal{L}_j=[1:2^{nr_j}]$ for some $r_j\geq 0$ for each $j\in [1:\mu]$. We denote by $\Gamma_j\subseteq [1:\mu]$ the set of indices of $\mathcal{L}$'s associated with $\mathcal{C}_j$ in a way that each codeword in $\mathcal{C}_j$ is indexed by the vector $l_{\Gamma_j}\in \prod_{i\in \Gamma_j} \mathcal{L}_i$ and hence $\mathcal{C}_j$ consists of $2^{n\sum_{i\in \Gamma_j}r_i}$ codewords, i.e., $\mathcal{C}_j=\{u_j^n(l_{\Gamma_j}): l_{\Gamma_j}\in \prod_{i\in \Gamma_j} \mathcal{L}_i\}$.  
Each codebook is constructed allowing superposition coding. Let $A_j\subseteq [1:\nu], j\in [1:\nu]$ denote the set of the indices of $\mathcal{C}$'s on which $\mathcal{C}_j$ is constructed by superposition.

\item Node operation: Node $k\in [1:N]$ operates according to the following three steps as illustrated in Fig. \ref{fig:scheme_flow}. 
\begin{itemize}
\item Simultaneous nonunique decoding:  After receiving $y_k^n$, node $k$ decodes some covering codewords of previous nodes simultaneously, where some are decoded uniquely and the others are decoded non-uniquely. 
We denote by $D_k\subseteq [1:\nu]$ and $B_k\subseteq [1:\nu]$ the sets of the indices of $\mathcal{C}$'s whose codewords are decoded uniquely and non-uniquely, respectively, at node $k$. 

\item Simultaneous compression: After decoding, node $k$ finds covering codewords $u_{W_k}^n$ simultaneously according to a conditional pmf $p(u_{W_k}|u_{D_k}, y_k)$ that carry some information about the received channel output sequence $y_k^n$ and uniquely decoded codewords $u_{D_k}^n$, where $W_k\subseteq [1:\nu]$ denotes the set of the indices of  $\mathcal{C}$'s  used for compression. 

\item Symbol-by-symbol mapping: After decoding and compression, node $k$ generates $x_k^n$ by a symbol-by-symbol mapping from uniquely decoded codewords $u_{D_k}^n$,  covered codewords $u_{W_k}^n$, and received channel output sequence $y_k^n$. Let $x_k(u_{D_k}, u_{W_k}, y_k)$ denote the function used for symbol-by-symbol mapping. 
\end{itemize}

\end{itemize}

\begin{figure}[t]
 \centering
  {
   \includegraphics[width=110mm]{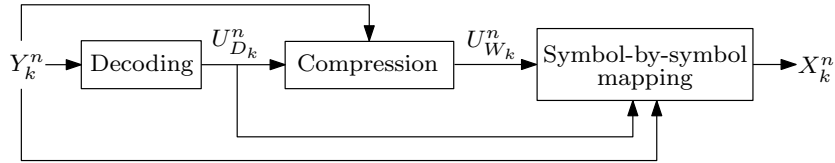}}
  \caption{Node $k\in [1:N]$ operates in three steps: 1) simultaneous nonunique decoding, 2) simultaneous compression, and 3) symbol-by-symbol mapping.} \label{fig:scheme_flow}
\end{figure}

In summary, our scheme requires the following set $\omega$ of coding parameters, where some constraints are added to  make the aforementioned codebook generation and node operation proper:
\begin{enumerate}
\item positive integers $\mu$ and $\nu$
\item alphabets $\mathcal{U}_j, j\in [1:\nu]$
\item $\mu$-rate tuple $(r_1,\ldots, r_{\mu})$
\item sets $\Gamma_j\subseteq [1:\mu]$, $A_j\subseteq [1:\nu]$, $\uin{D}_k\subseteq \uin{W}^{k-1}$, $\uin{B}_k\subseteq \uin{W}^{k-1}\setminus \uin{D}_k$, and  $\uin{W}_k\subseteq [1:\nu]\setminus \uin{W}^{k-1}$  for $k\in [1:N]$ and $j\in [1:\nu]$ that satisfy 
\begin{enumerate}[{A}-1]
\item $\Gamma_{\uin{W}_k}\setminus \Gamma_{\uin{D}_k}$'s are disjoint,
\item $\Gamma_{A_j}\subseteq \Gamma_j$ and  $j'<j$ if $j'\in A_j$,
\item $A_{\uin{W}_k}\subseteq \uin{W}_k\cup \uin{D}_k$, $A_{\uin{B}_k}\subseteq \uin{D}_k\cup \uin{B}_k$, and $A_{\uin{D}_k}\subseteq \uin{D}_k$.
\end{enumerate}
\item a set of conditional pmfs $p(u_{\uin{W}_k}|u_{\uin{D}_k},y_k)$ and functions $x_k(u_{\uin{D}_k},u_{\uin{W}_k},y_k)$ for $k\in [1:N]$ such that $p(x_{[1:N]},y_{[1:N]})$ induced by 
\begin{align}
\prod_{k=1}^N p(y_k|y^{k-1},x^{k-1})p(u_{\uin{W}_k}|u_{\uin{D}_k},y_k)\mathbbm{1}_{x_k=x_k(u_{\uin{D}_k},u_{\uin{W}_k},y_k)} \label{eqn:omega_distr}
\end{align}
is the same as the target distribution $p^*(x_{[1:N]},y_{[1:N]})$.
\end{enumerate}

Now, we are ready to present our main theorem, which gives a sufficient condition for achievability using the aforementioned scheme. For an ADMN $(\mathcal{X}_1,\ldots,\mathcal{X}_N, \mathcal{Y}_1,\ldots, \mathcal{Y}_N, \prod_{k=1}^N p(y_k|y^{k-1},x^{k-1}))$ and target distribution $p^*$, let $\Omega(\mathcal{X}_1,\ldots,\mathcal{X}_N, \mathcal{Y}_1,\ldots, \mathcal{Y}_N, \prod_{k=1}^N p(y_k|y^{k-1},x^{k-1}), p^*)$, shortly denoted as $\Omega(p^*)$ or $\Omega$,  denote the set of all possible $\omega$'s. 
\begin{theorem} \label{thm:main}
For an  $N$-node ADMN, $p^*$ is achievable if there exists $\omega\in \Omega$ such that for $1\leq k \leq N$ 
\begin{align}
\sum_{j\in \lset{S}_k} r_j &<\sum_{j\in  \uset{S}_k} I(U_j;U_{ \uset{S}_k[j]\cup  \uset{S}_k^c},Y_k|U_{A_j}) \label{eqn:main_sk}\\ 
\sum_{j\in \lset{T}_k} r_j&>\sum_{j\in  \uset{T}_k} I(U_j;U_{ \uset{T}_k[j]\cup \uin{D}_k},Y_k|U_{A_j})\label{eqn:main_tk}
\end{align}
for all $\lset{S}_k\subseteq \lin{D}_k\cup \lin{B}_k$ such that $\lset{S}_k\cap \lin{D}_k\neq \emptyset$ and for all $\lset{T}_k\subseteq \lin{W}_k$ such that $\lset{T}_k\neq \emptyset$, where $\lin{D}_k\triangleq \Gamma_{\uin{D}_k}$, $\lin{B}_k\triangleq \Gamma_{\uin{B}_k}\setminus \Gamma_{\uin{D}_k}$, $\lin{W}_k\triangleq \Gamma_{\uin{W}_k}\setminus \Gamma_{\uin{D}_k}$, 
\begin{align}
 \uset{S}_k&\triangleq \{j:j\in \uin{D}_k\cup \uin{B}_k, \Gamma_j\cap \lset{S}_k\neq \emptyset \}, \label{eqn:barSk}\\
 \uset{T}_k&\triangleq \{j:j\in \uin{W}_k, \Gamma_j\cap (\lset{T}_k\cup \lin{D}_k)^c=\emptyset \}. \label{eqn:barTk} 
\end{align}
\end{theorem}

\begin{remark}
For $k\in [1:N]$, the inequalities (\ref{eqn:main_sk}) and (\ref{eqn:main_tk}) are the conditions for successful simultaneous nonunique decoding and simultaneous compression, respectively, at node $k$. 
\end{remark}
\begin{remark}
Theorem \ref{thm:main} can be improved using coded time sharing \cite{HanKobayashi:81}.
\end{remark}

\begin{proof}
Consider $\omega\in \Omega$. Let $0<\epsilon_k<\epsilon_k'<\epsilon_k''$ for all $k\in [1:N]$ such that $\epsilon_{k-1}''<\epsilon_k$ and $\epsilon''_N<\epsilon$. Let $\mathcal{L}_j= [1:2^{nr_j}]$ for $j\in [1:\mu]$. In the following, $l_j\in\mathcal{L}_j$ for $j\in [1:\mu]$.

\subsubsection{Codebook generation} 
For each $j\in [1:\nu]$ and $l_{\Gamma_j}\in \prod_{i\in \Gamma_j} \mathcal{L}_i$, generate $u_j^n(l_{\Gamma_j})$ conditionally independently according to $\prod_{i=1}^n p(u_{j,i}|u_{A_j,i}(l_{\Gamma_{A_j}}))$. Let $u^n_S(l_{\Gamma_S})$ for $S\subseteq [1:\nu]$ denote $\{u_i^n(l_{\Gamma_i}): i\in S\}$.

\subsubsection{Operation at node $k\in [1:N]$}
After receiving $Y_k^n$, node $k$ finds the smallest $\hat{l}_{\lin{D}_k,k}$ such that 
\begin{align}
(u_{ \uin{D}_k\cup  \uin{B}_k}^n(\hat{l}_{\lin{D}_k,k},l_{\lin{B}_k}), y_k^n)\in \mathcal{T}_{\epsilon_k}^{(n)} \label{eqn:packing_typicality}
\end{align}
for some $l_{\lin{B}_k}$. If there is no such index vector, let $\hat{l}_{\lin{D}_k,k}=\mathbf{1}$. 

Next, node $k$ finds the smallest $l_{\lin{W}_k}$ such that\footnote{In (\ref{eqn:covering_typicality}), $(\hat{l}_{\Gamma_{ \uin{W}_k}\setminus \lin{W}_k,k}, l_{\lin{W}_k})$ suffices to specify the index set of $u_{ \uin{W}_k}^n$, but we write $(\hat{l}_{\lin{D}_k,k},  l_{\lin{W}_k})$ as the index set of $u_{ \uin{W}_k}^n$ for notational convenience. } 
\begin{align}
(u_{ \uin{D}_k}^n(\hat{l}_{\lin{D}_k,k}), u_{ \uin{W}_k}^n(\hat{l}_{\lin{D}_k,k}, l_{\lin{W}_k}), y_k^n)\in \mathcal{T}_{\epsilon_k'}^{(n)}.  \label{eqn:covering_typicality}
\end{align}
If there is no such index vector, let $l_{\lin{W}_k}=\mathbf{1}$. 
Send $x_{k,i}=x_k(u_{ \uin{D}_k,i}(\hat{l}_{\lin{D}_k,k}), u_{ \uin{W}_k,i}(\hat{l}_{\lin{D}_k,k}, l_{\lin{W}_k}), y_{k,i})$ for $i\in [1:n]$.

\subsubsection{Error analysis}
For $k\in [1:N]$, let $\hat{L}_{\lin{D}_k,k}$ and  $L_{\lin{W}_k}$ denote the chosen index vectors at node $k$. 
 Let us define the error event as follows: 
\begin{align*}
\mathcal{E}=\bigcup_{k=1}^{N}(\mathcal{E}_{k,1}\cup \mathcal{E}_{k,2} \cup \mathcal{E}_{k,3}\cup \mathcal{E}_{k,4})
\end{align*}
where 
\begin{align*}
\mathcal{E}_{k,1} &= \{(U^n_{ \uin{W}^{k-1}}(L_{\lin{W}^{k-1}}), Y_{[1:k]}^n)\notin \mathcal{T}_{\epsilon_k}^{(n)}\}\cr
\mathcal{E}_{k,2} &= \{(U_{ \uin{D}_k\cup  \uin{B}_k}^n(l_{\lin{D}_k\cup \lin{B}_k}), Y_k^n)\in \mathcal{T}_{\epsilon_k}^{(n)} \mbox{ for some } l_{\lin{D}_k}\neq L_{\lin{D}_k}, l_{\lin{B}_k}\} \cr 
\mathcal{E}_{k,3} &= \{(U_{ \uin{D}_k}^n(\hat{L}_{\lin{D}_k,k}), U_{ \uin{W}_k}^n(\hat{L}_{\lin{D}_k,k}, L_{\lin{W}_k}), Y_k^n)\notin \mathcal{T}_{\epsilon_k'}^{(n)} \}\cr
\mathcal{E}_{k,4} &= \{(U^n_{ \uin{W}^k}(L_{\lin{W}^k}), Y_{[1:k]}^n)\notin \mathcal{T}_{\epsilon_k''}^{(n)}\}.
\end{align*}
Note that $\mathcal{E}^c$ implies  $\hat{L}_{\lin{D}_k,k}=L_{\lin{D}_k}$ for all $k\in [1:N]$ and $(U^n_{ \uin{W}^N}(L_{\lin{W}^N}), Y_{[1:N]}^n)\in \mathcal{T}_{\epsilon}^{(n)}$, which means $(X_{[1:N]}^n, Y_{[1:N]}^n)\in \mathcal{T}_{\epsilon}^{(n)}$. Hence, $P_e^{(n)}(\epsilon)\leq P(\mathcal{E})$. 

The probability of the error event can be upper-bounded as follows: 
\begin{align}
P(\mathcal{E})&\leq \sum_{k=1}^N(P(\mathcal{E}_{k,1}\cap \bigcap_{j=1}^{k-1}(\mathcal{E}_{j,1}\cup \mathcal{E}_{j,2})^c \cap \mathcal{E}_{k-1,4}^c )+P(\mathcal{E}_{k,2}\cap \bigcap_{j=1}^{k-1}(\mathcal{E}_{j,1}\cup \mathcal{E}_{j,2}\cup \mathcal{E}_{j,3})^c)  \cr 
&~~~~~~~~+P(\mathcal{E}_{k,3}\cap \mathcal{E}_{k,1}^c \cap \mathcal{E}_{k,2}^c )+P(\mathcal{E}_{k,4}\cap \mathcal{E}_{k,1}^c \cap \mathcal{E}_{k,2}^c )). \label{eqn:error_analysis_ub}
\end{align}
Note that $(\mathcal{E}_{k,1}\cup \mathcal{E}_{k,2})^c$ implies $\hat{L}_{\lin{D}_k,k}=L_{\lin{D}_k}$.

Let us bound each term in the summation in (\ref{eqn:error_analysis_ub}) for given $k\in [1:N]$. First, we have 
\begin{align*}
&P(\mathcal{E}_{k,1}\cap \bigcap_{j=1}^{k-1}(\mathcal{E}_{j,1}\cup \mathcal{E}_{j,2})^c \cap \mathcal{E}_{k-1,4}^c )\cr
&\leq P((U^n_{ \uin{W}^{k-1}}(L_{\lin{W}^{k-1}}), Y_{[1:k-1]}^n)\in \mathcal{T}_{\epsilon_{k-1}''}^{(n)}, (U^n_{ \uin{W}^{k-1}}(L_{\lin{W}^{k-1}}), Y_{[1:k-1]}^n, Y_k^n)\notin \mathcal{T}_{\epsilon_{k}}^{(n)},\cr
&~~~~~~~~~~~ \hat{L}_{\lin{D}_j,j}=L_{\lin{D}_j} \mbox{ for all } j\in [1:k-1]), 
\end{align*}
which tends to zero as $n$ tends to infinity from the conditional typicality lemma \cite{ElGamalKim:11}.

Next, we show in Appendix \ref{appendix:sec_error} that the second term in the summation in (\ref{eqn:error_analysis_ub}) tends to zero as $n$ tends to infinity if
\begin{align*}
\sum_{j\in \lset{S}_k}r_j&<\sum_{j\in   \uset{S}_k} I(U_j;U_{  \uset{S}_k[j]\cup   \uset{S}_k^{c}}, Y_k|U_{A_j})-(1+\nu)\delta(\epsilon_k)
\end{align*}
for all $\lset{S}_k\subseteq \lin{D}_k\cup \lin{B}_k$ such that $\lset{S}_k\cap \lin{D}_k\neq \emptyset$, where $  \uset{S}_k$ is defined in (\ref{eqn:barSk}).

The third term in the summation in (\ref{eqn:error_analysis_ub}) is shown in Appendix \ref{appendix:third_error} to tend to zero as $n$ tends to infinity if
\begin{align}
\sum_{j\in \lset{T}_k}r_j &>\sum_{j\in   \uset{T}_k}I(U_j;U_{  \uset{T}_k[j]\cup  \uin{D}_k},Y_k|U_{A_j})+4(1+\nu)\delta(\epsilon_k') \label{eqn:tkcondition}
\end{align}
for all  $\lset{T}_k\subseteq \lin{W}_k $ such that $\lset{T}_k\neq \emptyset$, where $  \uset{T}_k$ is defined in (\ref{eqn:barTk}).

Finally, the fourth  term in the summation in (\ref{eqn:error_analysis_ub}) is proved in Appendix \ref{appendix:fourth_error} to tend to zero as $n$ tends to infinity for sufficiently small $\epsilon_k$ and $\epsilon_k'$ under the aforementioned condition (\ref{eqn:tkcondition}) for all  $\lset{T}_k\subseteq \lin{W}_k $ such that $\lset{T}_k\neq \emptyset$. 
  
Therefore, $P(\mathcal{E})$ and thus $P_e^{(n)}(\epsilon)$ tend to zero as $n$ tends to infinity if rate tuple $(r_1,\ldots,r_\mu)$ satisfies for $1\leq k \leq N$,
\begin{align*}
\sum_{j\in \lset{S}_k} r_j&<\sum_{j\in   \uset{S}_k} I(U_j;U_{  \uset{S}_k[j]\cup   \uset{S}_k^c},Y_k|U_{A_j})\\ 
\sum_{j\in \lset{T}_k} r_j&>\sum_{j\in   \uset{T}_k} I(U_j;U_{  \uset{T}_k[j]\cup  \uin{D}_k},Y_k|U_{A_j})
\end{align*}
for all $\lset{S}_k\subseteq \lin{D}_k\cup \lin{B}_k$ such that $\lset{S}_k\cap \lin{D}_k\neq \emptyset$ and for all $\lset{T}_k\subseteq \lin{W}_k$ such that $\lset{T}_k\neq \emptyset$. This completes the proof.
\end{proof}
Let $\Omega'$  denote the set of all possible $\omega$'s that satisfy additional conditions $\nu=\mu$ and $\Gamma_j=\{j\}\cup A_j$. In many cases, it is sufficient to consider $\Omega'$.

\begin{corollary} \label{corollary:main}
For an  $N$-node ADMN, $p^*$ is achievable if there exists $\omega'\in \Omega'$ such that for $1\leq k \leq N$ 
\begin{align*}
\sum_{j\in S_k} r_j&<\sum_{j\in S_k} I(U_j;U_{S_k[j]\cup S_k^c},Y_k|U_{A_j})\\ 
\sum_{j\in T_k} r_j&>\sum_{j\in T_k} I(U_j;U_{T_k[j]\cup D_k},Y_k|U_{A_j})
\end{align*}
for all $S_k\subseteq D_k\cup B_k$ such that $S_k\cap D_k \neq \emptyset$ and if $j\in S_k^c$, then $A_j \subseteq S_k^c$ and for all $T_k\subseteq W_k$ such that $T_k\neq \emptyset$ and if $j\in T_k$, then $A_j\cap W_k \subseteq T_k$. 
\end{corollary}
\begin{proof}
For $\omega'\in \Omega'$, we have $\lin{W}_k=\uin{W}_k, \lin{D}_k=\uin{D}_k, \lin{B}_k=\uin{B}_k$, $\lset{S}_k\subseteq   \uset{S}_k$, and $  \uset{T}_k\subseteq \lset{T}_k$ for all $\lset{S}_k\subseteq \lset{D}_k\cup \lset{B}_k$ such that $\lset{S}_k\cap \lset{D}_k \neq \emptyset$ and for all $\lset{T}_k\subseteq \lin{W}_k$ such that $\lset{T}_k\neq \emptyset$.  Hence, it is enough to consider $\lset{S}_k$ and $\lset{T}_k$ such that $\lset{S}_k=  \uset{S}_k$ and $  \lset{T}_k= \uset{T}_k$. 
\end{proof}

Now, let us show various examples to illustrate how to utilize the unified framework and unified coding theorem.
Throughout this paper, unspecified components $W_k, D_k, B_k$, and $A_j$ of $\omega\in \Omega$ or $\omega'\in \Omega'$ are assumed to be empty\footnote{For $\omega'\in \Omega'$, we do not explicitly specify $\Gamma_j$ since $\Gamma_j=\{j\}\cup A_j$.}.  

\subsection{Multicoding and binning}
Our scheme does not include multicoding and binning explicitly, but our scheme is general enough to emulate those. In the following two examples, we provide guidelines for choosing $\omega$ that correpond multicoding and binning operations. 
\begin{example}[Gelfand-Pinsker coding  \cite{GelfandPinsker:80}]
Consider channels with noncausal states represented by a two-node ADMN such that $Y_1=(M,S)$,  $p(y_1)=p(m)p(s)$ where $H(M)=R$, and $p(y_2|y_1,x_1)=p(y_2|x_1,s)$, as discussed in Section \ref{sec:framework}. We choose  $\omega'\in \Omega'$ in Corollary \ref{corollary:main} as $\mu=1, U_1=(M,U), W_1=\{1\}, D_2=\{1\}$, $p(u|y_1)=p(u|s)$, $x_1(u_1,y_1)=x_1(u,s)$, and $x_2(u_1,y_2)=m$. Then, from Corollary \ref{corollary:main}, we have the condition $r_1>R+I(U;S)$ for compression at node 1 and the condition $r_1<I(U;Y_2)$ for decoding at node 2. By the Fourier-Motzkin (F-M) elimination, we obtain $R<I(U;Y_2)-I(U;S)$. 

To see the relevance to the multicoding operation, let us assume that $2^R$ is an integer and $M\sim \mbox{Unif}[1:2^R]$. In this case, there are roughly $2^{nR}$ sets of $u_1^n$ sequences where each set consists of roughly $2^{n(r_1-R)}$ $u_1^n$ sequences having the same $m^n$ sequence. Note that $r_1-R>I(U;S)$. Then, when $y_1^n=(\tilde{m}^n,\tilde{s}^n)$ is received at node 1, among  $2^{n(r_1-R)}$ $u_1^n$'s having $\tilde{m}^n$, we can find a $\tilde{u}_1^n=(\tilde{u}^n, \tilde{m}^n)$ with high probability such that $\tilde{u}^n$ is jointly typical with $\tilde{s}^n$ with respect to $p(u,s)$. This corresponds to the multicoding operation in a sense that for each message $m^n$, multiple $u^n$'s are matched to satisfy the joint typicality with respect to $p(u,s)$.  
\end{example}

\begin{example}[Wyner-Ziv coding \cite{WynerZiv:76}]
Consider lossy source coding with side information represented by a two-node ADMN such that $Y_2=(Y_2',T)$ and $p(y_2|y_1,x_1)=p(y_2'|x_1)p(t|y_1)$ where $\max_{p(x_1)}I(X_1;Y_2')$ $=R$, and target distribution  $p^*$ such that $p^*(x_1)=\arg \max I(X_1;Y_2')$ and $\E[d(Y_1,X_2)]\leq \frac{d}{1+\epsilon}$ for $\epsilon\rightarrow 0$, as discussed in Section \ref{sec:framework}.  We choose  $\omega'\in \Omega'$ in Corollary \ref{corollary:main} as $\mu=1$, $W_1=\{1\}$, $D_2=\{1\}$,  $U_1=(U, X_1)$, $p(u_1|y_1)=p(u|y_1)p(x_1)$, and $x_2(u_1,y_2)=x_2(u,t)$ such that $p=p^*$. Then, from Corollary \ref{corollary:main}, we need the condition $r_1>I(U;Y_1)$ for compression at node 1 and the condition $r_1<R+I(U;T)$ for decoding at node 2. By the F-M elimination, we get  $R>I(U;Y_1)-I(U;T)=I(U;Y_1|T)$. 

To see the relevance to the binning operation, let us assume that $2^R$ is an integer, $|\mathcal{X}_1|=2^R$, and $Y_2'=X_1$.  In this case, there are roughly $2^{nR}$ sets of $u_1^n$ sequences where each set consists of roughly $2^{n(r_1-R)}$ $u_1^n$ sequences having the same $x_1^n$ sequence. Note that $r_1-R<I(U;T)$. Then, when $y_2^n=(\tilde{x}_1^n,\tilde{t}^n)$ is received at node 2, among $2^{n(r_1-R)}$ $u_1^n$ sequences having $\tilde{x}_1^n$, there is a unique $\tilde{u}_1^n=(\tilde{u}^n, \tilde{x}_1^n)$ with high probability such that $\tilde{u}^n$ is jointly typical with $\tilde{t}^n$ with respect to $p(u,t)$. This corresponds to the binning operation in a sense that multiple $u^n$ sequences are matched to the same bin index $x_1^n$ but node 2 can decode  $u^n$ using the joint typicality with respect to $p(u,t)$.

\end{example}

\subsection{Rate-splitting}
The following example shows how the rate-splitting can be incorporated. 
\begin{example}[Han-Kobayashi coding  \cite{HanKobayashi:81}]
To perform the rate-splitting, we represent the interference channel  as a four-node ADMN such that $Y_k=(M_{k0},M_{kk})$, $H(M_{k0})=R_{k0}$, $H(M_{kk})=R_{kk}, R_k=R_{k0}+R_{kk}$  for $k\in[1:2]$, $p(y_1)=p(m_{10})p(m_{11})$, $p(y_2|y_1,x_1)=p(m_{20})p(m_{22})$, $p(y_3|y_{[1:2]},x_{[1:2]})=p(y_3|x_{[1:2]})$, and $p(y_4|y_{[1:3]}, x_{[1:3]})=p(y_4|y_3,x_{[1:2]})$, and target distribution $p^*$ such that $X_3=Y_1, X_4=Y_2$. Here, $M_{k0}^n$ and $M_{kk}^n$ for $k\in[1:2]$ correspond to the rate-splitted messages at node $k$.

We choose $\omega'\in \Omega'$ in Corollary \ref{corollary:main} as follows:  
$\mu=4$, $U_1=(M_{10}, V_1), U_2=(M_{11}, X_1), U_3=(M_{20}, V_2), U_4=(M_{22}, X_2)$, $W_1=\{1,2\}, W_2=\{3,4\}, D_3=\{1,2\}, D_4=\{3,4\}, B_3=\{3\}, B_4=\{1\}, A_2=\{1\}, A_4=\{3\}$, $p(v_k,x_k|y_k)=p(v_k,x_k)$ for $k\in [1:2]$. Then, from Corollary \ref{corollary:main} followed by the F-M elimination, we get the Han-Kobayashi inner bound.
\end{example}

\subsection{Introduction of a virtual node}
Let us show an example where a simpler rate expression than previously known result can be obtained by using Proposition \ref{lemma:virtual_node}. 
\begin{example}[Interference decoding for a 3-SD-pair deterministic interference channel \cite{BandemerElGamal:11}]
In the 3-SD-pair deterministic channel \cite{BandemerElGamal:11}, source $k\in [1:3]$ encodes a message $I_k$ of rate $R_k\geq 0$ to channel input sequence $X_k^n$ and destination $k$ estimates $I_k$ as $\hat{I}_k$ from its channel output sequence $Z_k^n$, where $n$ denotes the number of channel uses. The channel output at destination $k\in[1:3]$ is given as $Z_k=f_k(X_{k,k},V_k)$ for some function $f_k$, where $X_{i,j}=g_{i,j}(X_i)$ for some function $g_{i,j}$ for $i\in [1:3]$ and $j\in [1:3]$, $V_1=h_1(X_{2,1},X_{3,1})$, $V_2=h_2(X_{1,2},X_{3,2})$, and $V_3=h_3(X_{1,3},X_{2,3})$ for some functions $h_1$, $h_2$, and $h_3$. $h_k$ and $f_k$ for $k\in[1:3]$ are assumed to be injective in their arguments. The probability or error and achievability of a rate triple $(R_1, R_2, R_3)$ are defined in the standard way. In Fig. \ref{fig:inter_dec}, the 3-SD-pair deterministic channel is illustrated for destination 1.  
\begin{figure}[t]
 \centering
  {
   \includegraphics[width=80mm]{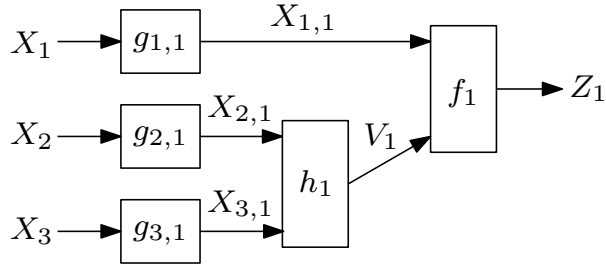}}
  \caption{The 3-SD-pair deterministic interference channel considered in \cite{BandemerElGamal:11}} \label{fig:inter_dec}
\end{figure}
By using a new technique of decoding the combined interference, Bandemer and El Gamal showed in  \cite{BandemerElGamal:11} the following achievable rate region.\footnote{For simplicity, we present the rate region without coded time sharing.} 
\begin{proposition}[Interference decoding inner bound \cite{BandemerElGamal:11}] \label{thm:inter_dec} The rate region $\bigcap_{k=1}^3 \mathcal{R}_k(X_1,X_2,X_3)$ is achievable for some $(X_1,X_2,X_3)\sim p(x_1)p(x_2)p(x_3)$, where $\mathcal{R}_1(X_1,X_2,X_3)$ is the set of rate triples $(R_1,R_2,R_3)$ such that 
\begin{align}
R_1&<H(X_{1,1})\label{eqn:id1}\\
R_1+\min(R_2,H(X_{2,1}))&<H(Z_1|X_{3,1})\label{eqn:id2}\\
R_1+\min(R_3,H(X_{3,1}))&<H(Z_1|X_{2,1})\label{eqn:id3}\\
R_1+\min(R_2+R_3,R_2+H(X_{3,1}),H(X_{2,1})+R_3,H(V_1))&<H(Z_1),\label{eqn:id4}
\end{align}
and $\mathcal{R}_2(X_1,X_2,X_3)$ and $\mathcal{R}_3(X_1,X_2,X_3)$ are defined similarly by replacing the subscripts as $1\mapsto 2 \mapsto 3 \mapsto 1$ and $1 \mapsto 3 \mapsto 2 \mapsto 1$  in $\mathcal{R}_1(X_1,X_2,X_3)$, respectively. 
\end{proposition}

Now, let us show that Corollary \ref{corollary:main} can recover the  interference decoding inner bound by applying Proposition \ref{lemma:virtual_node}. By introducing a virtual node from Proposition \ref{lemma:virtual_node}, the 3-SD-pair deterministic channel can be represented by the following ADMN and target distribution: 
\begin{itemize}
\item ADMN: $N=7$, $H(Y_k)=R_k$ for $k\in [1:3]$, $p(y_2|x_1,y_1)=p(y_2)$, $p(y_3|x_{[1:2]},y_{[1:2]})=p(y_3)$, $Y_4=(X_{1,2}, X_{1,3}, X_{2,1}, X_{2,3}, X_{3,1}, X_{3,2}, V_1, V_2, V_3)$, $X_{4}=\emptyset$, $Y_{5}=Z_1$, $Y_{6}=Z_2$, $Y_{7}=Z_3$.
\item Target distribution: $p^*$ such that $X_{5}=Y_1, X_{6}=Y_2, X_{7}=Y_3$.
\end{itemize}

For this ADMN and $p^*$, let us choose $\omega'\in \Omega'$ for Corollary \ref{corollary:main}. We let $\mu= 12$,  $W_k=\{k\}$ for $k\in [1:3]$, $W_4=[4:12]$, $U_j=(Y_j,X_j)$ and $p(x_j|y_j)=p(x_j)$ for $j\in [1:3]$, $U_4=X_{1,2}, U_5=X_{1,3}, U_6=X_{2,1}, U_7=X_{2,3}, U_8=X_{3,1}, U_9=X_{3,2}, U_{10}=V_1, U_{11}=V_2, U_{12}=V_3$, $D_{5}=\{1\}$, $D_{6}=\{2\}$, $D_{7}=\{3\}$. We let $B_{5}$ as follows: 
\begin{align*}
B_5=\begin{cases}
\{10\} &\mbox{ if } m_1=H(V_1)\\
\{2,3\} &\mbox{ if } R_2<H(X_{2,1}), R_3<H(X_{3,1}), m_1<H(V_1)\\
\{2,8\} &\mbox{ if } R_2<H(X_{2,1}), R_3\geq H(X_{3,1}), m_1<H(V_1) \\
\{3,6\} &\mbox{ if } R_2\geq H(X_{2,1}), R_3<H(X_{3,1}), m_1<H(V_1)\\
\end{cases},
\end{align*} 
where $m_1=\min(R_2+R_3, R_2+H(X_{3,1}), H(X_{2,1})+R_3, H(V_1))$. We choose $B_6$ and $B_7$ similarly. Then, by applying Corollary \ref{corollary:main}, we can obtain an inner bound that is at least as good as that in Proposition \ref{thm:inter_dec} and has a simpler form\footnote{When $m_1=H(V_1)$, our bound for the decoding at the first destination is given as $R_1<H(X_{1,1})$ and $R_1+H(V_1)<H(Z_1)$, while that in Proposition \ref{thm:inter_dec} has two additional inequalities (\ref{eqn:id2}) and (\ref{eqn:id3}).}.  Furthermore, the interference decoding inner bound in  \cite{BandemerElGamal:11} was improved in \cite{BandemerElGamal:Allerton11} by incorporating rate splitting, Marton coding, and superposition coding. We can choose $\omega'\in \Omega'$ that includes such coding techniques and obtain an inner bound from Corollary \ref{corollary:main} that includes that in \cite{BandemerElGamal:Allerton11} and has a simpler form. 
\end{example}

The following example illustrates the usefulness of Proposition \ref{lemma:common_part} in networks with correlated sources. 
\begin{example}[Lossless communication of two correlated sources over a multiple access channel \cite{CoverElGamalSalehi:80}]
By using Proposition \ref{lemma:common_part}, we can represent the problem of sending two correlated sources over a multiple access channel as the following ADMN and target distribution: 
\begin{itemize}
\item ADMN: $N=4$, $Y_1=V_1$, $Y_2=(V_2,X_1)$, $Y_3=(V_3,X_1)$, and $p(y_4|y_{[1:3]},x_{[1:3]})=p(y_4|x_{[2:3]})$ where $V_2$ and $V_3$ are two discrete memoryless sources, $V_1$ is the common part of $V_2$ and $V_3$, and $|\mathcal{X}_1|$ is arbitrarily large.
\item Target distribution: $p^*$ such that $X_4=(V_2,V_3)$.
\end{itemize}
We choose $\omega'\in \Omega'$ in Corollary \ref{corollary:main} as follows: $\mu=3$, $U_1=(V_1,U,X_1)$, $U_2=(V_2,X_2)$, $U_3=(V_3,X_3)$, $W_k=\{k\}$ for $k\in [1:3]$, $D_2=\{1\}$, $D_3=\{1\}$, $D_4=\{1,2,3\}$, $A_2=\{1\}$, $A_3=\{1\}$, $p(u,x_1|y_1)=p(u)/|\mathcal{X}_1|$, $p(x_2|y_2,u_1)=p(x_2|v_2,u)$, $p(x_3|y_3,u_1)=p(x_3|v_3,u)$.
By Corollary \ref{corollary:main} followed by the F-M elimination, the sufficient condition for lossless communication of two correlated sources over a multiple access channel \cite{CoverElGamalSalehi:80} is recovered. 
\end{example} 

\subsection{Application to DMNs}
Note that any strictly causal or causal network with blockwise operations can be represented as an ADMN by unfolding the network as illustrated in Section \ref{sec:framework}. The following lemma is useful when we apply Theorem \ref{thm:main} to unfolded networks.
\begin{lemma} \label{corollary:GDCF_reduced}
Consider $\omega\in \Omega$. For $\lset{S}_k\subseteq \lin{D}_k\cup \lin{B}_k$ such that $\lset{S}_k\cap \lin{D}_k\neq \emptyset$ and $\lset{T}_k\subseteq \lin{W}_k$ such that $\lset{T}_k\neq \emptyset$, the decoding and compression bounds, i.e., (\ref{eqn:main_sk}) and (\ref{eqn:main_tk}), in Theorem \ref{thm:main} are satisfied if 
\begin{align}
\sum_{j\in \bar{S}_k}r_j<\sum_{j\in S_k'}I(U_j;U_{S_k'[j]\cup S_k'^c},Y_k|U_{A_j})\\
\sum_{j\in \bar{T}_k}r_j>\sum_{j\in T_k'}I(U_j;U_{T_k'[j]\cup D_k},Y_k|U_{A_j}),
\end{align}
for some $S_k'\subseteq S_k$ such that $A_j\subseteq (S_k\setminus S_k')[j]\cup S_k^c$ for all $j\in S_k\setminus S_k'$ and for some $T_k'$ such that $T_k\subseteq T_k'$.
\end{lemma}
\begin{proof}
The compression part is straightforward. For the decoding part, we have
\begin{align*}
&\sum_{j\in S_k'}I(U_j;U_{S_k'[j]\cup S_k'^c},Y_k|U_{A_j})\\
&\overset{(a)}{\leq}\sum_{j\in S_k'}I(U_j;U_{S_k'[j]\cup S_k'^c},Y_k|U_{A_j})+\sum_{j\in S_k\setminus S_k'}(H(U_j|U_{A_j})-H(U_j|U_{(S_k\setminus S_k')[j]}, U_{S_k^c}, Y_k))\\
&=\sum_{j\in S_k} H(U_j|U_{A_j})-H(U_{S_k'}|U_{S_k'^c},Y_k)-H(U_{S_k\setminus S_k'}|U_{S_k^c},Y_k)\\
&=\sum_{j\in S_k} H(U_j|U_{A_j})-H(U_{S_k}|U_{S_k^c},Y_k)\\
&=\sum_{j\in S_k} (H(U_j|U_{A_j})-H(U_j|U_{S_k[j]},U_{S_k^c},Y_k))\\
&=\sum_{j\in S_k}I(U_j;U_{S_k[j]\cup S_k^c},Y_k|U_{A_j}),
\end{align*}
where $(a)$ is from the condition for $S_k'$ stated in the lemma. 
\end{proof}
 
In the following example, we show an achievable rate for a DMN by applying Theorem \ref{thm:main} to the unfolded network. 

\begin{example}[Noisy network coding  \cite{Lim:10}]
Consider a single-source multicast DMN $(\mathcal{X}_1,\ldots,\mathcal{X}_N, \mathcal{Y}_1$, $\ldots,\mathcal{Y}_N, p(y_{[1:N]}|x_{[1:N]}))$. Let node 1 denote the source node and let $\mathcal{D}\subseteq [2:N]$ denote the set of destination nodes. An $(R,n)$ code for the single-source multicast DMN consists of message $I$, uniformly distributed over $\mathcal{I}\triangleq[1:2^{nR}]$, encoding function at the source that maps $I\in \mathcal{I}$ to $x_1^n\in \mathcal{X}_1^n$,  processing function at node $k\in [2:N]$ at time $i\in [1:n]$ that maps $y_k^{i-1}\in \mathcal{Y}_k^{i-1}$ to $x_{k,i}\in \mathcal{X}_{k}$, and decoding function at destination $d\in \mathcal{D}$ that maps $y_d^n\in \mathcal{Y}_d^n$ to $\hat{I}_d\in \mathcal{I}$. 
The probability of error is defined as $P_e^{(n)}=P(\hat{I}_d \neq I \mbox{ for some } d\in \mathcal{D})$ and a rate $R$ is said to be achievable if there exists a sequence of $(R,n)$ codes such that $\lim_{n\rightarrow \infty} P_e^{(n)}=0$. 

For a single-source multicast DMN, noisy network coding (NNC) rate  \cite{Lim:10} is given as follows.
\begin{proposition}[Noisy network coding bound \cite{Lim:10}] \label{proposition:nnc}
For a single-source multicast DMN, a rate of $R$ is achievable if 
\begin{align*}
R<\min_{d\in \mathcal{D}}\min_{S\in [2:N]\setminus \{d\}} I(X_1,X_S;\hat{Y}_{S^c},Y_d|X_{S^c})-I(Y_S;\hat{Y}_S|X^N,\hat{Y}_{S^c}, Y_d)
\end{align*}
for some $p_{X_1}\prod_{k=2}^Np_{X_k}p_{\hat{Y}_k|X_k,Y_k}$. 
\end{proposition}

Now, let us obtain the NNC rate from Theorem \ref{thm:main}.  Fix  $p_{X_1}\prod_{k=2}^Np_{X_k}p_{\hat{Y}_k|X_k,Y_k}$. 
Achievability uses $B$ transmission blocks, each consisting of $n$ channel uses. Let $Y_{k,b}^n$ and $X_{k,b}^n$ for $k\in [1:N]$ and $b\in [1:B]$ denote the channel output and channel input sequences, respectively, at node $k$ at block $b$. Let us assume the following blockwise operation at each node: at the end of block $b-1$, where $b\in [1:B]$, node $k\in [1:N]$ encodes what to transmit in block $b$, i.e., $X_{k,b}^n$, using previously received channel outputs up to block $b-1$, i.e., $Y_{k,[1:b-1]}^n$. Then, we can unfold the network. 

In the unfolded network, we have $(B+1)N$ nodes. The operation of node $(k,b), k\in [1:N], b\in [1:B]$ corresponds to that of node $k$ of the original network transmitting in block $b$ based on the received channel outputs up to block $b-1$ and the operation of node $(d,B+1), d\in \mathcal{D}$ corresponds to that of node $d$ of the original network that estimates the message based on the received channel outputs up to block $B$. Let $Y_{k,b}^{\mathrm{unf}}$ and $X_{k,b}^{\mathrm{unf}}$ denote the channel output and channel input at node $(k,b)$ of the unfolded network, respectively. A message generated at the source is regarded as the channel output at node $(1,1)$, i.e., $Y_{1,1}^{\mathrm{unf}}=M$ such that $H(M)=BR$, and the message estimate at destination $d\in \mathcal{D}$ is regarded as the channel input at node $(d,B+1)$, i.e., $\mathcal{X}_{d,B+1}^{\mathrm{unf}}=\mathcal{M}$.  To reflect the fact that node $(k,b+1)$ is originally the same node as node $(k,b)$, we assume that node $(k,b)$ has an orthogonal link of sufficiently large rate to node $(k,b+1)$. Hence, for $k\in [1:N]$ and $b\in [1:B]$, we let $Y_{k,b+1}^{\mathrm{unf}}=(X_{k,b}^{\mathrm{unf}}, Y_{k,b})$ and let 
\begin{align*}
p(y_{k,b}|y_{[1:N],[1:b]}^{\mathrm{unf}},y_{[1:k-1],b+1}^{\mathrm{unf}},x_{[1:N],[1:b]}^{\mathrm{unf}},x_{[1:k-1],b+1}^{\mathrm{unf}})=p_{Y_k|Y_{[1:k-1]},X_{[1:N]}}(y_{k,b}|y_{[1:k-1],b},x_{[1:N],b})
\end{align*}
where $x_{k,b}^{\mathrm{unf}}=(x_{k,b}, z_{k,b})$ for $x_{k,b}\in \mathcal{X}_k$ and $z_{k,b}\in \mathcal{Z}_{k,b}$ for some $\mathcal{Z}_{k,b}$ with arbitrarily large cardinality. 
We assume a target joint distribution $p^*(x_{[1:N],[1:B+1]}^{\mathrm{unf}},y_{[1:N],[1:B+1]}^{\mathrm{unf}})$ such that $X_{d,B+1}^{\mathrm{unf}}=Y_{1,1}^{\mathrm{unf}}$ for all $d\in \mathcal{D}$. Note that the achievability of $p^*$ for the unfolded network implies the achievability of rate $R$ for the original network.  

Now, let us choose $\omega\in \Omega$ to obtain the NNC rate. Let $\mu=BN-B+1$ and $\nu=2BN-2B-N+2$. Consider a $\mu$-rate tuple $(r_0, r_{k,b}: k\in [2:N], b\in [0:B-1])$. For notational convenience, let us index the codebook $\mathcal{C}$ by the auxiliary random variable used for its generation, i.e., if a codebook consists of $u^n(1),\ldots,u^n(2^{nr})$ generated conditionally independently according to $\prod_{i=1}^n p(u_i|v_i)$ for some $r\geq 0$ and $p(u|v)$, we denote the codebook by $\mathcal{C}_U$. In addition, we index the index set $\mathcal{L}$ in the following way: $\mathcal{L}_{l_0}=[1:2^{nr_0}]$ and $\mathcal{L}_{l_{k,b}}=[1:2^{nr_{k,b}}]$ for $k\in [2:N]$ and $b\in [0:B-1]$. The remaining coding parameters associated with each node are given as follows:

\begin{itemize}
\item Node $(1,1)$: 
\begin{align*} 
W_{1,1}=\{U_0\}, \Gamma_{U_0}= \{l_0\}, U_0=(Y_{1,1}^{\mathrm{unf}}, X_{1,1}, \ldots, X_{1,B}), p(x_{1,1}, \ldots, x_{1,B}| y_{1,1}^{\mathrm{unf}})=\prod_{b\in [1:B]} p_{X_1}(x_{1,b}) 
\end{align*}

\item Node $(k,1)$, $k\in [2:N]$: 
\begin{align*}
W_{k,1}=\{X_{k,1}\}, \Gamma_{X_{k,1}}=\{l_{k,0}\}, p(x_{k,1}|y_{k,1}^{\mathrm{unf}})=p_{X_k}(x_{k,1}) 
\end{align*}

\item Node $(1,b)$, $b\in [2:B]$:
\begin{align*}
D_{1,b}&=W_{1,1} 
\end{align*}

\item Node $(k,b)$, $k\in [2:N], b\in [2:B]$:
\begin{align*}
D_{k,b}=W_{k}^{b-1}, W_{k,b}=\{\hat{Y}_{k,b-1}, X_{k,b}\},  \Gamma_{\hat{Y}_{k,b-1}}&=\{l_{k,b-1},l_{k,b-2}\},\Gamma_{X_{k,b}}=\{l_{k,b-1}\},A_{\hat{Y}_{k,b-1}}=\{X_{k,b-1}\} \\
p(W_{k,b}|D_{k,b}, y_{k,b}^{\mathrm{unf}})&=p_{\hat{Y}_k|X_k,Y_k}(\hat{y}_{k,b-1}|x_{k,b-1},y_{k,b-1})p_{X_k}(x_{k,b}) 
\end{align*}

\item Node $(d,B+1)$, $d\in \mathcal{D}$:
\begin{align*}
D_{d,B+1}=W_{d}^{B}\cup \{U_0\}, B_{d,B+1}= \{ X_{k,b}, \hat{Y}_{k,b}, k\in [2:N]\setminus \{d\}, b\in [1:B-1] \} 
\end{align*}

\end{itemize}
For $k\in [1:N]$ and $b\in [1:B]$, we let $X_{k,b}^{\mathrm{unf}}=(Y_{k,[1:b-1]}, W_{k,b}, D_{k,b})$. 
For $d\in \mathcal{D}$, let $X_{d,B+1}^{\mathrm{unf}}=U_0$. Note that the above choice of coding parameters shows the following blockwise i.i.d. property:
 \begin{align}
& p(x_{[1:N],[1:B-1]},y_{[2:N],[1:B-1]},\hat{y}_{[2:N],[1:B-1]}) \cr
&=\!\!\prod_{b\in [1:B-1]}\! \!\big( p_{X_1}(x_{1,b})\!\!\prod_{k\in [2:N]} p_{X_k}(x_{k,b})p_{\hat{Y}_k|X_k,Y_k}(\hat{y}_{k,b}|x_{k,b},y_{k,b}) p_{Y_{[2:N]}|X_{[1:N]}}(y_{[2:N],b}|x_{[1:N],b})\big). \label{eqn:GDCF_independence}
\end{align}

Now, we are ready to apply Theorem \ref{thm:main} to obtain the NNC rate. For each node in the unfolded network, the decoding and compression bounds i.e., (\ref{eqn:main_sk}) and (\ref{eqn:main_tk}), respectively, are given as follows.

\begin{itemize}
\item Compression at node (1,1): The bound for compression is given as follows: 
\begin{align}
r_0&>BR.\label{eqn:nnc1}
\end{align} 

\item Compression at node $(k,1), k\in [2:N]$: It can be easily shown that the bound for compression is inactive.  

\item Decoding at node $(1,b), b\in [2:B]$: Since $Y_{1,b}^{\mathrm{unf}}=D_{1,b}$, the bound for decoding becomes inactive.

\item Decoding and compression at node $(k,b), k\in [2:N], b\in [2:B]$: Since $Y_{k,b}^{\mathrm{unf}}$ contains $D_{k,b}$, the bound for decoding becomes inactive.  From the blockwise i.i.d. property (\ref{eqn:GDCF_independence}), the bound for compression is given as 
\begin{align}
r_{k,b-1}>I(\hat{Y}_k;Y_k|X_k). \label{eqn:GDCF_C}
\end{align}

\item Decoding  at node $(d,B+1), d\in \mathcal{D}$: Since  $D_{d,B+1}=W_{d}^B \cup \{U_0\}$ and $Y_{d,B+1}^{\mathrm{unf}}$ contains $W_{d}^B$, we only need to consider $\bar{S}_{d,B+1}\subseteq \{l_0, l_{k,b}: k\in [2:N]\setminus \{d\}, b\in [0:B-1]\}$ such that $l_0\in \bar{S}_{d,B+1}$. Note that such $\bar{S}_{d,B+1}$ can be represented as $\{l_0\}\cup \bigcup_{b\in [0:B-1]}\{l_{k,b}: k\in S_b\}$ for some 
$S_b\subseteq [2:N]\setminus \{d\}$ for $b\in [0:B-1]$. Then, from Lemma \ref{corollary:GDCF_reduced} and using the blockwise i.i.d. property shown in (\ref{eqn:GDCF_independence}), the bound for decoding is given as  
\begin{align}
&r_0+\sum_{b\in[0:B-1]}\sum_{k \in S_b}r_{k,b} < \sum_{b\in [1:B-1]}\Big(I(X_{1}; \hat{Y}_{S_{b-1}^c}, X_{S_{b-1}^c}, Y_{d})\cr
&+\sum_{k\in S_{b-1}} I(X_{k};X_{S_{b-1}[k]}, X_{1}, \hat{Y}_{S_{b-1}^c}, X_{S_{b-1}^c}, Y_{d}) +\sum_{k\in S_{b-1}} I(\hat{Y}_{k};\hat{Y}_{S_{b-1}[k]}, \hat{Y}_{S_{b-1}^c}, X^N,  Y_{d}|X_{k}) \Big)\cr
&<\sum_{b\in [1:B-1]}\Big(I(X_{1}, X_{S_{b-1}}; \hat{Y}_{S_{b-1}^c},  Y_{d}|X_{S_{b-1}^c})+\sum_{k\in S_{b-1}} I(\hat{Y}_{k};\hat{Y}_{S_{b-1}[k]}, \hat{Y}_{S_{b-1}^c}, X^N,  Y_{d}|X_{k}) \Big) \label{eqn:GDCF_packing1}
\end{align}
for all $S_b\subseteq [2:N]\setminus \{d\}$ for $b\in [0:B-1]$. \end{itemize}

Now, let $r_{k,b}=r_k, k\in [2:N], b\in [0:B-1]$ for some $r_k\geq 0$.
Then, (\ref{eqn:GDCF_C}) is satisfied if 
\begin{align}
r_{k}&>I(\hat{Y}_k;Y_k|X_k) \label{eqn:nnc2}
\end{align}
for $k\in [2:N]$, and (\ref{eqn:GDCF_packing1}) is satisfied if 
\begin{align}
r_0&<(B-1)\!\min_{S: S\subseteq [2:N]\setminus \{d\}} \Big(\!
I(X_1, X_S;\hat{Y}_{S^c},  Y_d|X_{S^c})+\sum_{k\in S} I(\hat{Y}_k;\hat{Y}_{S[k]}, \hat{Y}_{S^c},X^N,   Y_d|X_k) -\sum_{k \in S}r_{k}\Big)\cr 
&~~~~~~~~~~~~~~~~~~~~~~~~~~~~~~~~~~~~~~~~~~~~~~~~~~~~~~~~~~~~~~~~~-\sum_{k\in [2:N]} r_k \label{eqn:nnc3}
\end{align}
for all $d\in \mathcal{D}$.

By performing F-M elimination to (\ref{eqn:nnc1}), (\ref{eqn:nnc2}), and (\ref{eqn:nnc3}) and by taking $B\rightarrow \infty$, the NNC rate in Proposition \ref{proposition:nnc} is obtained. 
\end{example}
\subsection{Wiretap channel}
We note that a secrecy constraint cannot be imposed by using our definition of achievability based on joint typicality. Nevertheless, we show that our unified coding scheme can be specialized to the scheme that achieves the secrecy capacity of the wiretap channel. 
\begin{example}[Wiretap channel \cite{CsiszarKorner:78}]
Consider a three-node ADMN such that $Y_1=(M,M_1,M_2)$, $H(M)=R, H(M_1)=R_1, H(M_2)=R_2$, $p(y_1)=p(m)p(m_1)p(m_2)$, $p(y_2|y_1,x_1)=p(y_2|x_1)$, and $p(y_3|y_{[1:2]},x_{[1:2]})=p(y_3|y_2,x_1)$ and target distribution $p^*$ such that $X_2=M$. Here, nodes 1, 2, and 3 correspond to a source, legitimate destination, and wiretapper, respectively, in the wiretap channel. $M^n$ corresponds to the message and $M_1^n$ and $M_2^n$ play the roles of fictitious messages to confuse the wiretapper. We note that the achievability of $p^*$ implies the reliable communication to the legitimate destination, but does not guarantee the security.

If we choose $\omega'\in \Omega'$ in Corollary \ref{corollary:main} as $\mu=2$, $U_1=(M,M_1,U)$, $U_2=(M_2,X_1)$, $W_1=\{1,2\}$, $D_2=\{1\}$, $A_2=\{1\}$, $p(u,x_1|y_1)=p(u,x_1)$, we obtain following set of inequalities from Corollary \ref{corollary:main}:
\begin{itemize}
\item $T_1=\{1\}$: $r_1>R+R_1$
\item $T_1=\{1,2\}$ $r_1+r_2>R+R_1+R_2$ 
\item $S_1=\{1\}$: $r_1<I(U;Y_2)$ 
\end{itemize}

If $R_1>I(U;Y_3)$ and $R_2>I(X_1;Y_3|U)$ in addition to the above conditions, it can be shown through the analysis of the equivocation at the wiretapper that this coding scheme satisfies the secrecy constraint \cite{CsiszarKorner:78}. By the F-M elimination, the secrecy capacity $C_S= \max_{p(u,x_1)}I(U;Y_2)-I(U;Y_3)$ is recovered. 
\end{example}

More examples recovered by the unified coding theorem are relegated to Appendix \ref{appendix:special}.

\section{Duality} \label{sec:duality} 
In this section, we establish a duality theorem that shows interesting similarities among achievability conditions for ADMNs in dual relations. For an $N$-node ADMN \[(\mathcal{X}_1,\ldots,\mathcal{X}_N, \mathcal{Y}_1,\ldots,\mathcal{Y}_N, \prod_{k=1}^Np(y_k|y^{k-1},x^{k-1}))\] with target joint distribution $p^*(x_{[1:N]},y_{[1:N]})$, which we call the original problem in this section to distinguish from dual problems, we define three types of dual problems as follows:  

\begin{itemize}
\item Type-I dual problem consists of an $N$-node ADMN $(\mathcal{Y}_1,\ldots,\mathcal{Y}_N, \mathcal{X}_1,\ldots,\mathcal{X}_N,  \prod_{k=1}^Np_1(x_k|y^{k-1},x^{k-1}))$, i.e., the input and output alphabets are swapped,  and a target joint distribution $p^*_1(x_{[1:N]},y_{[1:N]})$. 

\item Type-II dual problem consists of an $N$-node ADMN $(\mathcal{X}_N,\ldots,\mathcal{X}_1, \mathcal{Y}_N,\ldots,\mathcal{Y}_1,  \prod_{k=1}^Np_2(y_k|y_{k+1}^{N},x_{k+1}^{N}))$, 
i.e.,  the order of nodes is reversed, and a target joint distribution $p^*_2(x_{[1:N]},y_{[1:N]})$. 

\item Type-III dual problem consists of an $N$-node ADMN $(\mathcal{Y}_N,\ldots,\mathcal{Y}_1, \mathcal{X}_N,\ldots,\mathcal{X}_1,   \prod_{k=1}^Np_3(x_k|y_{k+1}^{N},x_{k+1}^{N}))$,  
i.e., the input and output alphabets are swapped and the order of nodes is reversed, and a target joint distribution $p^*_3(x_{[1:N]},y_{[1:N]})$. 
\end{itemize} 
We note that $p^*_{t}(x_{[1:N]},y_{[1:N]})$ for $t\in [1:3]$ is not necessarily the same as   $p^*(x_{[1:N]},y_{[1:N]})$.

For the coding parameters of unified coding for the original problem and its dual problems, we define  $\Omega_d$ as the set of $\omega_d=(\mu, \mathcal{U}_j, r_j, W_k, D_k, p(u_{W_k}|u_{D_k},y_k), p_1(u_{W_k}|u_{D_k},x_k),  p_2(u_{D_k}|u_{W_k},y_k)$, $p_3(u_{D_k}|u_{W_k},x_k)$, $x_k(u_{W_k},u_{D_k},y_k)$, $y_{k,1}(u_{W_k},u_{D_k},x_k)$, $x_{k,2}(u_{W_k},u_{D_k},y_k)$, $y_{k,3}(u_{W_k},u_{D_k},x_k)$  for $k\in [1:N]$ and $j\in [1:\mu])$ such that 
\begin{align*}
(\mu, \mathcal{U}_j, r_j, W_k, D_k, B_k, A_j, p(u_{W_k}|u_{D_k},y_k), x_k(u_{W_k},u_{D_k},y_k)  \mbox{ for } k\in [1:N], j\in [1:\mu])\\
\in \Omega'(\mathcal{X}_1,\ldots,\mathcal{X}_N, \mathcal{Y}_1,\ldots,\mathcal{Y}_N, \prod_{k=1}^Np(y_k|y^{k-1},x^{k-1}), p^*)\\
(\mu, \mathcal{U}_j, r_j, W_k, D_k, B_k, A_j, p_1(u_{W_k}|u_{D_k},x_k), y_{k,1}(u_{W_k},u_{D_k},x_k)  \mbox{ for } k\in [1:N], j\in [1:\mu])\\
\in \Omega'(\mathcal{Y}_1,\ldots,\mathcal{Y}_N, \mathcal{X}_1,\ldots,\mathcal{X}_N,  \prod_{k=1}^Np_1(x_k|y^{k-1},x^{k-1}),p^*_1)\\
(\mu, \mathcal{U}_j, r_j, W_k, D_k, B_k, A_j, p_2(u_{D_k}|u_{W_k},y_k), x_{k,2}(u_{W_k},u_{D_k},y_k)  \mbox{ for } k\in [1:N], j\in [1:\mu])\\
\in \Omega'(\mathcal{X}_N,\ldots,\mathcal{X}_1, \mathcal{Y}_N,\ldots,\mathcal{Y}_1,  \prod_{k=1}^Np_2(y_k|y_{k+1}^{N},x_{k+1}^{N}), p^*_2)\\
(\mu, \mathcal{U}_j, r_j,W_k, D_k, B_k, A_j, p_3(u_{D_k}|u_{W_k},x_k), y_{k,3}(u_{W_k},u_{D_k},x_k)  \mbox{ for } k\in [1:N], j\in [1:\mu])\\
\in \Omega'(\mathcal{Y}_N,\ldots,\mathcal{Y}_1, \mathcal{X}_N,\ldots,\mathcal{X}_1,   \prod_{k=1}^Np_3(x_k|y_{k+1}^{N},x_{k+1}^{N})), p^*_3), 
\end{align*}
where $B_k=A_j=\emptyset$ for $k\in [1:N], j\in [1:\mu]$. We note that for type I and type III dual problems, where the input and output alphabets are swapped, $y_{k,t}$ for $k\in [1:N]$ and $t\in \{1,3\}$ is the function used for the symbol-by-symbol maping from $(U_{W_k}^n, U_{D_k^n}, X_k^n)$ to the channel input $Y_k^n$. 
We also note that for type-II and type-III dual problems, where the node order is reversed, the roles of $W_k$ and $D_k$ are swapped, i.e., node $k$ decodes $U_{W_k}^n$ and compresses the channel output sequence and decoded codewords as $U_{D_k}^n$ for $k\in [1:N]$.

The following duality theorem is directly obtained from Corollary \ref{corollary:main}.
\begin{theorem} \label{thm:duality}
Consider $\omega_d\in \Omega_d$. For the original network, $p^*$ is achievable if for $1\leq k \leq N$ 
\begin{subequations}
\begin{align}
\sum_{j\in S_k} r_j&<\sum_{j\in S_k} I(U_j;U_{S_k[j]\cup S_k^c},Y_k)\\ 
\sum_{j\in T_k} r_j&>\sum_{j\in T_k} I(U_j;U_{T_k[j]\cup D_k},Y_k) \label{eqn:duality_tk}
\end{align}\label{eqn:duality}
\end{subequations}
for all $S_k\subseteq D_k$ such that $S_k\neq \emptyset$ and for all $T_k\subseteq W_k$ such that $T_k\neq \emptyset$.

For the type-I dual network, $p_1^*$ is achievable if for $1\leq k \leq N$ 
\begin{subequations}
\begin{align}
\sum_{j\in S_k} r_j&<\sum_{j\in S_k} I_{1}(U_j;U_{S_k[j]\cup S_k^c},X_k)\\ 
\sum_{j\in T_k} r_j&>\sum_{j\in T_k} I_{1}(U_j;U_{T_k[j]\cup D_k},X_k) 
\end{align} \label{eqn:duality1}
\end{subequations}
for all $S_k\subseteq D_k$ such that $S_k\neq \emptyset$ and for all $T_k\subseteq W_k$ such that $T_k\neq \emptyset$, where the mutual informations are evaluated using the distribution $\prod_{k=1}^N p_1(u_{W_k}|u_{D_k},x_k)\mathbbm{1}_{y_k=y_{k,1}(u_{W_k},u_{D_k},x_k)}p_1(x_k|y^{k-1},x^{k-1})$. 

For the type-II dual network, $p_2^*$ is achievable if for $1\leq k \leq N$ 
\begin{subequations}
\begin{align}
\sum_{j\in S_k} r_j&<\sum_{j\in S_k} I_{2}(U_j;U_{S_k[j]\cup S_k^c},Y_k) \\ 
\sum_{j\in T_k} r_j&>\sum_{j\in T_k} I_{2}(U_j;U_{T_k[j]\cup W_k},Y_k)
\end{align}
\label{eqn:duality2}
\end{subequations}
for all $S_k\subseteq W_k$ such that $S_k\neq \emptyset$ and for all $T_k\subseteq D_k$ such that $T_k\neq \emptyset$, where the mutual informations are evaluated using the distribution $\prod_{k=1}^N p_2(u_{D_k}|u_{W_k},y_k)\mathbbm{1}_{x_k=x_{k,2}(u_{W_k},u_{D_k},y_k)}p_2(y_k|y_{k+1}^{N},x_{k+1}^{N})$.

For the type-III dual network, $p_3^*$ is achievable if for $1\leq k \leq N$ 
\begin{subequations}
\begin{align}
\sum_{j\in S_k} r_j&<\sum_{j\in S_k} I_{3}(U_j;U_{S_k[j]\cup S_k^c},X_k) \\ 
\sum_{j\in T_k} r_j&>\sum_{j\in T_k} I_{3}(U_j;U_{T_k[j]\cup W_k},X_k) 
\end{align}
\label{eqn:duality3}
\end{subequations}
for all $S_k\subseteq W_k$ such that $S_k\neq \emptyset$ and for all $T_k\subseteq D_k$ such that $T_k\neq \emptyset$, where the mutual informations are evaluated using the distribution $\prod_{k=1}^Np_3(u_{D_k}|u_{W_k},x_k)\mathbbm{1}_{y_k= y_{k,3}(u_{W_k},u_{D_k},x_k)}p_3(x_k|y_{k+1}^{N},x_{k+1}^{N})$.
\end{theorem}
\begin{remark}
Theorem \ref{thm:duality} shows some similarities among achievability conditions for dual problems. 
In the achievability conditions (\ref{eqn:duality1}) and (\ref{eqn:duality3}) for type-I and type-III dual problems, respectively, $X_k$ and $Y_k$ are swapped from (\ref{eqn:duality}). In the achievability conditions (\ref{eqn:duality2}) and (\ref{eqn:duality3}) for type-II and type-III dual problems, respectively, $W_k$ and $D_k$ are swapped from (\ref{eqn:duality}). 
\end{remark}

Theorem \ref{thm:duality} includes as special cases many known duality relationships. Examples include the duality between the point-to-point channel coding  (original problem) and the point-to-point source coding (type-I and type-II dual problems) and the duality between Gelfand-Pinsker coding \cite{GelfandPinsker:80} (original problem) and Wyner-Ziv coding \cite{WynerZiv:76} (type-II dual problem). 
Furthermore, using Theorem \ref{thm:duality}, duality can be shown among the achievability results for  multiple-access channel \cite{liao_thesis} (original problem), distributed source coding \cite{Berger:77}, \cite{tung_thesis} (type-I dual problem), multiple-description  \cite{CoverElGamal:82} (type-II dual problem), and broadcast channel \cite{Marton:79} (type-III dual problem). Let us describe the last duality in detail.  For the original network, consider the multiple-access channel in Example \ref{ex:mac} represented as a  three-node ADMN $(\mathcal{X}_1, \mathcal{X}_2, \mathcal{X}_3, \mathcal{Y}_1, \mathcal{Y}_2, \mathcal{Y}_3, \prod_{k=1}^3p(y_k|y^{k-1},x^{k-1}))$
such that $\mathcal{X}_3=\mathcal{X}_{3,1}\times \mathcal{X}_{3,2}$, 
$\mathcal{X}_{3,1}=\mathcal{Y}_1$, $\mathcal{X}_{3,2}=\mathcal{Y}_2$,   $H(Y_1)=R_1$, $p(y_2|y_1,x_1)=p(y_2)$, $H(Y_2)=R_2$, and $p(y_3|y_{[1:2]},x_{[1:2]})=p(y_3|x_{[1:2]})$ with a target distribution $p^*(x_{[1:3]},y_{[1:3]})$ such that $X_3=(Y_1,Y_2)$.
For the type-I dual problem, in which the input and output alphabets are swapped, we consider the distributed source coding problem, by assuming $p_1(x_1)$, $p_1(x_2|x_1,y_1)=p_1(x_2|x_1)$, and $p_1(x_3)=p_1(x_{3,1}|y_1)p_1(x_{3,2}|y_2)$ such that $\max_{p(y_1)}I(Y_1;X_{3,1})=R_1$ and $\max_{p(y_2)}I(Y_2;X_{3,2})=R_2$, and assuming target distribution $p^*_1(x_{[1:3]},y_{[1:3]})$ similarly as in  Example \ref{ex:distribuedlc}.
Next, for the type-II dual problem, in which the order of nodes is reversed, we consider the multiple-description scenario without combined reconstruction, which is a special case of Example \ref{ex:mdc}, by assuming $p_2(y_3)$, $p_2(y_2|y_3,x_3)=p_2(y_2|x_{3,2})$ such that $\max_{p(x_{3,2})}I(X_{3,2};Y_2)=R_2$, and $p_2(y_1|y_2,y_3,x_2,x_3)=p_2(y_1|x_{3,1})$ such that $\max_{p(x_{3,1})}$ $I(X_{3,1};Y_1)=R_1$, and assuming target distribution $p^*_2(x_{[1:3]},y_{[1:3]})$ similarly as in Example \ref{ex:mdc}. 
For the type-III dual problem, we consider the broadcast channel problem in Example  \ref{ex:bc} by assuming $p_3(x_3)=p_3(x_{3,1})p_3(x_{3,2})$ such that $H(X_{3,1})=R_1$ and $H(X_{3,2})=R_2$, $p_3(x_2|y_3,x_3)=p(x_2|y_3)$, and $p_3(x_1|y_2,y_3,x_2,x_3)=p_3(x_1|y_3,x_2)$, and assuming target distribution $p^*_3(x_{[1:3]},y_{[1:3]})$ such that $Y_1=X_{3,1}$ and $Y_2=X_{3,2}$. 

Choose $\omega_d\in\Omega_d$ as follows: $\mu=2, \mathcal{U}_1=\mathcal{Y}_1\times \mathcal{X}_1 \times \mathcal{V}_1,  \mathcal{U}_2=\mathcal{Y}_2\times \mathcal{X}_2\times \mathcal{V}_2,  W_1=\{1\}, W_2=\{2\}, D_3=\{1,2\}$. 
For the multiple access channel problem, we have $U_1=(Y_1, X_1)$ and $U_2=(Y_2,X_2)$, where $p(x_1|y_1)=p(x_1)$ and $p(x_2|y_2)=p(x_2)$ such that $p=p^*$. 
For the distributed source coding problem, we let $U_1=(V_1,Y_1)$, $U_2=(V_2,Y_2)$, and $y_{3,1}(u_1,u_2,x_3)=y_{3,1}(v_1,v_2)$, where $p_1(u_1|x_1)=p_1(y_1)p_1(v_1|x_1)$ and $p_1(u_2|x_2)=p_1(y_2)p_1(v_2|x_2)$ such that $p_1=p_1^*$. 
For the multiple description problem, we assume $U_1=(X_1,X_{3,1})$ and  $U_2=(X_2,X_{3,2})$, where $p_2(u_1,u_2|y_3)=p_2(x_{3,1})p_2(x_{3,2})p_2(x_1|y_3)p_2(x_2|y_3)$ such that $p_2=p_2^*$. 
For the broadcast channel problem, we assume $U_1=(X_{3,1}, V_1)$, $U_2=(X_{3,2},V_2)$, and $y_{3,3}(u_1,u_2)=y_{3,3}(v_1,v_2)$ such that $p_3(v_1,v_2|x_3)=p_3(v_1,v_2)$ and $p_3=p_3^*$. 
    
Now, we are ready to apply Theorem \ref{thm:duality}. For the multiple access channel, we obtain the following condition to achieve $p^*$: 
\begin{align*}
r_1&>I(U_1;Y_1)=R_1\cr
r_2&>I(U_2;Y_2)=R_2\cr
r_1&<I(U_1;U_2,Y_3)=I(X_1;Y_3|X_2)\cr
r_2&<I(U_2;U_1,Y_3)=I(X_2;Y_3|X_1) \cr
r_1+r_2&<I(U_1,U_2;Y_3)+I(U_1;U_2)=I(X_1,X_2;Y_3), 
\end{align*}
which corresponds to the capacity region of the multiple access channel \cite{liao_thesis} by performing the F-M elimination, taking the union over $p(x_1)p(x_2)$, and incorporating coded time sharing \cite{HanKobayashi:81}. 

Next, for the distributed source coding problem, we obtain the following condition to achieve $p_1^*$: 
\begin{align*}
r_1&>I_{1}(U_1;X_1)=I_{1}(V_1;X_1) \cr
r_2&>I_{1}(U_2;X_2)=I_{1}(V_2;X_2) \cr
r_1&<I_{1}(U_1;U_2,X_3)=R_1+I_{1}(V_1;V_2)\cr
r_2&<I_{1}(U_2;U_1,X_3)=R_2+I_{1}(V_1;V_2)\cr
r_1+r_2&<I_{1}(U_1,U_2;X_3)+I_{1}(U_1;U_2)=R_1+R_2+I_{1}(V_1;V_2), 
\end{align*}
which corresponds to the Berger-Tung inner bound \cite{Berger:77} by performing the F-M elimination.

Next, for the multiple description problem without combined reconstruction, we obtain the following condition to achieve $p_2^*$: 
\begin{align*}
r_1&<I_{2}(U_1;Y_1)=R_1\cr
r_2&<I_{2}(U_2;Y_2)=R_2\cr
r_1&>I_{2}(U_1;Y_3)=I_{2}(X_1;Y_3)\cr
r_2&>I_{2}(U_2;Y_3)=I_{2}(X_2;Y_3)\cr
r_1+r_2&>I_{2}(U_1,U_2;Y_3)+I_2(U_1;U_2)=I_{2}(X_1;Y_3)+I_{2}(X_2;Y_3),
\end{align*}
which corresponds to the optimal rate-distortion region by performing the F-M elimination and taking the union over $p(x_1|y_3)p(x_2|y_3)$ that satisfies the distortion constraints.

Lastly, for the broadcast channel problem, we obtain the following condition to achieve $p_3^*$: 
\begin{align*}
r_1&<I_{3}(U_1;X_1)=I_{3}(V_1;X_1)\cr
r_2&<I_{3}(U_2;X_2)=I_{3}(V_2;X_2)\cr
r_1&>I_{3}(U_1;X_3)=R_1\cr
r_2&>I_{3}(U_2;X_3)=R_2\cr
r_1+r_2&>I_{3}(U_1,U_2;X_3)+I_{3}(U_1;U_2)=R_1+R_2+I_{3}(V_1;V_2),
\end{align*}
which corresponds to the Marton bound \cite{Marton:79} by performing the F-M elimination.
 
\section{Generalized Decode-Compress-Amplify-and-Forward} \label{sec:GDCAF}
In this section, as an application of our unified coding theorem, we present a generalized decode-compress-amplify-and-forward (GDCAF) bound for a single-source single-destination $N$-node ADMN, which includes hybrid coding \cite{MineroLimKim:15} and distributed decode-and-forward (DDF)  \cite{LimKimKim:14} bounds applied to layered networks as special cases. Furthermore, we show  an example where our GDCAF scheme strictly outperforms many previously known schemes. 

For a single-source single-destination $N$-node ADMN where node 1 and node $N$ are the source and the destination, respectively, we let $H(Y_1)=R$ for some $R\geq 0$, $p(y_k|x^{k-1},y^{k-1})=p(y_k|x^{k-1},y^{k-1}_2)$ for $k\in [2:N]$, and $p^*$ such that $X_N=Y_1$, i.e., $Y_1^n$ corresponds to the message of rate $R$ that does not affect the remaining channels and node $N$ wishes to decode $Y_1^n$ reliably. The following theorem gives the GDCAF bound using Corollary \ref{corollary:main}. 

\begin{theorem}[GDCAF bound for a single-source single-destination ADMN] \label{thm:GDCAF} 
For  a single-source single-destination $N$-node ADMN, a rate of $R$ is achievable if 
 \begin{align*} 
R& < \min_{S,T: S\subseteq T\subseteq [2:N-1]}I(X_1,U_S,\hat{Y}_T;\hat{Y}_{T^c}, Y_N|U_{S^c}) -\sum_{j\in T} I(\hat{Y}_j;Y_j|U_{[2:N-1]}, \hat{Y}_{T[j]}, X_1)\\
&~~~~~~~~~~~~~~~~~~~~~~~~~~~~~~~~~+ H(U_{S^c})-\sum_{j\in S^c} H(U_j|Y_j)
\end{align*}
 for some $p(x_1,u_2,\ldots,u_N)\prod_{j\in [2:N-1]}p(\hat{y}_j|y_j,u_j)$ and functions $x_k(u_k,\hat{y}_k,y_k)$ for $k\in [2:N-1]$ such that 
 \begin{align*}
\sum_{j\in S'} I(U_j;U_{S'[j]}) <\sum_{j\in S'} I(U_j;Y_j)
\end{align*} 
for all $S'\subseteq [2:N-1]$.
\end{theorem}

\begin{remark}
In the GDCAF bound, $U_k$ and $\hat{Y}_k$ for $k\in[2:N-1]$ correspond to the partial information about the message decoded by node $k$ and the compressed version of $Y_k$ at node $k$, respectively. 
\end{remark}
\begin{remark}
The GDCAF bound recovers hybrid coding bound \cite{MineroLimKim:15} applied to layered networks by letting $U_k=\emptyset$ for $k\in [2:N-1]$. 
\end{remark}
\begin{remark}
The GDCAF bound recovers DDF bound \cite{LimKimKim:14} applied to layered networks  by letting $\hat{Y}_k=\emptyset$ and $U_k=(V_k,X_k)$ for $k\in [2:N-1]$, and $p(x_1,u_2,\ldots,u_{N-1})=\prod_{k=2}^{N-1}p(x_k)p(x_1|x_2^{N-1})p(v_2^N|x^{N-1})$. 
\end{remark}

\begin{proof} 
We apply Corollary \ref{corollary:main} to derive the GDCAF bound. We choose $\omega'\in \Omega'$ as follows: $\mu=2N-2$, $W_1=[1:N]$, $A_{N}=[1:N-1], p(u_1, \ldots, u_{N}|y_1)=\mathbbm{1}_{u_1=y_1}\cdot p(u_2,\ldots, u_{N})$, $X_1=U_N$, $D_k=\{k\}$, $W_k=\{N+k-1\}$,  $A_{N+k-1}=\{k\}$ for $k\in [2:N-1]$, $D_N=\{1\}$, and $B_N=[2:2N-2]$. For notational convenience, let $r_k'$ and $\hat{Y}_k$ denote $r_{N+k-1}$ and $U_{N+k-1}$, respectively, for $k\in [2:N-1]$.   

By applying Corollary \ref{corollary:main} for the aforementioned choice of $\omega'$, we obtain the following bounds:  
\begin{align*}
r_1&>R \cr 
\sum_{j\in S}r_j&>\sum_{j\in S} I(U_j;U_{S[j]}) \\
\sum_{j\in [1:N]}r_j&>R+\sum_{j\in [2:N-1]} I(U_j;U^{j-1})\\
r_k&<I(U_k;Y_k) \\
r_k'&>I(\hat{Y}_k;Y_k|U_k)\\
r_1+r_N+\sum_{j\in S}r_j+\sum_{j\in T} r_j' &< \sum_{j\in S} I(U_j;U_{S[j]\cup S^c}, \hat{Y}_{T^c}, Y_N)+I(X_1; \hat{Y}_{T^c}, Y_N|U_{[2:N-1]})\cr
&~~~~~~+\sum_{j\in T} I(\hat{Y}_j;X_1,U_{[2:N-1]}, \hat{Y}_{T[j] \cup T^c}, Y_N|U_j) 
\end{align*}
for all $k\in [2:N]$ and for all $S$ and $T$ such that $S\subseteq T \subseteq [2:N-1]$. By performing the F-M elimination, Theorem \ref{thm:GDCAF} is proved. 
\end{proof}

\begin{figure}[t]
 \centering
   {\includegraphics[width=110mm]{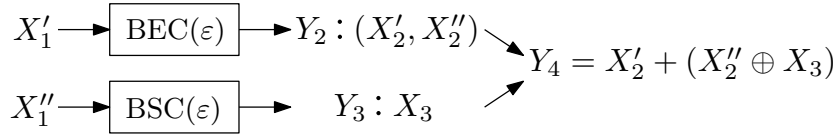} }
   \caption{An example of diamond network}  \label{fig:diamond_ex}
\end{figure}

Now, let us show that the GDCAF scheme strictly outperforms many existing schemes for the diamond network illustrated in Fig. \ref{fig:diamond_ex}. In this diamond network, $X_1=(X_1',X_1'')$, the channel from $X_1'$ to $Y_2$ is a binary erasure channel (BEC) with erasure probability $\varepsilon$, and the channel from $X_1''$ to $Y_3$ is a binary symmetric channel (BSC) with cross over probability $\varepsilon$, where the two channels are correlated, i.e., cross over in the BSC happens iff an erasure occurs in the BEC. We have $X_2=(X_2',X_2'')$ and $Y_4=X_2'+(X_2''\oplus X_3)$, where $X_2',X_2''$, and $X_3$ are binary and $\oplus$ denotes the XOR operation. Let us assume $\varepsilon=1-H(1/3)$. Then, the capacity of the BEC is $1-\varepsilon \approx 0.9183$ and the capacity of the BSC is $1-H(\varepsilon)\approx 0.5918$. The cut-set bound is given as $\log 3\approx 1.5850$.

In this diamond network, the GDCAF scheme achieves the cut-set bound, which means the capacity is $\log 3$. Let $\hat{Y}_2=\hat{Y}_3=U_3=\emptyset$, $U_2=(X_1',X_2')$, $X_1'\sim\operatorname{Bern}(1/2)$, $(X_1'',X_2')\sim \operatorname{Unif}(\{(0,0),(0,1),(1,1)\})$, $X_2''=0$ if $Y_2$ is not erased, $X_2''=1$ if $Y_2$ is erased, and $X_3=Y_3$, where $X_1'$ and $(X_1'',X_2')$ are independent. Then, the GDCAF rate is given as 
\begin{align*}
&\min\{I(X_1,X_2';Y_4),I(X_1;Y_4|X_1',X_2')+I(X_1',X_2';Y_2)\}=\log 3. 
\end{align*}

Our choice of coding parameters for GDCAF scheme indicates that node 2 performs a combination of partial DF and AF and node 3 employs AF. Now, let us show the suboptimality of DF \cite{schein_thesis}, partial DF, DDF \cite{LimKimKim:14}, and hybrid coding \cite{MineroLimKim:15}  schemes. First, the  DF rate  \cite{schein_thesis} is given as $\min\{I(X_1;Y_2)$, $I(X_1;Y_3)$, $I(X_2,X_3;Y_4)\}$ for some $p(x_1)p(x_2,x_3)$, which is upper-bounded by 1.
Next, a partial DF scheme can be constructed by combining Marton coding with common information \cite{liang_thesis} for the first hop and the optimal scheme \cite{SlepianWolf:73} for the MAC with common information for the second hop. From the first hop, we have
\begin{align*}
R&<I(U_0,U_1;Y_2)+I(U_2;Y_3|U_0)-I(U_1;U_2|U_0)\\
&\le I(U_0,U_1,X_1;Y_2)+I(U_2,X_1;Y_3|U_0)\\
&\le I(X_1;Y_2)+I(X_1;Y_3)\\
&\le 1-\varepsilon+1-H(\varepsilon)\approx 1.5101
\end{align*}
for some $p(u_0,u_1,u_2)$ and a function $x_1(u_0,u_1,u_2)$.
Thus, such a partial DF scheme is also suboptimal.

On the other hand, the DDF rate \cite{LimKimKim:14} for the diamond channel is given as follows: 
\begin{align*}
R&<\min\{I(X_1;U_2,U_3|X_2,X_3)-I(U_2;X_1,X_3|X_2,Y_2)-I(U_3;U_2,X_1,X_2|X_3,Y_3),\\
&~~~~~~~~~~~I(X_1,X_2;U_3,Y_4|X_3)-I(U_3;X_1,X_2|X_3,Y_3),\\
&~~~~~~~~~~~I(X_1,X_3;U_2,Y_4|X_2)-I(U_2;X_1,X_3|X_2,Y_2),I(X_2,X_3;Y_4)\}\\
&=\min\{H(U_{2},U_3|X_{2},X_3)-H(U_2|X_2,Y_2)-H(U_3|X_3,Y_3),\\
&~~~~~~~~~~~I(X_2;Y_4|U_3,X_3)+I(U_3;Y_3|X_3),I(X_3;Y_4|U_2,X_2)+I(U_2;Y_2|X_2),I(X_{2},X_3;Y_4)\}
\end{align*}
for $p(x_2)p(x_3)p(x_1,u_2,u_3|x_2,x_3)$. The DDF rate is upper-bounded as follows: 
\begin{align*}
R&\le H(U_{2},U_3|X_{2},X_3)-H(U_2|X_2,Y_2)-H(U_3|X_3,Y_3)\\
&\le H(U_2|X_2)+H(U_3|X_3)-H(U_2|X_2,Y_2)-H(U_3|X_3,Y_3)\\
&=I(U_2;Y_2|X_2)+I(U_3;Y_3|X_3)\\
&\le I(U_2,X_1;Y_2|X_2)+I(U_3,X_1;Y_3|X_3)\\
&\le H(Y_2)-H(Y_2|X_1)+H(Y_3)-H(Y_3|X_1)\\
&=I(X_1;Y_2)+I(X_1;Y_3)\\
&\le 1-\varepsilon+1-H(\varepsilon)\approx 1.5101,
\end{align*}
hence it is suboptimal. 

Lastly, the hybrid coding scheme achieves the following rate
\begin{align*}
&\min\{I(X_1;\hat{Y}_2,\hat{Y}_3,Y_4),I(X_1,\hat{Y}_2;\hat{Y}_3,Y_4)-I(\hat{Y}_2;Y_2|X_1),\\
&\;\;\;\;\;\;\;\;I(X_1,\hat{Y}_3;\hat{Y}_2,Y_4)-I(\hat{Y}_3;Y_3|X_1),I(X_1,\hat{Y}_2,\hat{Y}_3;Y_4)-I(\hat{Y}_2,\hat{Y}_3;Y_2,Y_3|X_1)\}
\end{align*}
for some $p(x_1)p(\hat{y}_2|y_2)p(\hat{y}_3|y_3)$ and functions $x_2(\hat{y}_2,y_2)$ and $x_3(\hat{y}_3,y_3)$.
Let us show the suboptimality of the hybrid coding scheme by contradiction. Assume hybrid coding achieves the capacity. Fix $p(x_1)p(\hat{y}_2|y_2)p(\hat{y}_3|y_3)$ and functions $x_2(\hat{y}_2,y_2)$ and $x_3(\hat{y}_3,y_3)$ that achieves the capacity.
Then \[I(X_1,\hat{Y}_2,\hat{Y}_3;Y_4)-I(\hat{Y}_2,\hat{Y}_3;Y_2,Y_3|X_1)\ge \log 3\] should hold, which implies
\begin{align}
H(Y_4)&=\log 3 \label{eqn:dia_y4} \\
H(Y_4|X_1,\hat{Y}_2,\hat{Y}_3)&=0 \label{eqn:dia_y4_con}\\
I(\hat{Y}_2,\hat{Y}_3;Y_2,Y_3|X_1)&=0. \label{eqn:dia_markov}
\end{align}  
 
From (\ref{eqn:dia_markov}), we have $p(\hat{y}_2,\hat{y}_3|y_2,y_3,x_1)=p(\hat{y}_2,\hat{y}_3|x_1)$ for all $(y_2,y_3,x_1)$ such that $p(y_2,y_3,x_1)>0$. On the other hand, for all $(y_2,y_3,x_1)$ such that $p(y_2,y_3,x_1)>0$, $p(\hat{y}_2,\hat{y}_3|y_2,y_3,x_1)=p(\hat{y}_2|y_2)p(\hat{y}_3|y_3)$ should hold. Thus, we conclude that $p(\hat{y}_2|y_2)p(\hat{y}_3|y_3)=p(\hat{y}_2,\hat{y}_3|x_1)$ for all $(y_2,y_3,x_1)$ such that $p(y_2,y_3,x_1)>0$. Note that $p(x_1)>0$ for at least three different values of $x_1$ since otherwise the achievable rate is upper bounded by $R\le H(X_1)\le 1$. This means that there exists $x\in \{0,1\}$ such that $p(x_1'=0,x_1''=x)>0$ and $p(x_1'=1,x_1''=x)>0$. Then, we have
\begin{align*}
p(\hat{y}_2|y_2=0)p(\hat{y}_3|y_3=x)&=p(\hat{y}_2,\hat{y}_3|x_1'=0,x_1''=x)\\
p(\hat{y}_2|y_2=E)p(\hat{y}_3|y_3=1-x)&=p(\hat{y}_2,\hat{y}_3|x_1'=0,x_1''=x)\\
p(\hat{y}_2|y_2=1)p(\hat{y}_3|y_3=x)&=p(\hat{y}_2,\hat{y}_3|x_1'=1,x_1''=x)\\
p(\hat{y}_2|y_2=E)p(\hat{y}_3|y_3=1-x)&=p(\hat{y}_2,\hat{y}_3|x_1'=1,x_1''=x).
\end{align*}
Therefore, we get
\begin{align*}
p(\hat{y}_2|y_2=0)p(\hat{y}_3|y_3=x)=
p(\hat{y}_2|y_2=E)p(\hat{y}_3|y_3=1-x)=
p(\hat{y}_2|y_2=1)p(\hat{y}_3|y_3=x)
\end{align*}
for all $\hat{y}_2$ and $\hat{y}_3$. Thus, we conclude $p(\hat{y}_3|y_3)=p(\hat{y}_3)$ by summing $p(\hat{y}_2|y_2=0)p(\hat{y}_3|y_3=x)=p(\hat{y}_2|y_2=E)p(\hat{y}_3|y_3=1-x)$ over $\hat{y}_2$. Similarly, we conclude $p(\hat{y}_2|y_2)=p(\hat{y}_2)$. Thus, we get $p(x_1,y_2,y_3,\hat{y}_2,\hat{y}_3)=p(x_1)p(y_2,y_3|x_1)p(\hat{y}_2)p(\hat{y}_3)$.

From (\ref{eqn:dia_y4_con}), we have
\begin{align}
H(Y_4|X_1=x_1,\hat{Y}_2=\hat{y}_2,\hat{Y}_3=\hat{y}_3)=0 \label{eqn:dia_cond2}
\end{align} 
for all $(x_1,\hat{y}_2,\hat{y}_3)$ such that $p(x_1,\hat{y}_2,\hat{y}_3)=p(x_1)p(\hat{y}_2)p(\hat{y}_3)>0$. 
On the other hand, since we are assuming hybrid coding achieves the capacity, we must have $I(X_1;\hat{Y}_2,\hat{Y}_3,Y_4)\ge H(Y_4)=\log 3$. Note that $I(X_1;\hat{Y}_2,\hat{Y}_3,Y_4)=I(X_1;Y_4|\hat{Y}_2,\hat{Y}_3)\le H(Y_4|\hat{Y}_2,\hat{Y}_3)$. Thus, we conclude $H(Y_4|\hat{Y}_2,\hat{Y}_3)=H(Y_4)=\log 3$. This means 
\begin{align}
H(Y_4|\hat{Y}_2=\hat{y}_2, \hat{Y}_3=\hat{y}_3)=\log 3 \label{eqn:dia_cond1}
\end{align}
for all $\hat{y}_2$ and $\hat{y}_3$ such that $p(\hat{y}_2)p(\hat{y}_3)>0$. 

Because $Y_4$ is a function of $X_2$ and $X_3$, there exists a function $y_4(\cdot)$ such that $Y_4=y_4(Y_2,Y_3,\hat{Y}_2,\hat{Y}_3)$. Fix $(\hat{y}_2,\hat{y}_3)$ such that $p(\hat{y}_2)p(\hat{y}_3)>0$. Because $p(x_1'=0, x_1''=x)>0$, we have $p(x_1'=0, x_1''=x, y_2=0, y_3=x)>0$ and $p(x_1'=0, x_1''=x, y_2=E, y_3=1-x)>0$. Due to (\ref{eqn:dia_cond2}), this implies 
\begin{align}
y_4(0,x,\hat{y}_2,\hat{y}_3)=y_4(E,1-x,\hat{y}_2,\hat{y}_3).
\end{align}
Similarly, from $p(x_1'=1, x_1''=x)>0$, we conclude 
\begin{align}
y_4(1,x,\hat{y}_2,\hat{y}_3)=y_4(E,1-x,\hat{y}_2,\hat{y}_3).
\end{align}
Hence, we have 
\begin{align}
y_4(0,x,\hat{y}_2,\hat{y}_3)=y_4(1,x,\hat{y}_2,\hat{y}_3)=y_4(E,1-x,\hat{y}_2,\hat{y}_3). \label{eqn:dia_same}
\end{align}
From (\ref{eqn:dia_cond2}), (\ref{eqn:dia_cond1}), and (\ref{eqn:dia_same}), both $p(x_1'=0, x_1''=1-x)$  and $p(x_1'=1, x_1''=1-x)$ should be positive, which implies 
\begin{align}
y_4(0,1-x,\hat{y}_2,\hat{y}_3)=y_4(1,1-x,\hat{y}_2,\hat{y}_3)=y_4(E,x,\hat{y}_2,\hat{y}_3).\label{eqn:dia_same2}
\end{align}
Note that (\ref{eqn:dia_same}) and (\ref{eqn:dia_same2}) implies that for $(\hat{Y}_2,\hat{Y}_3)=(\hat{y}_2, \hat{y}_3)$, $Y_4$ has only two possibilities, which is contradictory to (\ref{eqn:dia_cond1}). Therefore, we conclude that the hybrid coding scheme is strictly suboptimal.

\section{Acyclic Gaussian Network} \label{sec:Gaussian} 
In this section, we consider an $N$-node acyclic Gaussian network (AGN), in which the channel output $Y_k$ and channel input $X_k$ at node $k$ are $r_k$-dimensional and $t_k$-dimensional vectors, respectively, and the channel from nodes $1,\ldots,k-1$ to node $k$ is given as 
\begin{align}
Y_k=\sum_{j\in [1:k-1]}H_{kj}X_j+\sum_{j\in [1:k-1]}H_{kj}'Y_j+Y_k', \label{eqn:gaussian_model}
\end{align}
where $H_{kj}$ is an $r_k\times t_j$ matrix, $H_{kj}'$ is an $r_k\times r_j$ matrix, and $Y_k'\sim \mathcal{N}(0, \Lambda_{Y_k'})$ is independent from $X^{k-1}$ and $Y^{k-1}$. Let $n$ denote the number of channel uses and $Y_k^n$ and $X_k^n$ for $k\in [1:N]$ denote the $r_k\times n$ and $t_k\times n$ matrices, respectively, where the $i$-th vectors $Y_{k,i}$ and $X_{k,i}$ of $Y_k^n$ and $X_k^n$ are the channel output and channel input, respectively, at node $k$ at the $i$-th channel use. Similarly as in ADMNs, node operations are sequential, i.e., $Y_1^n$ is generated according to (\ref{eqn:gaussian_model}) and node 1 maps it to $X_1^n$, $Y_2^n$ is generated according to  (\ref{eqn:gaussian_model}) and node 2 maps it to $X_2^n$, and so on.

The objective of this network is specified by  a $\theta$-dimensional nonnegative vector $\Theta^*$ and a nonnegative  function $\Theta$ that maps $(x_1, \ldots, x_N, y_1, \ldots, y_N)\in \mathcal{X}_1\times \ldots\times \mathcal{X}_N\times \mathcal{Y}_1\times \ldots \times \mathcal{Y}_N$ to a $\theta$-dimensional vector whose elements are all nonnegative. For a set of node processing functions $Y_k^n\rightarrow X_k^n$, $k=1,\ldots, N$, the $\epsilon$-probability of error for $\epsilon>0$ is defined as 
\begin{align*}
P_e^{(n)}(\Theta, \Theta^*,\epsilon)&=1-P\Big(\frac{1}{n}\sum_{i=1}^n \Theta(X_{1,i},\ldots, X_{N,i}, Y_{1,i},\ldots, Y_{N,i})\prec (1+\epsilon)\Theta^* \Big).
\end{align*}

We say $(\Theta, \Theta^*)$ is achievable if there exists a sequence of node processing functions $Y_k^n\rightarrow X_k^n$, $k=1,\ldots, N$, such that $\lim_{n\rightarrow \infty} P_e^{(n)}(\Theta, \Theta^*,\epsilon)=0$ for any $\epsilon>0$. We note that $(\Theta, \Theta^*)$ can be used for imposing input power constraint and quadratic distortion constraint between a source sequence and a reconstructed sequence. 

To derive a sufficient condition for achieving $(\Theta, \Theta^*)$ by using Theorem \ref{thm:main}, let us first assume a distribution $f^*(x_{[1:N]}, y_{[1:N]})$ such that $\E(\Theta(X_1, \ldots, X_N, Y_1, \ldots, Y_N))\prec \Theta^*$. For such $f^*(x_{[1:N]}, y_{[1:N]})$, consider a subset $\Omega_g(f^*)$ of $\Omega(f^*)$,\footnote{The definition of $\Omega$ in Section \ref{sec:main} can be naturally generalized for continous random variables.}  where $U_j$ is an $a_j$-dimensional vector and $U_{ \uin{W}_k}$ and $X_k$ have the form of 
\begin{align}
U_{ \uin{W}_k}&=G_{k}U_{ \uin{D}_k}+G_k'Y_k+U_{ \uin{W}_k}' \label{eqn:uwk_ori}\\
X_k&=F_{k}U_{ \uin{D}_k\cup  \uin{W}_k}+F_k'Y_k. \nonumber
\end{align}
In the above, $G_k$ is a $\sum_{j\in  \uin{W}_k}a_j \times \sum_{j\in  \uin{D}_k}a_j$ matrix, $G_k'$ is a $\sum_{j\in  \uin{W}_k}a_j \times r_k$ matrix, $U_{ \uin{W}_k}'\sim \mathcal{N}(0, \Lambda_{U_{ \uin{W}_k}'})$ is a  $\sum_{j\in  \uin{W}_k}a_j$-dimensional Gaussian vector, $F_k$ is a $t_k \times \sum_{j\in  \uin{D}_k\cup  \uin{W}_k}a_j$ matrix, and $F_k'$ is a $t_k \times r_k$ matrix, where $U_{ \uin{W}_k}'$ is independent from $U_{ \uin{D}_k}$ and $Y_k$.

Note that $U_{ \uin{W}_k}$ and $X_k$ can be rewritten as follows:  
\begin{align}
U_{ \uin{W}_k}&=\sum_{j\in [1:k-1]}G_{kj}U_{ \uin{W}_j}+G_k'Y_k+U_{ \uin{W}_k}' \label{eqn:uwk} \\
X_k&=\sum_{j\in[1:k]}F_{kj}U_{ \uin{W}_j}+F_k'Y_k, \label{eqn:xk}
\end{align}
where the columns of $G_{kj}$ and $F_{kj}$ corresponding to $U_{ \uin{D}_k}$ and $U_{ \uin{D}_k\cup  \uin{W}_k}$ are from $G_k$ and $F_k$, respectively, and the other columns of $G_{kj}$ and $F_{kj}$ are zero vectors. 

The following lemma gives $U_{ \uin{W}_k}$ and $Y_k$, whose proof is given at the end of this section.
\begin{lemma} \label{lemma:nonrecursive}
For $k\in [1:N]$, we have 
\begin{align}
\begin{bmatrix}U_{ \uin{W}_k} \\ Y_k \end{bmatrix}=\sum_{j\in [1:k]} \Phi_{kj} \Psi_j, \label{eqn:uwk_nonrecursive}
\end{align}
where 
\begin{align}
\Phi_{kj}\triangleq
\begin{cases} 
\sum_{S: \{j,k\}\subseteq S \subseteq [j:k]}\prod_{i\in [1:|S|-1]}\Upsilon_{S_{[i+1]}S_{[i]}} & \mbox{ if } j<k \\
\mathbf{I} &\mbox{ if } j=k 
\end{cases}, \label{eqn:Mkj}
\end{align}

\begin{align} \label{eqn:Mkjprime2}
\Psi_j\triangleq \begin{bmatrix}
G_j'Y_j'+U_{ \uin{W}_j}' \\ Y_j'
\end{bmatrix},
\end{align}

\begin{align} \label{eqn:Mkjprime1}
\Upsilon_{k'j'}\triangleq \begin{bmatrix}
G_{k'j'}+G_{k'}'\sum_{i\in [j':k'-1]}H_{k'i}F_{ij'} & G_{k'}'(H_{k'j'}F_{j'}'+H_{k'j'}') \\
\sum_{i\in [j':k'-1]}H_{k'i}F_{ij'} & H_{k'j'}F_{j'}'+H_{k'j'}' \end{bmatrix}. \end{align}
 \end{lemma}

From Lemma \ref{lemma:nonrecursive}, for any $S\subseteq  \uin{W}^N$, we can construct a matrix $\Phi_{U_S}$ such that 
$U_S=\Phi_{U_S}\Psi^{k(S)}$, where $k(S)=\max(\{k:S\cap  \uin{W}_k \neq \emptyset\})$. Also, for any $k\in [1:N]$ and $S\subseteq  \uin{W}^k$, we can construct matrices $\Phi_{U_S,Y_k}$ such that 
$[U_S^t~ Y_k^t]^t=\Phi_{U_S,Y_k}\Psi^k$. 

Now, we are ready to present a sufficient condition for achieving $(\Theta, \Theta^*)$ for an $N$-node AGN. 
\begin{theorem}\label{thm:gaussian}
For an $N$-node AGN, $(\Theta, \Theta^*)$ is achievable if there exists $\omega_g\in \Omega_g(f^*)$ for some $f^*$ that satisfies  $\E(\Theta(X_1, \ldots, X_N, Y_1, \ldots, Y_N))\prec \Theta^*$ such that for $1\leq k \leq N$
\begin{align*}
\sum_{j\in \lset{S}_k} r_j&<\frac{1}{2}\log \frac{|\Phi_{U_{  \uset{S}_k^c}, Y_k}\Lambda_{\Psi^k} \Phi_{U_{  \uset{S}_k^c}, Y_k}^t  |\cdot \prod_{j\in   \uset{S}_k}|\Phi_{U_{j},U_{A_j}}\Lambda_{\Psi^{k(\{j\}\cup A_j)}} \Phi_{U_{j},U_{A_j}}^t|}{|\Phi_{U_{ \uin{D}_k\cup  \uin{B}_k},Y_k}\Lambda_{\Psi^k} \Phi_{U_{ \uin{D}_k\cup  \uin{B}_k},Y_k}^t|\cdot \prod_{j\in   \uset{S}_k}|\Phi_{U_{A_j}}\Lambda_{\Psi^{k(A_j)}} \Phi_{U_{A_j}}^t|}\\ 
\sum_{j\in \lset{T}_k} r_j&>\frac{1}{2}\log \frac{\prod_{j\in   \uset{T}_k}|\Phi_{U_{j},U_{A_j}}\Lambda_{\Psi^{k(\{j\}\cup A_j)}} \Phi_{U_{j},U_{A_j}}^t|}{|\Lambda_{U_{  \uset{T}_k}'}|\cdot \prod_{j\in   \uset{T}_k}|\Phi_{U_{A_j}}\Lambda_{\Psi^{k(A_j)}} \Phi_{U_{A_j}}^t|}
\end{align*}
for all $\lset{S}_k\subseteq \lin{D}_k\cup \lin{B}_k$ such that $\lset{S}_k\cap \lin{D}_k\neq \emptyset$ and for all $\lset{T}_k\subseteq \lin{W}_k$ such that $\lset{T}_k\neq \emptyset$, where $\lin{D}_k$, $\lin{B}_k$, $\lin{W}_k$, $ \uset{S}_k$, and $\uset{T}_k$ are defined in Theorem \ref{thm:main} and $\Lambda_{\Psi^k}$ is a block diagonal matrix with diagonal blocks 
\begin{align*}
\Lambda_{\Psi_j}=
\begin{bmatrix}
\Lambda_{U_{\uin{W}_j}'}+G_j'\Lambda_{Y_j'}G_j'^t & G_j'\Lambda_{Y_j'} \\
\Lambda_{Y_j'}G_j'^t & \Lambda_{Y_j'}
\end{bmatrix},~~ j\in [1:k].
\end{align*}
\end{theorem}

\begin{proof}   
We apply Theorem \ref{thm:main} for $f^*(x_{[1:N]}, y_{[1:N]})$ such that $\E(\Theta(X_1, \ldots, X_N, Y_1, \ldots, Y_N))\prec \Theta^*$ and $\omega_g\in \Omega_g(f^*)$. Then, the inequalities in Theorem \ref{thm:main} is given as follows: for $k\in [1:N]$, $\lset{S}_k\subseteq \lin{D}_k\cup \lin{B}_k$ such that $\lset{S}_k\cap \lin{D}_k\neq \emptyset$, and $\lset{T}_k\subseteq \lin{W}_k$ such that $\lset{T}_k\neq \emptyset$, we have 
\begin{align*}
\sum_{j\in \lset{S}_k} r_j&<\sum_{j\in   \uset{S}_k}I(U_j;U_{  \uset{S}_k[j]\cup   \uset{S}_k^c},Y_k|U_{A_j})\cr
&=\sum_{j\in   \uset{S}_k}h(U_j|U_{A_j})-h(U_{  \uset{S}_k}|U_{  \uset{S}_k^c},Y_k)\cr
&=\sum_{j\in   \uset{S}_k}\left(h(U_j,U_{A_j})-h(U_{A_j})\right)-h(U_{ \uin{D}_k\cup  \uin{B}_k},Y_k)+h(U_{  \uset{S}_k^c}, Y_k)\cr
&=\frac{1}{2}\log \frac{|\Phi_{U_{  \uset{S}_k^c}, Y_k}\Lambda_{\Psi^k} \Phi_{U_{  \uset{S}_k^c}, Y_k}^t  |\cdot \prod_{j\in   \uset{S}_k}|\Phi_{U_{j},U_{A_j}}\Lambda_{\Psi^{k(\{j\}\cup A_j)}} \Phi_{U_{j},U_{A_j}}^t|}{|\Phi_{U_{ \uin{D}_k\cup  \uin{B}_k},Y_k}\Lambda_{\Psi^k} \Phi_{U_{ \uin{D}_k\cup  \uin{B}_k},Y_k}^t|\cdot \prod_{j\in   \uset{S}_k}|\Phi_{U_{A_j}}\Lambda_{\Psi^{k(A_j)}} \Phi_{U_{A_j}}^t|}
\end{align*}
and 
\begin{align*}
\sum_{j\in \lset{T}_k} r_j&>\sum_{j\in   \uset{T}_k}I(U_j;U_{  \uset{T}_k[j]\cup  \uin{D}_k},Y_k|U_{A_j})\cr
&=\sum_{j\in   \uset{T}_k}h(U_j|U_{A_j})-h(U_{  \uset{T}_k}|U_{ \uin{D}_k},Y_k)\cr
&=\sum_{j\in   \uset{T}_k}\left(h(U_j,U_{A_j})-h(U_{A_j})\right)-h(U_{  \uset{T}_k}')\cr
&=\frac{1}{2}\log \frac{\prod_{j\in   \uset{T}_k}|\Phi_{U_{j},U_{A_j}}\Lambda_{\Psi^{k(\{j\}\cup A_j)}} \Phi_{U_{j},U_{A_j}}^t|}{|\Lambda_{U_{  \uset{T}_k}'}|\cdot \prod_{j\in   \uset{T}_k}|\Phi_{U_{A_j}}\Lambda_{\Psi^{k(A_j)}} \Phi_{U_{A_j}}^t|}. 
\end{align*}
Now, by following the standard discretization procedure \cite{McEliece:77} and the typical average lemma \cite{ElGamalKim:11}, Theorem \ref{thm:gaussian} is proved. 
\end{proof}

\begin{remark}
If $U_j, j\in [1:\nu]$ has the form of $G_j''U_{A_j}+U''_j$ as a special case of (\ref{eqn:uwk_ori}) for some matrix $G_j''$, where $U''_j$ is independent of $U_{A_j}$, then we have
\begin{align*}
\frac{|\Phi_{U_{j},U_{A_j}}\Lambda_{\Psi^{k(\{j\}\cup A_j)}} \Phi_{U_{j},U_{A_j}}^t|}{|\Phi_{U_{A_j}}\Lambda_{\Psi^{k(A_j)}} \Phi_{U_{A_j}}^t|}=|\Lambda_{U''_j}|.
\end{align*}
\end{remark}

\subsubsection*{Proof of Lemma \ref{lemma:nonrecursive}}
By substituting (\ref{eqn:xk}) into (\ref{eqn:gaussian_model}), $Y_k$ is written as follows: 
\begin{align}
Y_k&=\sum_{i\in[1:k-1]}H_{ki}\left(\sum_{j\in [1:i]}F_{ij}U_{ \uin{W}_j}+F_i'Y_i\right)+\sum_{j\in [1:k-1]}H_{kj}'Y_j+Y_k'\cr
&=\sum_{i\in[1:k-1]}\sum_{j\in [1:i]}H_{ki}F_{ij}U_{ \uin{W}_j}+\sum_{j\in [1:k-1]}(H_{kj}F_j'+H_{kj}')Y_j+Y_k'\cr
&\overset{(a)}{=}\sum_{j\in[1:k-1]}\sum_{i\in [j:k-1]}H_{ki}F_{ij}U_{ \uin{W}_j}+\sum_{j\in [1:k-1]}(H_{kj}F_j'+H_{kj}')Y_j+Y_k', \label{eqn:yk1}
\end{align}
where $(a)$ is by changing the summation order. 

Next, by substituting (\ref{eqn:yk1}) into (\ref{eqn:uwk}), $U_{ \uin{W}_k}$ is given as follows:
\begin{align*}
U_{ \uin{W}_k}&=\sum_{j\in [1:k-1]}G_{kj}U_{ \uin{W}_j}+G_k'\left(\sum_{j\in[1:k-1]}\sum_{i\in [j:k-1]}H_{ki}F_{ij}U_{ \uin{W}_j}+\sum_{j\in [1:k-1]}(H_{kj}F_j'+H_{kj}')Y_j+Y_k'\right)+U_{ \uin{W}_k}' \cr
&=\sum_{j\in [1:k-1]}(G_{kj}+G_k'\sum_{i\in [j:k-1]}H_{ki}F_{ij} )U_{ \uin{W}_j}+\sum_{j\in [1:k-1]}G_k'(H_{kj}F_j'+H_{kj}')Y_j+G_k'Y_k'+U_{ \uin{W}_k}'.
\end{align*}
Hence, we have
\begin{align}
\begin{bmatrix}U_{ \uin{W}_k} \\ Y_k \end{bmatrix} = \sum_{j\in [1:k-1]} \Upsilon_{kj} \begin{bmatrix} U_{ \uin{W}_j} \\ Y_j\end{bmatrix} + \Psi_k \label{eqn:uwk_recursive}
\end{align}
where $ \Upsilon_{kj}$ and $\Psi_k$ are defined in (\ref{eqn:Mkjprime1}) and (\ref{eqn:Mkjprime2}), respectively. 

We prove Lemma \ref{lemma:nonrecursive} by solving the recursive formula in (\ref{eqn:uwk_recursive}) using strong induction. For $k=1$, $[U_{ \uin{W}_1}^t~Y_1^t]^t=\Psi_1$ from (\ref{eqn:uwk_recursive}), and hence Lemma \ref{lemma:nonrecursive} holds trivially. For $k>1$, assume that $[U_{ \uin{W}_j}^t~Y_j^t]^t=\sum_{i\in [1:j]} \Phi_{ji} \Psi_i$ for all $j<k$. Then, 
\begin{align}
[U_{ \uin{W}_k}^t~Y_k^t]^t &\overset{(a)}{=}\sum_{j\in [1:k-1]} \Upsilon_{kj}[U_{ \uin{W}_j}^t~Y_j^t]^t + \Psi_k \cr
&\overset{(b)}{=}\sum_{j\in [1:k-1]} \sum_{i\in [1:j]} \Upsilon_{kj} \Phi_{ji} \Psi_i + \Psi_k \cr
&\overset{(c)}{=}\sum_{i\in [1:k-1]} \sum_{j\in [i:k-1]} \Upsilon_{kj} \Phi_{ji} \Psi_i + \Psi_k, 
\end{align}
where $(a)$ is from (\ref{eqn:uwk_recursive}), $(b)$ is from the induction assumption, and $(c)$ is by changing the summation order. Now, we have 
\begin{align*}
\sum_{j\in [i:k-1]} \Upsilon_{kj} \Phi_{ji} &= \Upsilon_{ki}+\sum_{j\in [i+1:k-1]} \Upsilon_{kj} \Phi_{ji} \cr
&= \Upsilon_{ki}+\sum_{j\in [i+1:k-1]} \Upsilon_{kj} \sum_{S: \{i,j\}\subseteq S \subseteq [i:j]}\prod_{i'\in [1:|S|-1]}\Upsilon_{S_{[i'+1]}S_{[i']}} \cr
&=\sum_{S: \{i,k\}\subseteq S \subseteq [i:k]} \prod_{i'\in [1:|S|-1]}\Upsilon_{S_{[i'+1]}S_{[i']}} \cr
&=\Phi_{ki}.
\end{align*}
Hence, we have $[U_{ \uin{W}_k}^t~Y_k^t]^t=\sum_{j\in [1:k]} \Phi_{kj} \Psi_j$, and Lemma \ref{lemma:nonrecursive} is proved by strong induction. 
\endproof

\section{Conclusion} \label{sec:conclusion}
  
 We showed a unified achievability theorem that generalizes most of achievability results in network information theory that are based on random coding. Our single-letter rate expression has a very simple form.   This was made possible due to our framework, where many different ingredients in network information theory are treated in a unified way, and our coding scheme that consists of a few basic ingredients but is at least as powerful as many existing schemes. Using our result, obtaining many new achievability results in network information theory can now be done more easily. As a simple application of our main theorem, we derived a generalized decode-compress-amplify-and-forward bound and showed it strictly outperforms previously known results.  Because the final expression of our main theorem has a simple form, it enables us to get new insights. As an example, we showed how to derive three types of network duality from our main theorem. Our result can be made more general if other coding strategies such as structured codes can also be incorporated in our setting. However, such a task does not seem easy and the rate expression may become too complicated.

  \appendix
\subsection{Bounding the second term in the summation in (\ref{eqn:error_analysis_ub}) } \label{appendix:sec_error}
For given $k\in [1:N]$, we have 
\begin{align*}
&P(\mathcal{E}_{k,2}\cap \bigcap_{j=1}^{k-1}(\mathcal{E}_{j,1}\cup \mathcal{E}_{j,2}\cup \mathcal{E}_{j,3})^c) \cr
&\leq P((U_{ \uin{D}_k\cup  \uin{B}_k}^n(l'_{\lin{D}_k\cup \lin{B}_k}), Y_k^n)\in \mathcal{T}_{\epsilon_k}^{(n)} \mbox{ for some } l'_{\lin{D}_k\cup \lin{B}_k} \mbox{ s.t. } l'_{\lin{D}_k}\neq L_{\lin{D}_k}, \hat{L}_{\lin{D}_j,j}=L_{\lin{D}_j},\cr 
&~~(U_{ \uin{D}_j\cup  \uin{B}_j}^n(L_{\lin{D}_j},L_{\lin{B}_j}), Y_j^n)\in \mathcal{T}_{\epsilon_j}^{(n)},(U_{ \uin{D}_j}^n(L_{\lin{D}_j}), U_{ \uin{W}_j}^n(L_{\lin{D}_j}, L_{\lin{W}_j}), Y_j^n)\in \mathcal{T}_{\epsilon_j'}^{(n)} \mbox{ for all } j\in [1:k-1])\cr 
&\leq\sum_{\lset{S}_k\subseteq \lin{D}_k\cup \lin{B}_k\atop \lset{S}_k\cap \lin{D}_k\neq \emptyset}P((U_{ \uin{D}_k\cup  \uin{B}_k}^n(l'_{\lset{S}_k},L_{\lset{S}_k^c}), Y_k^n)\in \mathcal{T}_{\epsilon_k}^{(n)} \mbox{ for some } l'_{\lset{S}_k} \mbox{ s.t. } l'_i\neq L_i \mbox{ for all } i\in \lset{S}_k, \hat{L}_{\lin{D}_j,j}=L_{\lin{D}_j},\cr 
&~~(U_{ \uin{D}_j\cup  \uin{B}_j}^n(L_{\lin{D}_j},L_{\lin{B}_j}), Y_j^n)\in \mathcal{T}_{\epsilon_j}^{(n)},(U_{ \uin{D}_j}^n(L_{\lin{D}_j}), U_{ \uin{W}_j}^n(L_{\lin{D}_j}, L_{\lin{W}_j}), Y_j^n)\in \mathcal{T}_{\epsilon_j'}^{(n)} \mbox{ for all } j\in [1:k-1]).
\end{align*}

For $\lset{S}_k\subseteq \lin{D}_k\cup \lin{B}_k$ such that $\lset{S}_k\cap \lin{D}_k\neq \emptyset$, we have  
\begin{align}
&P((U_{ \uin{D}_k\cup  \uin{B}_k}^n(l'_{\lset{S}_k},L_{\lset{S}_k^c}), Y_k^n)\in \mathcal{T}_{\epsilon_k}^{(n)} \mbox{ for some } l'_{\lset{S}_k} \mbox{ s.t. } l'_i\neq L_i \mbox{ for all } i\in \lset{S}_k, \hat{L}_{\lin{D}_j,j}=L_{\lin{D}_j},
\cr &~~~~(U_{ \uin{D}_j\cup  \uin{B}_j}^n(L_{\lin{D}_j},L_{\lin{B}_j}), Y_j^n)\in \mathcal{T}_{\epsilon_j}^{(n)},(U_{ \uin{D}_j}^n(L_{\lin{D}_j}), U_{ \uin{W}_j}^n(L_{\lin{D}_j}, L_{\lin{W}_j}), Y_j^n)\in \mathcal{T}_{\epsilon_j'}^{(n)} \mbox{ for all } j\in [1:k-1])\cr
=&P((U_{ \uin{D}_k\cup  \uin{B}_k}^n(l'_{\lset{S}_k},l_{\lset{S}_k^c}), Y_k^n)\in \mathcal{T}_{\epsilon_k}^{(n)} \mbox{ for some } l'_{\lset{S}_k} \mbox{ s.t. } l'_i\neq l_i \mbox{ for all } i\in \lset{S}_k, \hat{L}_{\lin{D}_j,j}=l_{\lin{D}_j},
\cr &~~~~(U_{ \uin{D}_j\cup  \uin{B}_j}^n(l_{\lin{D}_j},l_{\lin{B}_j}), Y_j^n)\in \mathcal{T}_{\epsilon_j}^{(n)},(U_{ \uin{D}_j}^n(l_{\lin{D}_j}), U_{ \uin{W}_j}^n(l_{\lin{D}_j}, l_{\lin{W}_j}), Y_j^n)\in \mathcal{T}_{\epsilon_j'}^{(n)} \mbox{ for all } j\in [1:k-1],\cr  &~~~~~~~~L_{\lin{W}^{k-1}}=l_{\lin{W}^{k-1}} \mbox{ for some } l_{\lin{W}^{k-1}})\cr
=&P((U^n_{ \uin{D}_k\cup  \uin{B}_k}(l'_{\lset{S}_k},l_{\lset{S}_k^c}), Y_k^n)\in \mathcal{T}_{\epsilon_k}^{(n)} \mbox{ for some } l'_{\lset{S}_k} \mbox{ s.t. } l'_i\neq l_i \mbox{ for all } i\in \lset{S}_k, U_{ \uin{D}_j\cup  \uin{B}_j}^n(\cdot)\in \mathcal{A}_{j}(Y_j^n,  l_{\lin{D}_j\cup \lin{B}_j}),\cr
&~~~~ (U_{ \uin{D}_j}^n(l_{\lin{D}_j}),U_{ \uin{W}_j}^n(l_{\lin{D}_j},\cdot))\in \mathcal{B}_{j}(Y_j^n, l_{\lin{D}_j\cup \lin{W}_j}) \mbox{ for all } j\in[1:k-1] \mbox{ for some } l_{\lin{W}^{k-1}} ) \label{eqn:sec_error_1}
\end{align}
where $\mathcal{A}_{j}(y_j^n,l_{\lin{D}_j\cup \lin{B}_j})$ is the ensemble of codebooks $u_{ \uin{D}_j\cup  \uin{B}_j}^n(\cdot)$ such that node $j$ with received channel output $y_j^n$ decodes $\hat{l}_{\lin{D}_j,j}$ as $l_{\lin{D}_j}$ and the joint typicality (\ref{eqn:packing_typicality}) is satisfied for given $l_{\lin{D}_j\cup \lin{B}_j}$ and $\mathcal{B}_{j}(y_j^n,l_{\lin{D}_j\cup \lin{W}_j})$ is the ensemble of codebooks $(u_{ \uin{D}_j}^n(l_{\lin{D}_j}),u_{ \uin{W}_j}^n(l_{\lin{D}_j},\cdot))$ such that node $j$ with received channel output $y_j^n$ and decoded index vector $l_{\lin{D}_j}$ chooses $l_{\lin{W}_j}$ and the joint typicality (\ref{eqn:covering_typicality}) is satisfied for given $l_{\lin{D}_j\cup \lin{W}_j}$. More precisely, we define 
\begin{align*}
\mathcal{A}_{j}(y_j^n, l_{\lin{D}_j\cup \lin{B}_j})&\triangleq  \{u_{ \uin{D}_j\cup  \uin{B}_j}^n(\cdot):(u_{ \uin{D}_j\cup  \uin{B}_j}^n(\tilde{l}_{\lin{D}_j\cup\lin{B}_j}), y_j^n)\notin T_{\epsilon_j} \mbox{ for all } \tilde{l}_{\lin{D}_j}<l_{\lin{D}_j}, \tilde{l}_{\lin{B}_j}, \cr
&~~~~~~~~ (u_{ \uin{D}_j\cup  \uin{B}_j}^n(l_{\lin{D}_j\cup \lin{B}_j}), y_j^n)\in T_{\epsilon_j} \} \cr
\mathcal{B}_{j}(y_j^n,l_{\lin{D}_j\cup \lin{W}_j})&\triangleq  \{(u_{ \uin{D}_j}^n(l_{\lin{D}_j}),u_{ \uin{W}_j}^n(l_{\lin{D}_j},\cdot)):(u_{ \uin{D}_j}^n(l_{\lin{D}_j}), u^n_{ \uin{W}_j}(l_{\lin{D}_j},\tilde{l}_{\lin{W}_j}),y_j^n)\notin T_{\epsilon_j'} \mbox{ for all } \tilde{l}_{\lin{W}_j}<l_{\bar{W}_j}, \cr
&~~~~~~~~ (u_{ \uin{D}_j}^n(l_{\lin{D}_j}), u^n_{ \uin{W}_j}(l_{\lin{D}_j\cup \lin{W}_j}),y_j^n)\in T_{\epsilon_j'} \}.
\end{align*}
Continuing with the bound in (\ref{eqn:sec_error_1}), we have
\begin{align}
&P((U^n_{ \uin{D}_k\cup  \uin{B}_k}(l'_{\lset{S}_k},l_{\lset{S}_k^c}), Y_k^n)\in \mathcal{T}_{\epsilon_k}^{(n)} \mbox{ for some } l'_{\lset{S}_k} \mbox{ s.t. } l'_i\neq l_i \mbox{ for all } i\in \lset{S}_k, U_{ \uin{D}_j\cup  \uin{B}_j}^n(\cdot)\in \mathcal{A}_{j}(Y_j^n, l_{\lin{D}_j\cup \lin{B}_j}), \cr
&~~(U_{ \uin{D}_j}^n(l_{\lin{D}_j}),U_{ \uin{W}_j}^n(l_{\lin{D}_j},\cdot))\in \mathcal{B}_{j}(Y_j^n, l_{\lin{D}_j\cup \lin{W}_j}) \mbox{ for all } j\in[1:k-1] \mbox{ for some } l_{\lin{W}^{k-1}} )\cr
&=P((U^n_{ \uin{D}_k\cup  \uin{B}_k}(l'_{\lset{S}_k},l_{\lset{S}_k^c}), Y_k^n)\in \mathcal{T}_{\epsilon_k}^{(n)},U_{ \uin{D}_j\cup  \uin{B}_j}^n(\cdot)\in \mathcal{A}_{j}(Y_j^n,l_{\lin{D}_j\cup \lin{B}_j}), (U_{ \uin{D}_j}^n(l_{\lin{D}_j}),U_{ \uin{W}_j}^n(l_{\lin{D}_j},\cdot))\in \mathcal{B}_{j}(Y_j^n, l_{\lin{D}_j\cup \lin{W}_j})\cr
&~~ \mbox{ for all } j\in [1:k-1] \mbox{ for some } l_{\lin{W}^{k-1}} \mbox{ s.t. }l_i\neq l'_i \mbox{ for all } i\in \lset{S}_k  \mbox{ for some } l'_{\lset{S}_k} )\cr
&\leq \sum_{l_{\lset{S}_k}'}P((U^n_{ \uin{D}_k\cup  \uin{B}_k}(l'_{\lset{S}_k},l_{\lset{S}_k^c}), Y_k^n)\in \mathcal{T}_{\epsilon_k}^{(n)},U_{ \uin{D}_j\cup  \uin{B}_j}^n(\cdot)\in \mathcal{A}_{j}(Y_j^n,l_{\lin{D}_j\cup \lin{B}_j}), \cr
&(U_{ \uin{D}_j}^n(l_{\lin{D}_j}),U_{ \uin{W}_j}^n(l_{\lin{D}_j},\cdot))\in \mathcal{B}_{j}(Y_j^n, l_{\lin{D}_j\cup \lin{W}_j}) \mbox{ for all } j\in [1:k-1] \mbox{ for some } l_{\lin{W}^{k-1}}  \mbox{ s.t. }l_i\neq l'_i \mbox{ for all } i\in \lset{S}_k )\cr
&\overset{(a)}{=}\sum_{l_{\lset{S}_k}'}P((U^n_{ \uin{D}_k\cup  \uin{B}_k}(l'_{\lset{S}_k},L''_{\lset{S}_k^c}), \tilde{Y}_k^n)\in \mathcal{T}_{\epsilon_k}^{(n)}, U_{ \uin{D}_j\cup  \uin{B}_j}^n(\cdot)\in \mathcal{A}_{j}(Y_j^n,l_{\lin{D}_j\cup \lin{B}_j}), \cr
&(U_{ \uin{D}_j}^n(l_{\lin{D}_j}),U_{ \uin{W}_j}^n(l_{\lin{D}_j},\cdot))\in \mathcal{B}_{j}(Y_j^n, l_{\lin{D}_j\cup \lin{W}_j}) \mbox{ for all } j\in [1:k-1] \mbox{ for some } l_{\lin{W}^{k-1}}  \mbox{ s.t. }l_i\neq l'_i \mbox{ for all } i\in \lset{S}_k )\label{eqn:ldoubleprime}
\end{align}
where $\tilde{Y}_k^n$ is the channel output sequence at node $k$ assuming that decoded index vector $\hat{l}''_{\lin{D}_j,j}=\hat{l}''_{\lin{D}_j,j}(y_j^n,l'_{\lset{S}_k},u_{ \uin{D}_j\cup  \uin{B}_j}^n(\tilde{l}_{\lin{D}_j\cup \lin{B}_j})$ for all $\tilde{l}_{\lin{D}_j\cup \lin{B}_j} \mbox{ such that } \tilde{l}_i\neq l_i' \mbox{ for all } i\in \lset{S}_k\cap(\lin{D}_j\cup \lin{B}_j))$ and covering index vector $l''_{\lin{W}_j}=l''_{\lin{W}_j}(y_j^n,l'_{\lset{S}_k},\hat{l}''_{\lin{D}_j,j},$ $u^n_{ \uin{D}_j}(\hat{l}''_{\lin{D}_j,j}), u_{ \uin{W}_j}^n(\hat{l}''_{\lin{D}_j,j},\tilde{l}_{\lin{W}_j})$  for all 
$\tilde{l}_{\lin{W}_j} \mbox{ such that } \tilde{l}_i\neq l_i' \mbox{ for all } i\in \lset{S}_k\cap \lin{W}_j)$ at node $j\in [1:k-1]$ are chosen according to the following rule: 
\begin{itemize}
\item Find the smallest  $\hat{l}''_{\lin{D}_j,j}\in \{\tilde{l}_{\lin{D}_j}:\tilde{l}_i\neq l_i' \mbox{ for all } i\in \lset{S}_k\cap \lin{D}_j  \}$ such that 
\begin{align*}
(u_{ \uin{D}_j\cup  \uin{B}_j}^n(\hat{l}''_{\lin{D}_j,j},\tilde{\tilde{l}}''_{\lin{B}_j}), y_j^n)\in \mathcal{T}_{\epsilon_j}^{(n)}
\end{align*}
for some $\tilde{\tilde{l}}''_{\lin{B}_j}\in \{\tilde{l}_{\lin{B}_j}:\tilde{l}_i\neq l_i' \mbox{ for all } i\in \lset{S}_k\cap \lin{B}_j \}$. If there is no such index vector, let $\hat{l}''_{\lin{D}_j,j}$ be the smallest one in $\{\tilde{l}_{\lin{D}_j}:\tilde{l}_i\neq l_i' \mbox{ for all } i\in \lset{S}_k \cap \lin{D}_j\}$. 

\item Find the smallest $l''_{\lin{W}_j}\in \{\tilde{l}_{\lin{W}_j}:\tilde{l}_i\neq l_i' \mbox{ for all } i\in \lset{S}_k\cap \lin{W}_j \}$ such that 
\begin{align*}
(u_{ \uin{D}_j}^n(\hat{l}''_{\lin{D}_j,j}), u_{ \uin{W}_j}^n(\hat{l}''_{\lin{D}_j,j}, l''_{\lin{W}_j}), y_j^n)\in \mathcal{T}_{\epsilon_j'}^{(n)}. 
\end{align*}
If there is no such index vector, let $l''_{\lin{W}_j}$ be the smallest one in $\{\tilde{l}_{\lin{W}_j}:\tilde{l}_i\neq l_i' \mbox{ for all } i\in \lset{S}_k\cap \lin{W}_j \}$.
\end{itemize}
Note that $(a)$ follows because if $U_{ \uin{D}_j\cup  \uin{B}_j}^n(\cdot)\in \mathcal{A}_{j}(Y_j^n,l_{\lin{D}_j\cup \lin{B}_j}), (U_{\lin{D}_j}^n(l_{\lin{D}_j}),U_{ \uin{W}_j}^n(l_{\lin{D}_j},\cdot))\in \mathcal{B}_{j}(Y_j^n, l_{\lin{D}_j\cup \lin{W}_j})$ for all $j\in [1:k-1]$ for some $l_{\lin{W}^{k-1}}$ such that $l_i\neq l'_i$ for all $i\in \lset{S}_k$, then $\hat{L}''_{\lin{D}_j,j}=l_{\lin{D}_j}$, $L''_{\lin{W}_j}=l_{\lin{W}_j}$ for all $j\in[1:k-1]$, and hence $\tilde{Y}_k^n=Y_k^n$.

By discarding some constraints, (\ref{eqn:ldoubleprime}) is upper-bounded by 
\begin{align*}
\sum_{l_{\lset{S}_k}'}P((U^n_{ \uin{D}_k\cup  \uin{B}_k}(l'_{\lset{S}_k},L''_{\lset{S}_k^c}), \tilde{Y}_k^n)\in \mathcal{T}_{\epsilon_k}^{(n)}).
\end{align*}

Now, we can show that the joint distribution of $(U^n_{ \uin{D}_k\cup  \uin{B}_k}(l'_{\lset{S}_k},L''_{\lset{S}_k^c}), \tilde{Y}_k^n)$ is given as follows: 
\begin{align}
P(U^n_{ \uin{D}_k\cup  \uin{B}_k}(l'_{\lset{S}_k},L''_{\lset{S}_k^c})=u^n_{ \uin{D}_k\cup  \uin{B}_k}, \tilde{Y}_k^n=y_k^n)=p(u_{  \uset{S}_k^c}^n, y_k^n)\prod_{j\in   \uset{S}_{k}}\prod_{i=1}^np(u_{j,i}|u_{A_j,i}), \label{eqn:seconderror_distr}
\end{align}
where $  \uset{S}_k$ is defined in (\ref{eqn:barSk}).

Now, we can obtain the following upper bound: 
\begin{align*}
&\sum_{l_{\lset{S}_k}'}P((U^n_{ \uin{D}_k\cup  \uin{B}_k}(l'_{\lset{S}_k},L''_{\lset{S}_k^c}), \tilde{Y}_k^n)\in \mathcal{T}_{\epsilon_k}^{(n)}) \cr 
&=2^{n\sum_{j\in \lset{S}_k}r_j}\sum_{(u_{  \uset{S}_k^c}^n, y_k^n)\in \mathcal{T}_{\epsilon_k}^{(n)}} \sum_{ u_{  \uset{S}_k}^n \in \mathcal{T}_{\epsilon_k}^{(n)}(U_{  \uset{S}_k}|u_{  \uset{S}_k^c}^n, y_k^n)}p(u_{  \uset{S}_k^c}^n, y_k^n)\prod_{j\in   \uset{S}_{k}}\prod_{i=1}^np(u_{j,i}|u_{A_j,i})  \cr 
&\leq 2^{n\sum_{j\in \lset{S}_k}r_j}\sum_{(u_{  \uset{S}_k^c}^n, y_k^n)\in \mathcal{T}_{\epsilon_k}^{(n)}} p(u_{  \uset{S}_k^c}^n, y_k^n)\cdot 2^{n(H(U_{  \uset{S}_k}|U_{  \uset{S}_k^c},Y_k)+\delta(\epsilon_k))} \cdot  \prod_{j\in   \uset{S}_k}2^{-n(H(U_j|U_{A_j})-\delta(\epsilon_k))}\cr 
&\leq 2^{n\sum_{j\in \lset{S}_k}r_j}\cdot 2^{-n(\sum_{j\in   \uset{S}_k }H(U_j|U_{A_j})-H(U_{  \uset{S}_k}|U_{  \uset{S}_k^c},Y_k)-(1+\nu)\delta(\epsilon_k))}\cr
&=2^{n\sum_{j\in \lset{S}_k}r_j}\cdot 2^{-n(\sum_{j\in   \uset{S}_k }(H(U_j|U_{A_j})-H(U_{j}|U_{A_j}, U_{  \uset{S}_k[j]}, U_{  \uset{S}_k^c},Y_k))-(1+\nu)\delta(\epsilon_k))}\cr
&= 2^{n\sum_{j\in \lset{S}_k}r_j} \cdot 2^{-n(\sum_{j\in   \uset{S}_k} I(U_j;U_{  \uset{S}_k[j]\cup   \uset{S}_k^{c}}, Y_k|U_{A_j})-(1+\nu)\delta(\epsilon_k))},
\end{align*}
which tends to zero as $n$ tends to infinity if 
\begin{align}
\sum_{j\in \lset{S}_k}r_j&<\sum_{j\in   \uset{S}_k} I(U_j;U_{  \uset{S}_k[j]\cup   \uset{S}_k^{c}}, Y_k|U_{A_j})-(1+\nu)\delta(\epsilon_k).  \label{eqn:sec_error_given_sk}
\end{align}

Therefore, $P(\mathcal{E}_{k,2}\cap \bigcap_{j=1}^{k-1}(\mathcal{E}_{j,1}\cup \mathcal{E}_{j,2}\cup \mathcal{E}_{j,3})^c)$ tends to zero as $n$ tends to infinity when (\ref{eqn:sec_error_given_sk}) is satisfied 
for all $\lset{S}_k\subseteq \lin{D}_k\cup \lin{B}_k$ such that $\lset{S}_k\cap \lin{D}_k\neq \emptyset$.


\subsection{Bounding the third term in the summation in (\ref{eqn:error_analysis_ub})} \label{appendix:third_error}
The proof follows similar steps to the mutual covering lemma in \cite{ElGamlvanderMeulen:81}. For given $k\in [1:N]$, we have 
\begin{align*}
&P(\mathcal{E}_{k,3}\cap \mathcal{E}_{k,1}^c \cap \mathcal{E}_{k,2}^c )\cr
&\leq P((U_{ \uin{D}_k}^n(L_{\lin{D}_k}), Y_k^n)\in \mathcal{T}_{\epsilon_k}^{(n)}, (U_{ \uin{D}_k}^n(L_{\lin{D}_k}), U_{ \uin{W}_k}^n(L_{\lin{D}_k}, l_{\lin{W}_k}), Y_k^n)\notin \mathcal{T}_{\epsilon_k'}^{(n)} \mbox{ for all } l_{\lin{W}_k} )\cr 
&= \sum_{l_{\lin{D}_k},(u_{ \uin{D}_k}^n, y_k^n)\in \mathcal{T}_{\epsilon_k}^{(n)}}P(L_{\lin{D}_k}=l_{\lin{D}_k}, U_{ \uin{D}_k}^n(l_{\lin{D}_k})=u_{ \uin{D}_k}^n,Y_k^n=y_k^n)P(|\mathcal{L}_{k}(l_{\lin{D}_k},u_{ \uin{D}_k}^n,y_k^n)|=0|U_{ \uin{D}_k}^n(l_{\lin{D}_k})=u_{ \uin{D}_k}^n)
\end{align*}
where 
\begin{align*}
\mathcal{L}_{k}(l_{\lin{D}_k},u_{ \uin{D}_k}^n,y_k^n)\triangleq \{l_{\lin{W}_k}: (u_{ \uin{D}_k}^n, U_{ \uin{W}_k}^n(l_{\lin{D}_k\cup\lin{W}_k}), y_k^n)\in \mathcal{T}_{\epsilon_k'}^{(n)} \}.
\end{align*}

Consider $l_{\lin{D}_k}$ and $(u_{ \uin{D}_k}^n, y_k^n)\in \mathcal{T}_{\epsilon_k}^{(n)}$. From the Chebyshev lemma, we have 
\begin{align*}
P(|\mathcal{L}_{k}(l_{\lin{D}_k},u_{ \uin{D}_k}^n,y_k^n)|=0|U_{ \uin{D}_k}^n(l_{\lin{D}_k})=u_{ \uin{D}_k}^n)\leq \frac{\var(|\mathcal{L}_{k}(l_{\lin{D}_k},u_{ \uin{D}_k}^n,y_k^n)||U_{ \uin{D}_k}^n(l_{\lin{D}_k})=u_{ \uin{D}_k}^n)}{(\E[|\mathcal{L}_{k}(l_{\lin{D}_k},u_{ \uin{D}_k}^n,y_k^n)||U_{ \uin{D}_k}^n(l_{\lin{D}_k})=u_{ \uin{D}_k}^n])^2}.
\end{align*}
Now, define the indicator function 
\begin{align*}
I(l_{\lin{W}_k})=
\begin{cases}
1 & \mbox{if }(u_{ \uin{D}_k}^n, U_{ \uin{W}_k}^n(l_{\lin{D}_k\cup \lin{W}_k}), y_k^n)\in \mathcal{T}_{\epsilon_k'}^{(n)}\\ 
0 & \mbox{otherwise}
\end{cases}
\end{align*}
for each $l_{\lin{W}_k}$. Note that $|\mathcal{L}_{k}(l_{\lin{D}_k},u_{ \uin{D}_k}^n,y_k^n)|=\sum_{l_{\lin{W}_k}}I(l_{\lin{W}_k})$.

Due to the symmetry of the codebook generation, for any $l_{\lin{W}_k}, l'_{\lin{W}_k}, \tilde{l}_{\lin{W}_k}, \tilde{l}'_{\lin{W}_k}$, and $\lset{T}_k\subseteq \lin{W}_k$ such that $l_{\lset{T}_k}=l'_{\lset{T}_k}$, $\tilde{l}_{\lset{T}_k}=\tilde{l}'_{\lset{T}_k}$  and  $l_{i}\neq l'_{i}$, $\tilde{l}_{i}\neq \tilde{l}'_{i}$ for all $i\notin \lset{T}_k$, we have $\E[I(l_{\lin{W}_k})I(l_{\lin{W}_k}')|U_{ \uin{D}_k}^n(l_{\lin{D}_k})=u_{ \uin{D}_k}^n]=\E[I(\tilde{l}_{\lin{W}_k})I(\tilde{l}_{\lin{W}_k}')|U_{ \uin{D}_k}^n(l_{\lin{D}_k})=u_{ \uin{D}_k}^n]$. Let $p_{\lset{T}_k}$ for $\lset{T}_k\subseteq \lin{W}_k$ denote $\E[I(l_{\lin{W}_k})I(l_{\lin{W}_k}')|l_{\lset{T}_k}=l'_{\lset{T}_k}, l_{i}\neq l'_{i} \mbox{ for all } i\notin \lset{T}_k, U_{ \uin{D}_k}^n(l_{\lin{D}_k})=u_{ \uin{D}_k}^n]$. 

Then, we have 
\begin{align}
\E[|\mathcal{L}_{k}(l_{\lin{D}_k},u_{ \uin{D}_k}^n,y_k^n)||U_{ \uin{D}_k}^n(l_{\lin{D}_k})=u_{ \uin{D}_k}^n]&=\sum_{l_{\lin{W}_k}} \E[I(l_{\lin{W}_k})|U_{ \uin{D}_k}^n(l_{\lin{D}_k})=u_{ \uin{D}_k}^n]\cr
&=\sum_{l_{\lin{W}_k}} \E[I^2(l_{\lin{W}_k})|U_{ \uin{D}_k}^n(l_{\lin{D}_k})=u_{ \uin{D}_k}^n]\cr
&=2^{n\sum_{i\in \lin{W}_k}r_i} p_{\lin{W}_k} \label{eqn:third_exp}
\end{align}
and 
\begin{align*}
\E[|\mathcal{L}_{k}(l_{\lin{D}_k},u_{ \uin{D}_k}^n,y_k^n)|^2|U_{ \uin{D}_k}^n(l_{\lin{D}_k})=u_{ \uin{D}_k}^n]&=\sum_{l_{\lin{W}_k}}\sum_{\lset{T}_k\subseteq \lin{W}_k}\sum_{l'_{\lset{T}_k^c}: l'_i\neq l_i, \forall i\in \lset{T}_k^c}\E[I(l_{\lin{W}_k})I(l_{\lset{T}_k},l'_{\lset{T}_k^c})|U_{ \uin{D}_k}^n(l_{\lin{D}_k})=u_{ \uin{D}_k}^n]\cr 
&=\sum_{\lset{T}_k\subseteq \lin{W}_k}\sum_{l_{\bar{W}_k}}\sum_{l'_{\lset{T}_k^c}: l'_i\neq l_i, \forall i\in \lset{T}_k^c}\E[I(l_{\lin{W}_k})I(l_{\lset{T}_k},l'_{\lset{T}_k^c})|U_{ \uin{D}_k}^n(l_{\lin{D}_k})=u_{ \uin{D}_k}^n]\cr 
&\leq\sum_{\lset{T}_k\subseteq \lin{W}_k}2^{n(\sum_{i\in \lset{T}_k}r_i+2\sum_{i\in \lset{T}_k^c}r_i)} p_{\lset{T}_k}.
\end{align*}

Since $p_{\emptyset}=p_{\lin{W}_k}^2$, we have 
\begin{align}
&\var(|\mathcal{L}_{k}(l_{\lin{D}_k},u_{ \uin{D}_k}^n,y_k^n)||U_{ \uin{D}_k}^n(l_{\lin{D}_k})=u_{ \uin{D}_k}^n)\cr
&= \E[|\mathcal{L}_{k}(l_{\lin{D}_k},u_{ \uin{D}_k}^n,y_k^n)|^2|U_{ \uin{D}_k}^n(l_{\lin{D}_k})=u_{ \uin{D}_k}^n]- \E^2[|\mathcal{L}_{k}(l_{\lin{D}_k},u_{ \uin{D}_k}^n,y_k^n)||U_{ \uin{D}_k}^n(l_{\lin{D}_k})=u_{ \uin{D}_k}^n]\cr
&\leq \sum_{\lset{T}_k\subseteq \lin{W}_k, \lset{T}_k\neq \emptyset }2^{n(\sum_{i\in \lset{T}_k}r_i+2\sum_{i\in \lset{T}_k^c}r_i)} p_{\lset{T}_k}. \label{eqn:third_var}
\end{align}

Now, for $\lset{T}_k\subseteq \lin{W}_k $ such that $\lset{T}_k\neq \emptyset$, we have 
\begin{align*}
p_{\lset{T}_k}&= P((u_{ \uin{D}_k}^n, U_{ \uin{W}_k}^n(l_{\lin{D}_k\cup \lin{W}_k}), y_k^n)\in \mathcal{T}_{\epsilon_k'}^{(n)}, (u_{ \uin{D}_k}^n, U_{ \uin{W}_k}^n(l_{\lin{D}_k}, l_{\lin{W}_k}'), y_k^n)\in \mathcal{T}_{\epsilon_k'}^{(n)} |l_{\lset{T}_k}=l'_{\lset{T}_k}, \cr
&~~~~~~~~~~~~~~l_{i}\neq l'_{i} \mbox{ for all } i\notin \lset{T}_k, U_{ \uin{D}_k}^n(l_{\lin{D}_k})=u_{ \uin{D}_k}^n) \cr
&\leq 2^{-n(\sum_{j\in   \uset{T}_k}I(U_j;U_{  \uset{T}_k[j]\cup  \uin{D}_k},Y_k|U_{A_j})+2\sum_{j\in   \uset{T}_k^c} I(U_j;U_{  \uset{T}_k\cup   \uset{T}_k^c[j]\cup  \uin{D}_k},Y_k|U_{A_j})-2(1+\nu)\delta(\epsilon_k'))}
\end{align*}
by the joint typicality lemma \cite{ElGamalKim:11}, where $  \uset{T}_k$ is defined in (\ref{eqn:barTk}). Similarly, for $\lset{T}_k\subseteq \lin{W}_k $ such that $\lset{T}_k\neq \emptyset$, we have 
\begin{align*}
p_{\lin{W}_k}\geq 2^{-n(\sum_{j\in   \uset{T}_k}I(U_j;U_{  \uset{T}_k[j]\cup  \uin{D}_k},Y_k|U_{A_j})+\sum_{j\in   \uset{T}_k^c} I(U_j;U_{  \uset{T}_k\cup   \uset{T}_k^c[j]\cup  \uin{D}_k},Y_k|U_{A_j})+(1+\nu)\delta(\epsilon_k'))}. 
\end{align*}

By substituting the above bounds into (\ref{eqn:third_exp}) and (\ref{eqn:third_var}), we obtain 
\begin{align*}
\frac{\var(|\mathcal{L}_{k}(l_{\lin{D}_k},u_{ \uin{D}_k}^n,y_k^n)||U_{ \uin{D}_k}^n(l_{\lin{D}_k})=u_{ \uin{D}_k}^n)}{(\E[|\mathcal{L}_{k}(l_{\lin{D}_k},u_{ \uin{D}_k}^n,y_k^n)||U_{ \uin{D}_k}^n(l_{\lin{D}_k})=u_{ \uin{D}_k}^n])^2}\leq \sum_{\lset{T}_k\subseteq \lin{W}_k, \lset{T}_k\neq \emptyset}2^{-n(\sum_{j\in \lset{T}_k}r_j-\sum_{j\in   \uset{T}_k}I(U_j;U_{  \uset{T}_k[j]\cup  \uin{D}_k},Y_k|U_{A_j})-4(1+\nu)\delta(\epsilon_k'))}.
\end{align*}
Therefore, $P(|\mathcal{L}_{k}(l_{\lin{D}_k},u_{ \uin{D}_k}^n,y_k^n)|=0|U_{ \uin{D}_k}^n(l_{\lin{D}_k})=u_{ \uin{D}_k}^n)$ and thus  $P(\mathcal{E}_{k,3}\cap \mathcal{E}_{k,1}^c \cap \mathcal{E}_{k,2}^c )$  tend to zero as $n$ tends to infinity if 
\begin{align}
\sum_{j\in \lset{T}_k}r_j &>\sum_{j\in   \uset{T}_k}I(U_j;U_{  \uset{T}_k[j]\cup  \uin{D}_k},Y_k|U_{A_j})+4(1+\nu)\delta(\epsilon_k') \label{eqn:covering_pf_eq}
\end{align}
for all  $\lset{T}_k\subseteq \lin{W}_k $ such that $\lset{T}_k\neq \emptyset$.

\subsection{Bounding the fourth term in the summation in (\ref{eqn:error_analysis_ub})} \label{appendix:fourth_error}
For given $k\in [1:N]$, we have 
\begin{align*}
&P(\mathcal{E}_{k,4}\cap \mathcal{E}_{k,1}^c \cap \mathcal{E}_{k,2}^c )\cr
&\leq P((U^n_{ \uin{W}^{k-1}}(L_{\lin{W}^{k-1}}), Y_{[1:k]}^n)\in \mathcal{T}_{\epsilon_k}^{(n)},(U^n_{ \uin{W}^k}(L_{\lin{W}^k}), Y_{[1:k]}^n)\notin \mathcal{T}_{\epsilon_k''}^{(n)}, \hat{L}_{\lin{D}_k,k}=L_{\lin{D}_k})\cr
&\leq P((U^n_{ \uin{W}^{k-1}}(\hat{L}_{\lin{D}_k,k}, L_{\lin{W}^{k-1}\setminus \lin{D}_k}), Y_{[1:k]}^n)\in \mathcal{T}_{\epsilon_k}^{(n)},\cr 
&~~~~~~~(U^n_{ \uin{W}^{k-1}}(\hat{L}_{\lin{D}_k,k},L_{\lin{W}^{k-1}\setminus \lin{D}_k}), Y_{[1:k]}^n, U^n_{ \uin{W}_k}(\hat{L}_{\lin{D}_k,k}, L_{\lin{W}_k}))\notin \mathcal{T}_{\epsilon_k''}^{(n)})\cr
&=\sum_{l_{\lin{D}_k},(u^n_{ \uin{W}^{k-1}}, y_{[1:k]}^n)\in \mathcal{T}_{\epsilon_k}^{(n)}}P(\hat{L}_{\lin{D}_k,k}=l_{\lin{D}_k}, U^n_{ \uin{W}^{k-1}}(l_{\lin{D}_k}, L_{\lin{W}^{k-1}\setminus \lin{D}_k})=u^n_{ \uin{W}^{k-1}},Y_{[1:k]}^n=y_{[1:k]}^n )\cr
&~~\cdot P((u^n_{ \uin{W}^{k-1}}, y_{[1:k]}^n, U^n_{ \uin{W}_k}(l_{\lin{D}_k}, L_{\lin{W}_k}))\notin \mathcal{T}_{\epsilon_k''}^{(n)}|\hat{L}_{\lin{D}_k,k}=l_{\lin{D}_k}, U^n_{ \uin{W}^{k-1}}(l_{\lin{D}_k}, L_{\lin{W}^{k-1}\setminus \lin{D}_k})=u^n_{ \uin{W}^{k-1}},Y_{[1:k]}^n=y_{[1:k]}^n) 
\end{align*}

We use the following modified Markov lemma to bound the above, which can be proved from the proof of the Markov lemma in \cite{tung_thesis}, \cite{ElGamalKim:11} with some minor modification.

\begin{lemma} \label{lemma:markov}
Consider random variables $X, Y, Z, A$ such that $X\rightarrow Y\rightarrow Z$ form a Markov chain. 
Let $(x^n, y^n)\in \mathcal{T}_{\epsilon}^{(n)}$ and $a\in \mathcal{A}$. Suppose that $P(Z^n=z^n|X^n=x^n, Y^n=y^n, A=a)=P(Z^n=z^n|Y^n=y^n, A=a)$, where $P(Z^n=z^n|Y^n=y^n, A=a)$ satisfies the following conditions for $\epsilon'>\epsilon$:
\begin{enumerate}
\item $\lim_{n\rightarrow \infty} P((y^n, Z^n)\in \mathcal{T}_{\epsilon'}^{(n)}|Y^n=y^n, A=a)=1$.
\item For every $z^n\in \mathcal{T}_{\epsilon'}^{(n)}(Z|y^n)$ and $n$ sufficiently large, \[P(Z^n=z^n|Y^n=y^n, A=a)\leq 2^{-n(H(Z|Y)-\delta(\epsilon'))}.\] 
\end{enumerate}
Then, for sufficiently small $\epsilon$ and $\epsilon'$ such that $\epsilon<\epsilon'<\epsilon''$, 
\begin{align*}
\lim_{n\rightarrow \infty} P((x^n,y^n,Z^n)\in \mathcal{T}_{\epsilon''}^{(n)}|X^n=x^n, Y^n=y^n, A=a)=1.
\end{align*}
\end{lemma}

Fix $l_{\lin{D}_k}$ and $(u^n_{ \uin{W}^{k-1}}, y_{[1:k]}^n)\in \mathcal{T}_{\epsilon_k}^{(n)}$.  We use Lemma \ref{lemma:markov} to show 
\begin{align}
&\lim_{n\rightarrow \infty} P((u^n_{ \uin{W}^{k-1}}, y_{[1:k]}^n, U^n_{ \uin{W}_k}(l_{\lin{D}_k}, L_{\lin{W}_k}))\in \mathcal{T}_{\epsilon_k''}^{(n)}|\hat{L}_{\lin{D}_k,k}=l_{\lin{D}_k}, \cr 
&~~~~~~~~~~~~~~~ U^n_{ \uin{W}^{k-1}}(l_{\lin{D}_k}, L_{\lin{W}^{k-1}\setminus \lin{D}_k})=u^n_{ \uin{W}^{k-1}},Y_{[1:k]}^n=y_{[1:k]}^n)=1. \label{eqn:markov_toshow}
\end{align}

Note that $(U_{ \uin{W}^{k-1}\setminus  \uin{D}_k}, Y_{[1:k-1]})-(U_{ \uin{D}_k},Y_k)-U_{ \uin{W}_k}$ form a Markov chain and 
\begin{align*}
&P(U^n_{ \uin{W}_k}(l_{\lin{D}_k}, L_{\lin{W}_k})=u^n_{ \uin{W}_k}|\hat{L}_{\lin{D}_k,k}=l_{\lin{D}_k}, U^n_{ \uin{W}^{k-1}}(l_{\lin{D}_k}, L_{\lin{W}^{k-1}\setminus \lin{D}_k})=u^n_{ \uin{W}^{k-1}},Y_{[1:k]}^n=y_{[1:k]}^n)\cr
&=P(U^n_{ \uin{W}_k}(l_{\lin{D}_k}, L_{\lin{W}_k})=u^n_{ \uin{W}_k}|\hat{L}_{\lin{D}_k,k}=l_{\lin{D}_k}, U^n_{ \uin{D}_k}(l_{\lin{D}_k})=u^n_{ \uin{D}_k},Y_{k}^n=y_{k}^n).
\end{align*}

The first condition in Lemma \ref{lemma:markov} is safisfied if (\ref{eqn:covering_pf_eq}) is satisfied for all  $\lset{T}_k\subseteq \lin{W}_k $ such that $\lset{T}_k\neq \emptyset$ since 
\begin{align*}
&P((u^n_{ \uin{D}_k}, y_{k}^n, U^n_{ \uin{W}_k}(l_{\lin{D}_k}, L_{\lin{W}_k}))\in \mathcal{T}_{\epsilon_k'}^{(n)}|\hat{L}_{\lin{D}_k,k}=l_{\lin{D}_k}, U^n_{ \uin{D}_k}(l_{\lin{D}_k})=u^n_{ \uin{D}_k},Y_{k}^n=y_{k}^n)\cr 
&=1-P((u^n_{ \uin{D}_k}, y_{k}^n, U^n_{ \uin{W}_k}(l_{\lin{D}_k}, l_{\lin{W}_k}))\notin \mathcal{T}_{\epsilon_k'}^{(n)} \mbox{ for all } l_{\lin{W}_k}|\hat{L}_{\lin{D}_k,k}=l_{\lin{D}_k}, U^n_{ \uin{D}_k}(l_{\lin{D}_k})=u^n_{ \uin{D}_k},Y_{k}^n=y_{k}^n) \cr
&=1-P((u^n_{ \uin{D}_k}, y_{k}^n, U^n_{ \uin{W}_k}(l_{\lin{D}_k}, l_{\lin{W}_k}))\notin \mathcal{T}_{\epsilon_k'}^{(n)} \mbox{ for all } l_{\lin{W}_k}|U^n_{ \uin{D}_k}(l_{\lin{D}_k})=u^n_{ \uin{D}_k}) 
\end{align*}
and we have showed in the analysis of the third error event
\begin{align*}
\lim_{n\rightarrow \infty}P((u^n_{ \uin{D}_k}, y_{k}^n, U^n_{ \uin{W}_k}(l_{\lin{D}_k}, l_{\lin{W}_k}))\notin \mathcal{T}_{\epsilon_k'}^{(n)} \mbox{ for all } l_{\lin{W}_k}|U^n_{ \uin{D}_k}(l_{\lin{D}_k})=u^n_{ \uin{D}_k}) =0
\end{align*}
under the aforementioned condition. 

Now, let us show that the second condition in Lemma \ref{lemma:markov} is satisfied. For every $u^n_{ \uin{W}_k}\in \mathcal{T}_{\epsilon_k'}^{(n)}(U_{ \uin{W}_k}|u^n_{ \uin{D}_k}, y_{k}^n)$,
\begin{align*}
&P(U^n_{ \uin{W}_k}(l_{\lin{D}_k}, L_{\lin{W}_k})=u^n_{ \uin{W}_k}|\hat{L}_{\lin{D}_k,k}=l_{\lin{D}_k}, U^n_{ \uin{D}_k}(l_{\lin{D}_k})=u^n_{ \uin{D}_k}, Y_{k}^n=y_{k}^n)\cr
&=P(U^n_{ \uin{W}_k}(l_{\lin{D}_k}, L_{\lin{W}_k})=u^n_{ \uin{W}_k},U^n_{ \uin{W}_k}(l_{\lin{D}_k}, L_{\lin{W}_k})\in \mathcal{T}_{\epsilon_k'}^{(n)}(U_{ \uin{W}_k}| u^n_{ \uin{D}_k}, y_{k}^n)|\hat{L}_{\lin{D}_k,k}=l_{\lin{D}_k}, U^n_{ \uin{D}_k}(l_{\lin{D}_k})=u^n_{ \uin{D}_k}, Y_{k}^n=y_{k}^n)\cr 
&\leq P(U^n_{ \uin{W}_k}(l_{\lin{D}_k}, L_{\lin{W}_k})=u^n_{ \uin{W}_k}|\hat{L}_{\lin{D}_k,k}=l_{\lin{D}_k}, U^n_{ \uin{D}_k}(l_{\lin{D}_k})=u^n_{ \uin{D}_k}, Y_{k}^n=y_{k}^n,U^n_{ \uin{W}_k}(l_{\lin{D}_k}, L_{\lin{W}_k})\in \mathcal{T}_{\epsilon_k'}^{(n)}(U_{ \uin{W}_k}|u^n_{ \uin{D}_k}, y_{k}^n))\cr 
&=\sum_{l_{\lin{W}_k}} P(L_{\lin{W}_k}=l_{\lin{W}_k}|\hat{L}_{\lin{D}_k,k}=l_{\lin{D}_k}, U^n_{ \uin{D}_k}(l_{\lin{D}_k})=u^n_{ \uin{D}_k}, Y_{k}^n=y_{k}^n,U^n_{ \uin{W}_k}(l_{\lin{D}_k}, L_{\lin{W}_k})\in \mathcal{T}_{\epsilon_k'}^{(n)}(U_{ \uin{W}_k}|u^n_{ \uin{D}_k}, y_{k}^n))\cr 
&~~~~~~~~~~~~~~~~~~~~\cdot P(U^n_{ \uin{W}_k}(l_{\lin{D}_k\cup \lin{W}_k})=u^n_{ \uin{W}_k}|\hat{L}_{\lin{D}_k,k}=l_{\lin{D}_k}, L_{\lin{W}_k}=l_{\lin{W}_k},U^n_{ \uin{D}_k}(l_{\lin{D}_k})=u^n_{ \uin{D}_k}, Y_{k}^n=y_{k}^n,\cr 
&~~~~~~~~~~~~~~~~~~~~~~~~~~~~~~~~~~~~~~~~U^n_{ \uin{W}_k}(l_{\lin{D}_k\cup \lin{W}_k})\in \mathcal{T}_{\epsilon_k'}^{(n)}(U_{ \uin{W}_k}|u^n_{ \uin{D}_k}, y_{k}^n)). 
\end{align*}
For given  $l_{\lin{W}_k}$, we have 
\begin{align*}
&P(U^n_{ \uin{W}_k}(l_{\lin{D}_k\cup \lin{W}_k})=u^n_{ \uin{W}_k}|\hat{L}_{\lin{D}_k,k}=l_{\lin{D}_k}, L_{\lin{W}_k}=l_{\lin{W}_k},U^n_{ \uin{D}_k}(l_{\lin{D}_k})=u^n_{ \uin{D}_k}, Y_{k}^n=y_{k}^n, \cr
&~~~~~~~~~~~~~~~~~~~~~~~~~~~~~~~~~~~~~~~~~~~~~~~~~~~~~~~~~U^n_{ \uin{W}_k}(l_{\lin{D}_k\cup \lin{W}_k})\in \mathcal{T}_{\epsilon_k'}^{(n)}(U_{ \uin{W}_k}|u^n_{ \uin{D}_k}, y_{k}^n))\cr 
&\overset{(a)}{=}P(U^n_{ \uin{W}_k}(l_{\lin{D}_k\cup\lin{W}_k})=u^n_{ \uin{W}_k}|U^n_{ \uin{D}_k}(l_{\lin{D}_k})=u^n_{ \uin{D}_k}, U^n_{ \uin{W}_k}(l_{\lin{D}_k\cup\lin{W}_k})\in \mathcal{T}_{\epsilon_k'}^{(n)}(U_{ \uin{W}_k}|u^n_{ \uin{D}_k}, y_{k}^n))\cr 
&=\frac{P(U^n_{ \uin{W}_k}(l_{\lin{D}_k\cup\lin{W}_k})=u^n_{ \uin{W}_k}|U^n_{ \uin{D}_k}(l_{\lin{D}_k})=u^n_{ \uin{D}_k})}{P(U^n_{ \uin{W}_k}(l_{\lin{D}_k\cup \lin{W}_k})\in \mathcal{T}_{\epsilon_k'}^{(n)}(U_{ \uin{W}_k}|u^n_{ \uin{D}_k}, y_{k}^n)|U^n_{ \uin{D}_k}(l_{\lin{D}_k})=u^n_{ \uin{D}_k})}\cr
&\leq 2^{-n(H(U_{ \uin{W}_k}|U_{ \uin{D}_k},Y_k)-\delta(\epsilon_k'))},
\end{align*}
where $(a)$ is because 
\begin{align*}
&P(U^n_{ \uin{W}_k}(l_{\lin{D}_k\cup\lin{W}_k})=u^n_{ \uin{W}_k},\hat{L}_{\lin{D}_k,k}=l_{\lin{D}_k}, L_{\lin{W}_k}=l_{\lin{W}_k},U^n_{ \uin{D}_k}(l_{\lin{D}_k})=u^n_{ \uin{D}_k}, Y_{k}^n=y_{k}^n, \cr
&~~~~~~~~~~~~~~~~~~~~~~~~~~~~~~~~~~~~~~~~~~~~~~~~~~~~~~~~~U^n_{ \uin{W}_k}(l_{\lin{D}_k\cup \lin{W}_k})\in \mathcal{T}_{\epsilon_k'}^{(n)}(U_{ \uin{W}_k}|u^n_{ \uin{D}_k}, y_{k}^n))\cr 
&=P(\hat{L}_{\lin{D}_k,k}=l_{\lin{D}_k},U^n_{ \uin{D}_k}(l_{\lin{D}_k})=u^n_{ \uin{D}_k}, Y_{k}^n=y_{k}^n)\cr
&~~~\times P(U^n_{ \uin{W}_k}(l_{\lin{D}_k\cup \lin{W}_k})=u^n_{ \uin{W}_k},U^n_{ \uin{W}_k}(l_{\lin{D}_k\cup \lin{W}_k})\in \mathcal{T}_{\epsilon_k'}^{(n)}(U_{ \uin{W}_k}|u^n_{ \uin{D}_k}, y_{k}^n)| U^n_{ \uin{D}_k}(l_{\lin{D}_k})=u^n_{ \uin{D}_k})\cr
&~~~~~\times  	P( L_{\lin{W}_k}=l_{\lin{W}_k}|\hat{L}_{\lin{D}_k,k}=l_{\lin{D}_k},U^n_{ \uin{D}_k}(l_{\lin{D}_k})=u^n_{ \uin{D}_k}, Y_{k}^n=y_{k}^n,U^n_{ \uin{W}_k}(l_{\lin{D}_k\cup \lin{W}_k})\in \mathcal{T}_{\epsilon_k'}^{(n)}(U_{ \uin{W}_k}|u^n_{ \uin{D}_k}, y_{k}^n))\cr
&=P(\hat{L}_{\lin{D}_k,k}=l_{\lin{D}_k},U^n_{ \uin{D}_k}(l_{\lin{D}_k})=u^n_{ \uin{D}_k}, Y_{k}^n=y_{k}^n,U^n_{ \uin{W}_k}(l_{\lin{D}_k\cup \lin{W}_k})\in \mathcal{T}_{\epsilon_k'}^{(n)}(U_{ \uin{W}_k}|u^n_{ \uin{D}_k}, y_{k}^n),L_{\lin{W}_k}=l_{\lin{W}_k})\cr
&~~~\times P(U^n_{ \uin{W}_k}(l_{\lin{D}_k\cup \lin{W}_k})=u^n_{ \uin{W}_k}| U^n_{ \uin{D}_k}(l_{\lin{D}_k})=u^n_{ \uin{D}_k},U^n_{ \uin{W}_k}(l_{\lin{D}_k\cup \lin{W}_k})\in \mathcal{T}_{\epsilon_k'}^{(n)}(U_{ \uin{W}_k}|u^n_{ \uin{D}_k}, y_{k}^n)).
\end{align*}
Hence, the second condition in Lemma \ref{lemma:markov} is satisfied. 

Now, from  Lemma \ref{lemma:markov}, (\ref{eqn:markov_toshow}) holds and hence $P(\mathcal{E}_{k,4}\cap \mathcal{E}_{k,1}^c \cap \mathcal{E}_{k,2}^c )$ tends to zero as $n$ tends to infinity for sufficiently small $\epsilon_k$ and $\epsilon_k'$ if (\ref{eqn:covering_pf_eq}) is satisfied for all  $\lset{T}_k\subseteq \lin{W}_k $ such that $\lset{T}_k\neq \emptyset$.

\subsection{Special cases} \label{appendix:special}
In this appendix, we present more examples of previous results that can be obtained as simple corollaries of our theorem.

\subsubsection{Channels with action-dependent states \cite{Weissman:10}}
\begin{itemize}
\item ADMN: $N=3, H(Y_1)=R, Y_2=(X_1,Y_1,S)$, $p(s|x_1,y_1)=p(s|x_1)$, $p(y_3|y_{[1:2]},x_{[1:2]})=p(y_3|x_2,s)$.
\item Target distribution: $p^*$ such that $X_3=Y_1$.
\item Choice of $\omega'\in \Omega'$ in Corollary \ref{corollary:main}: $\mu=2, U_1=(Y_1,X_1), W_1=\{1\}, W_2=\{2\}, D_2=\{1\}, D_3=\{1\}, B_3=\{2\}, A_2=\{1\}$, $p(x_1|y_1)=p(x_1)$, $p(u_2|y_2,u_1)=p(u_2|s,x_1)$, $x_2(y_2,u_1,u_2)=x_2(s,u_2)$, $x_3(y_3,u_1)=y_1$.
\end{itemize}

\subsubsection{Marton's inner bound with common message \cite{liang_thesis}}
\begin{itemize}
\item ADMN: $N=3$, $Y_1=(M_1,M_2,M_3)$, $H(M_k)=R_k$ for $k\in [1:3]$, $p(y_1)=p(m_1)p(m_2)p(m_3)$,  $p(y_2|y_1,x_1)=p(y_2|x_1)$, $p(y_3|y_{[1:2]},x_{[1:2]})=p(y_3|x_1,y_2)$.
\item Target distribution: $p^*$ such that $X_2=(M_1,M_2)$, $X_3=(M_1,M_3)$.
\item Choice of $\omega'\in \Omega'$ in Corollary \ref{corollary:main}: $\mu=3$, $U_j=(M_j,V_j)$ for $j\in [1:3]$, $W_1=\{1,2,3\}, D_2=\{1,2\}, D_3=\{1,3\}, A_2=\{1\}, A_3=\{1\}$, $p(v_{[1:3]}|y_1)=p(v_{[1:3]})$, $x_1(u_{[1:3]},y_1)=x_1(v_{[1:3]})$, $x_2(u_{[1:2]}, y_2)= (m_1, m_2)$, $x_3(u_{\{1,3\}}, y_3)= (m_1, m_3)$.
\end{itemize}

\subsubsection{Three-receiver multilevel broadcast channel \cite{NairElGamal:09}}
\begin{itemize}
\item ADMN: $N=4$, $Y_1=(M_0, M_{10}, M_{11})$, $H(M_0)=R_0, H(M_{10})=R_{10}, H(M_{11})=R_{11}$, $R_1=R_{10}+R_{11}$, $p(y_1)=p(m_0)p(m_{10})p(m_{11})$, $p(y_2|y_1,x_1)=p(y_2|x_1)$, $p(y_3|y_{[1:2]}, x_{[1:2]})=p(y_3|y_2)$, $p(y_4|y_{[1:3]}, x_{[1:3]})=p(y_4|y_2,x_1)$.
\item Target distribution: $p^*$ such that $X_2=(M_0, M_{10}, M_{11})$, $X_3=M_0$, $X_4=M_0$.
\item Choice of $\omega'\in \Omega'$ in Corollary \ref{corollary:main}:  $\mu=3$, $U_1=(M_0, V_1)$, $U_2=(M_{10}, V_2)$, $U_3=(M_{11}, X_1)$, $W_1=\{1,2,3\}$, $D_2=\{1,2,3\}$, $D_3=\{1\}$, $D_4=\{1\}$, $B_4=\{2\}$, $A_2=\{1\}$, $A_3=\{1,2\}$, $p(v_{[1:2]},x_1|y_1)=p(v_{[1:2]})p(x_1|v_2)$.
\end{itemize}


\end{document}